\documentclass[twocolumn,pdftex]{svjour3}

\usepackage{cancel}
\usepackage[normalem]{ulem}
\usepackage{microtype}
\usepackage{url}
\usepackage[utf8]{inputenc}
\usepackage{amsmath}
\usepackage{balance}
\usepackage{relsize} 
\usepackage{wrapfig}
\usepackage{cite} 
\usepackage{amssymb}
\usepackage{stmaryrd}
\usepackage{extarrows}
\usepackage{mathtools}
\usepackage[mathscr]{euscript}
\usepackage{rotating}

\usepackage{amsthm}
\usepackage{thmtools, thm-restate} 
\usepackage[table]{xcolor}
\usepackage{epsfig}
\usepackage{bbold}
\usepackage{multicol}
\usepackage{tikz}
\usepackage{calc}
\usepackage{xspace} 
\usepackage[caption=false]{subfig}
\usepackage{graphicx,tikz}
\graphicspath{ {images/} }
\usetikzlibrary{arrows,automata,calc,shapes}
\usetikzlibrary{positioning}
\usetikzlibrary{decorations.pathmorphing}
\pgfdeclarelayer{background}
\pgfsetlayers{background,main}
\usepackage{wrapfig} 
\usepackage{subfig}
\usepackage{multirow}
\usepackage{array}
\usepackage[outdir=./]{epstopdf}
\usepackage{bm} 
\usepackage{paralist} 
\usepackage{algorithm} 
\usepackage[noend]{algpseudocode}

\allowdisplaybreaks

\algrenewcommand\algorithmicindent{2ex}%
\newcommand{\LineFor}[2]{%
    \State \algorithmicfor\ {#1}\ \algorithmicdo\ {#2} %
}

\usepackage[utf8]{inputenc}
\usepackage{tikz}
\usetikzlibrary{topaths,calc}
\usepackage{lipsum,adjustbox}
\usepackage{silence}
 \newcommand{\subparagraph}{} 
 \usepackage[explicit]{titlesec}

 \titleformat{\paragraph}[runin]
     {\normalfont\bfseries}
     {}
     {0pt}
     {#1. }

\titlespacing{\paragraph}
    {0pt} 
    {1ex plus .3ex minus .1ex}
    {1ex} 

\makeatletter\@ifclassloaded{article}{\theoremstyle{plain}}{}\makeatother

\theoremstyle{definition}

\newcommand{\oomtime}{two\xspace}

\newcommand{\bi}{BI\xspace}
\newcommand{\cer}{CER\xspace}
\newcommand{\esper}{$E$\xspace}
\newcommand{\dbt}{$DBT$\xspace}
\newcommand{\zstream}{$Z$\xspace}
\newcommand{\sase}{$SE$\xspace}
\newcommand{\tesla}{$T$\xspace}

\newcommand{\reduc}{\ensuremath{\rho}}

\newcommand{\gjt}{GJT\xspace}
\newcommand{\gjf}{GJF\xspace}

\newcommand{\gjts}{GJTs\xspace}
\newcommand{\gjfs}{GJFs\xspace}
\newcommand{\cq}{CQ\xspace}
\newcommand{\cqs}{CQs\xspace}
\newcommand{\gcq}{GCQ\xspace}
\newcommand{\gcqs}{GCQs\xspace}

\newcommand{\dtree}[1]{{#1}-rep\xspace}

\newcommand{\ivm}{IVM\xspace}
\newcommand{\hivm}{HIVM\xspace}

\newcommand{\gmr}{GMR\xspace}

\newcommand{\cde}{CDE\xspace}

\DeclareMathOperator{\drep}{\mathcal{D}}

\DeclareMathOperator{\var}{\textit{var}}
\DeclareMathOperator{\pred}{\textit{pred}}

\DeclareMathOperator{\free}{\textit{out}}

\DeclareMathOperator{\supp}{supp}

\DeclareMathOperator{\proj}{\pi}

\DeclareMathOperator{\child}{ch}

\DeclareMathOperator{\nodes}{nodes}
\DeclareMathOperator{\leafs}{leafs}
\DeclareMathOperator{\treeroot}{root}

\DeclareMathAlphabet{\mathcal}{OMS}{cmsy}{m}{n}

\newcommand{\seq}[1]{\overline{#1}}
\newcommand{\tup}[1]{\vec{#1}}

\newcommand{\sem}[1]{#1}

\DeclareMathOperator{\routenum}{\algenum}

\newcommand{\aset}{\ats}
\newcommand{\ats}{\ensuremath{\mathcal{A}}}

\DeclareMathOperator{\preds}{\textit{pred}}

\DeclareMathOperator{\at}{\textit{at}}
\DeclareMathOperator{\atoms}{\textit{at}}
\DeclareMathOperator{\predicates}{\textit{pred}}
\DeclareMathOperator{\hypergraph}{\textit{hyp}}

\DeclareMathOperator{\db}{\textit{db}}

\DeclareMathOperator{\upd}{\mathit{u}}

\newcommand{\delt}[1]{\ensuremath{\Delta{#1}}}

\DeclareMathOperator{\semijoin}{\ltimes}

\newcommand{\card}[1]{|{#1}|}
\newcommand{\size}[1]{\parallel\hspace{-0.75ex}{#1}\hspace{-0.75ex}\parallel}

\newcommand{\dyn}{\text{\sc Dyn}\xspace}
\newcommand{\edyn}{\text{\sc GDyn}\xspace}
\newcommand{\iedyn}{\text{\sc IEDyn}\xspace}

\newcommand{\algenum}{\textsc{enum}}

\newcommand{\trip}[1]{\mathcal{#1}}
\newcommand{\ftrip}[1]{\mathbb{#1}}
\DeclareMathOperator{\ext}{\textit{ext}}

\DeclareMathOperator{\isolated}{\textit{isol}}

\DeclareMathOperator{\equijoinvars}{\textit{jv}}
\DeclareMathOperator{\forest}{\textit{forest}}

\newcommand{\nf}[1]{\ensuremath{{#1}\!\!\downarrow}}
\DeclareMathOperator*{\rewr}{\rightsquigarrow}
\DeclareMathOperator{\hypertrip}{\mathcal{H}}
\DeclareMathOperator*{\cse}{\sqsubseteq}
\newcommand{\fpreds}[1]{\ensuremath{\Theta_{\ftrip #1}}}

\begin{document}

\title{
       Conjunctive Queries with Theta Joins Under Updates
}


\author{
       Muhammad Idris \and
       Martín Ugarte \and
       Stijn Vansummeren \and
       Hannes Voigt \and
       Wolfgang Lehner 
}

\institute{M. Idris \at
  Universit\'e Libre de Bruxelles, Belgium and TU Dresden, Germany \\
              \email{midris@ulb.ac.be}           
           \and
           M. Ugarte \at
  Universit\'e Libre de Bruxelles, Belgium\\
  \email{mugartec@ulb.ac.be}
  \and 
  S. Vansummeren \at
  Universit\'e Libre de Bruxelles, Belgium\\
  \email{svsummer@ulb.ac.be}
  \and
  Hannes Voigt \at
  neo4j, Germany \\
  This work was done while the author was
    affiliated to TU Dresden, Germany \\
  \email{hannes.voigt@neo4j.com}
  \and
  Wolfgang Lehner \at 
  TU Dresden, Germany\\
  \email{wolfgang.lehner@tu-dresden.de}
}

\date{Received: date / Accepted: date}


\maketitle


\begin{abstract}
  Modern application domains such as Composite Event Recognition
  (\cer) and real-time Analytics require the ability to dynamically
  refresh query results under high update rates. Traditional
  approaches to this problem are based either on the materialization
  of subresults (to avoid their recomputation) or on the recomputation
  of subresults (to avoid the space overhead of materialization). Both
  techniques have recently been shown suboptimal: instead of
  materializing results and subresults, one can maintain a data
  structure that supports efficient maintenance under updates and can
  quickly enumerate the full query output, as well as the changes
  produced under single updates.  Unfortunately, these data structures
  have been developed only for aggregate-join queries composed of
  equi-joins, limiting their applicability in domains such as \cer
  where temporal joins are commonplace. In this paper, we present a
  new approach for dynamically evaluating queries with multi-way
  $\theta$-joins under updates that is effective in avoiding both
  materialization and recomputation of results, while supporting a
  wide range of applications. To do this we generalize Dynamic
  Yannakakis, an algorithm for dynamically processing acyclic
  equi-join queries.  In tandem, and of independent interest, we
  generalize the notions of acyclicity and free-connexity to arbitrary
  $\theta$-joins and show how to compute corresponding join trees. We
  instantiate our framework to the case where $\theta$-joins are only
  composed of equalities and inequalities ($<, \leq, >, \geq$) and
  experimentally compare our algorithm to state of the art \cer
  systems as well as incremental view maintenance engines. Our
  approach performs consistently better than the competitor systems
  with up to \oomtime orders of magnitude improvements in both time
  and memory consumption.
\end{abstract}


\section{Introduction}

The ability to analyze dynamically changing data is a key requirement
of many contemporary applications, usually associated with Big Data,
that require such analysis in order to obtain timely insights and
implement reactive and proactive measures. Example
applications include Financial Systems \cite{cormode2007}, Industrial
Control Systems \cite{groover2007}, Stream Processing
\cite{Stonebraker:2005}, Composite Event Recognition (\cer, also known as
Complex Event Processing) \cite{buchmann:2009,cugola:2012}, and
Business Intelligence (\bi) \cite{sahay2008}. Generally, the analysis that
needs to be kept up-to-date, or at least their basic elements, are
specified in a query language. The main task is then to efficiently
update the query results under frequent data updates.

In this paper, we focus on the problem of dynamic evaluation for
queries that feature multi-way $\theta$-joins in addition to standard
equi-joins. To illustrate our setting, consider that we wish to detect
potential credit card frauds. Credit card transactions have a timestamp (\textit{ts}), account number
(\textit{acc}), and amount (\textit{amnt}), among other attributes. A typical
fraud pattern is that the criminal tests the credit card
with a few small purchases to then make larger purchases
(cf. \cite{DBLP:conf/debs/Schultz-MollerMP09}). In this respect, we
would like to dynamically evaluate the following query, assuming new transactions arrive
in a streaming fashion and the pattern must be detected in less than 1 hour.

\begin{verbatim}
SELECT * FROM Trans S1, Trans S2, Trans L
WHERE S1.ts < S2.ts AND S2.ts < L.ts 
AND L.ts < S1.ts + 1h
AND S1.acc = S2.acc AND S2.acc = L.acc
AND S1.amnt < 100 AND S2.amnt < 100 
AND L.amnt > 400
\end{verbatim}

Queries like this with inequality joins appear in both \cer and \bi
scenarios. Traditional techniques to process these queries dynamically
can be categorized in two approaches: relational and
automaton-based. We next discuss both approaches, their strengths and drawbacks.

\paragraph*{Relational} Relational approaches, such as
\cite{DBLP:journals/vldb/KochAKNNLS14, DBLP:conf/sigmod/MeiM09,
  DBLP:books/sp/16/ArasuBBCDIMSW16} are based on a form of
\emph{Incremental View Maintenance} (\ivm). To process a query $Q$
over a database $\db$, \ivm techniques materialize the output $Q(\db)$
and evaluate \emph{delta queries}. Upon update $\upd$, delta queries
use $\db$, $\upd$ and the materialized $Q(\db)$ to compute the set of
tuples to add/delete from $Q(\db)$ in order to obtain
$Q(\db+\upd)$. If $u$ is small w.r.t.\@ $\db$, this is expected to be
faster than recomputing $Q(\db+\upd)$ from scratch. To further speed
up dynamic query processing, we may materialize not only $Q(\db)$ but
also the result of some subqueries. This is known as Higher-Order \ivm
(\hivm for short) \cite{DBToaster2016,
  DBLP:journals/vldb/KochAKNNLS14,DBLP:conf/pods/Koch10}. Both \ivm
and \hivm have drawbacks, however. First, materialization of $Q(\db)$
requires $\Omega(\size{Q(\db)})$ space, where $\size{\db}$ denotes the
size of $\db$. Therefore, when $Q(\db)$ is large compared to $\db$,
materializing $Q(\db)$ quickly becomes impractical, especially for
main-memory based systems. \hivm is even more affected by this problem
because it not only materializes the result of $Q$ but also the
results to some subqueries. For example, in our fraud query \hivm
would materialize the results of the following join in order to
respond quickly to the arrival of a potential transaction $L$:
\begin{equation}
  \label{eq:1}
  \sigma_{\text{amnt} < 100}(S_1) \Join_{S1.ts < S2.ts \wedge S1.acc =
    S2.acc} \sigma_{\text{amnt} < 100}(S_2)   \tag{$\star$}
\end{equation}
If we assume that there are $N$ small transactions
in the time window, all of the same account, this materialization will take $\Theta(N^2)$ space.
This becomes rapidly impractical when $N$ becomes large.

\paragraph*{\bf Automata} Automaton-based approaches (e.g.,
\cite{brenna07,DBLP:conf/sigmod/WuDR06,DBLP:journals/jss/CugolaM12,DBLP:conf/debs/CugolaM10,Sase2014,DBLP:conf/sigmod/AgrawalDGI08})
are primarily employed in \cer systems. In contrast to the relational
approaches, they assume that the arrival order of event tuples
corresponds to the timestamp order (i.e., there are no out-of-order
events) and build an automaton to recognize the desired temporal
patterns in the input stream. Broadly speaking, there are two
automata-based recognition approaches. In the first approach, followed
by \cite{DBLP:conf/sigmod/WuDR06,DBLP:conf/sigmod/AgrawalDGI08},
events are cached per state and once a final state is reached a search
through the cached events is done to recognize the complex
events. While it is no longer necessary to check the temporal constraints 
during the search, the additional constraints (in our example, $L.ts <
S1.ts + 1h$ and $S_1.\text{acc} = S_2.\text{acc} = L.\text{acc}$) must
still be verified. If the additional constraints are highly selective
this approach creates an unnecessarily large update latency, given
that each event triggering a transition to a final state may cause
re-evaluation of a sub-join on the cached data, only to find few new output
tuples. 

In the second approach, followed by
\cite{brenna07,Sase2014,DBLP:journals/jss/CugolaM12,DBLP:conf/debs/CugolaM10},
partial runs are materialized according to the automaton's
topology. For our example query, this means that, just like \hivm, the
join \eqref{eq:1} is materialized and maintained so it is available
when a large amount transaction $L$ arrives. This approach hence
shares with \hivm its high memory overhead and maintenance cost.

\medskip
It has been recently shown that the drawbacks of these two approaches
can be overcome by a rather simple
idea~\cite{dyn:2017,olteanu:fivm}. 
Instead of fully materializing (potentially large) results and
subresults, we can build a compact representation of the query result
that supports efficient maintenance under updates. The representation
is equipped with index structures so that, whenever necessary, we can
generate the actual query result one tuple at a time, spending a limited
amount of work to produce each new result tuple. This makes
the generation performance-wise competitive with enumeration from a
fully materialized (non-compact) output. In essence, we are hence
separating dynamic query processing into two stages: (1) an update
stage where we only maintain under updates the (small) information that is necessary
to be able to efficiently generate the query result and (2)
an enumeration stage where the query result is efficiently
enumerated. Moreover, for single-tuple updates
the representation also supports efficient enumeration of the changes
to the query result. This is relevant for push-based query processing
systems, where users do not ping the system for the complete current
query answer, but instead ask to be notified of the changes to the
query results when the database changes.

This idea was first presented by a subset of the authors in the
Dynamic Yannakakis Algorithm (\dyn for short)~\cite{dyn:2017}, an
algorithm for efficiently processing acyclic aggregate-join queries. 
\dyn is worst-case optimal for two classes of queries, namely 
the q-hierarchical and free-connex acyclic conjunctive queries. A
different approach named F-IVM, based on so-called \emph{factorized
  databases}, was later developed to dynamically process
aggregate-join queries that are not necessarily acyclic or need to
support complex aggregates~\cite{olteanu:fivm}.

Unfortunately, both \dyn and F-IVM are only applicable to queries with
equality joins, and as such they do not support analytical queries
with other types of joins like the ones with inequalities $(\leq, <,
\geq, >)$ or
disequalities $(\not =)$. 
Therefore, the current state of the art techniques for dynamically
processing queries with joins beyond equality suffer either from a
high update latency (if subresults are not materialized) or a high
memory footprint (if subresults are materialized).
In this paper, we overcome these problems by generalizing the Dynamic
Yannakakis algorithm to conjunctive queries with arbitrary
$\theta$-joins. We show that, in the specific case of inequality
joins, this generalization performs consistently better than the state
of the art, with up to \oomtime orders of magnitude improvements in
processing time and memory consumption.

\paragraph*{\bf Contributions} 
We focus on the class of Generalized Conjunctive Queries (\gcqs for
short), which are conjunctive queries with $\theta$-joins, that are evaluated under
multiset semantics. 

(1) We devise a succinct and efficiently updatable data structure to dynamically process \gcqs. To this end, we first generalize the notions of acyclicity and free-connexity to queries with arbitrary $\theta$-joins (Section~\ref{sec:acycl-aggr-join}). Our data structure degrades gracefully: if a \gcq only contains equalities our approach inherits the worst-case optimality provided by \dyn.

(2) We present \edyn, a general framework for extending \dyn to
free-connex acyclic \gcqs. Our treatment is general in the sense that
the $\theta$-join predicates are treated abstractly. \edyn hence
applies to all predicates, not only inequality joins. We analyze the
complexity of \edyn, and identify properties of indexing structures
that are required in order for \edyn to support efficient
enumeration of results as well as efficient update
processing (Section~\ref{sec:gdyn}). 

(3) We instantiate \edyn to the particular case of inequality and
equality joins. We show that updates can be processed in
time $O(n^2\cdot\log(n))$, where $n$ is the size of the database plus the size of the update, and results can be enumerated with logarithmic delay. Moreover, if there is at most one inequality between
any pair of relations, updates take time $O(n\cdot\log(n))$ and enumeration is with constant delay. We call the resulting algorithm
\iedyn. We first illustrate this algorithm by means of an extensive
example (Section~\ref{sec:iedyn}), and then describe the required data
structures formally at the end of Section~\ref{sec:gdyn}.

(4) The operation of \edyn and \iedyn is driven by a \emph{Generalized
  Join Tree} (\gjt). \gjts are essentially query plans that specify
the data structure to be materialized, how it should be updated, and
how to enumerate the query results. We present an algorithm
that can be used both to check whether a \gcq is (free-connex) acyclic
and to construct a corresponding \gjt if this is the
case. (Section~\ref{sec:gyo}).
  
(5) We experimentally compare \iedyn with state-of-the-art \hivm and
\cer frameworks. 
\iedyn performs
consistently better, with up to \oomtime\ order of magnitude
improvements in both speed and memory consumption
(Section~\ref{sec:experiments} and Section~\ref{sec:exper-eval}).

We introduce the required background in
Section~\ref{sec:preliminaries}.

\paragraph*{Additional material} This article presents the following
additional contributions compared to its previously published
conference version~\cite{DBLP:journals/pvldb/IdrisUVVL18}:

(1) \emph{Correctness proofs.} The conference version only sketched
why \edyn and \iedyn work correctly and within the claimed bounds. In
contrast, here we formally prove correctness.

(2) \emph{Novel algorithm for computing GJTs.} As outlined, above,
\edyn and \iedyn work on acyclic \gcqs and their operation is driven
by the specification of a \gjt for such queries. 
The conference version  only stated that an
algorithm for checking acyclicity and free-connexity and computing
GJTs exists. In contrast, here, we fully present this algorithm and
illustrate its correctness.

\paragraph*{Additional related work}
In addition to the work already cited on \cer and (H)\ivm, our setting
is closely related to query evaluation with constant delay
enumeration~\cite{Bagan:2007,DBLP:journals/sigmod/Segoufin15,Berkholz:2017,braultbaron,Olteanu:2015,Bakibayev:2013,Olteanu:2015,Schleich:2016,dyn:2017,olteanu:fivm}. This
setting, however, deals with equi-joins only. Also related, although
restricted to the static setting, is the practical evaluation of
binary~\cite{DBLP:conf/vldb/DeWittNS91, DBLP:conf/vldb/HellersteinNP95,DBLP:conf/sigmod/EnderleHS04} and
multi-way~\cite{DBLP:journals/is/BernsteinG81,DBLP:conf/vldb/YoshikawaK84}
inequality joins. %
Our work, in contrast, considers dynamic processing of multi-way
$\theta$-joins, with a specialization to inequality joins. 
Recently, Khayyat et
al.~\cite{DBLP:journals/vldb/KhayyatLSOPQ0K17} proposed fast multi-way
inequality join algorithms based on sorted arrays and space efficient
bit-arrays. They focus on the case where there are exactly two
inequality conditions per pairwise join. While they also present an
incremental algorithm for pairwise joins, their algorithm makes no
effort to minimize the update cost in the case of multi-way joins. As
a result, they either materialize subresults (implying a space
overhead that can be more than linear), or recompute subresults. We do
neither.


\section{Preliminaries}
\label{sec:preliminaries}

\label{sec:ql}
Traditional conjunctive queries are cross products between relations,
restricted by equalities. Similarly, generalized conjunctive queries (\gcqs) are
cross products between relations, but restricted by arbitrary
predicates. We use the following notation for queries.

\paragraph*{\bf  Query Language}
Throughout the paper, let $x, y, z, \dots$ denote \emph{variables}
(also commonly called \emph{column names} or \emph{attributes}). A
\emph{hyperedge} is a finite set of variables. We use $\seq{x},
\seq{y}$, \dots to denote hyperedges.  A \gcq is an expression of the
form
\begin{equation}\label{eq:gcq}
    Q = \proj_{\seq{y}} \big(r_1(\seq{x_1}) \Join \dots \Join r_n(\seq{x_n})\mid \bigwedge_{i=1}^m \theta_i(\seq{z_i})\big)
\end{equation}
Here $r_1,\dots, r_n$ are \emph{relation symbols}; $\seq{x_1},\dots,
\seq{x_n}$ are hyperedges (of the same arity as $r_1,\dots, r_n$);
$\theta_1,\ldots,\theta_m$ are predicates over
$\seq{z_1},\dots,\seq{z_m}$, respectively; and both $\seq{y}$ and
$\bigcup_{i=1}^m\seq{z_i}$ are subsets of $\bigcup_{i=1}^n
\seq{x_i}$. We treat predicates abstractly: for our purpose, a
predicate over $\seq{x}$ is a (not necessarily finite) decidable set
$\theta$ of tuples over $\seq{x}$. For example, $\theta(x,y)=x < y$ is
the set of all tuples $(a,b)$ satisfying $a<b$. We indicate that
$\theta$ is a predicate over $\seq{x}$ by writing
$\theta(\seq{x})$. Throughout the paper, we consider only non-nullary
predicates, i.e., predicates with $\seq{x} \not = \emptyset$.

\begin{example}
  \label{ex:running-gcq}
  The following query is a \gcq.
  \begin{equation*}
     \proj_{y,z,w,u} \big(r(x,y)\Join s(y,z,w)\Join t(u,v) \mid x<z \wedge w<u\big)
  \end{equation*}
  Intuitively, the query asks for the natural join between $r(x,y)$, $s(y,z,w)$, and $t(u,v)$, and from this
  result select only those tuples that satisfy both $x < z$ and $w < u$.
\end{example}
We call $\seq{y}$ the \emph{output variables} of $Q$ and denote
it by $\free(Q)$. If $\seq{y} = \seq{x_1} \cup \dots \cup \seq{x_n}$
then $Q$ is called a \emph{full query} and we may omit the symbol
$\pi_{\seq{y}}$ altogether for brevity. The elements $r_i(\seq{x_i})$ are called \emph{atomic queries} (or
\emph{atoms}). We write $\atoms(Q)$ for the set of all atoms in $Q$,
and $\predicates(Q)$ for the set of all predicates in $Q$.  A
\emph{normal conjunctive query} (\cq for short) is a \gcq where
$\predicates(Q) = \emptyset$.

\paragraph*{Semantics} We evaluate \gcqs over Generalized Multiset
Relations (GMRs for short)~\cite{DBLP:journals/vldb/KochAKNNLS14,DBLP:conf/pods/Koch10,dyn:2017}. A GMR
over $\seq{x}$ is a relation $R$ over $\seq{x}$ (i.e., a finite
set of tuples with schema $\seq{x}$) in which each tuple $\tup{t}$ is
associated with a non-zero integer multiplicity $R(\tup{t}) \in
\mathbb{Z} \setminus \{0\}$.\footnote{In their full generality, GMRs can carry
  multiplicities that are taken from an arbitrary algebraic ring
  structure (cf., \cite{DBLP:conf/pods/Koch10}), which can be useful
  to describe the computation of aggregations over the result of a
  \gcq. To keep the notation and discussion simple, we fix the ring
  $\mathbb{Z}$ of integers throughout the paper but our result
  generalize trivially to arbitrary rings.} In
contrast to classical multisets, the multiplicity of a tuple in a GMR
can hence be negative, allowing to treat insertions and deletions
uniformly.  We write $\supp(R)$ for the finite
set of all tuples in $R$; $\tup{t} \in R$ to indicate $\tup{t} \in
\supp(R)$; and $\card{R}$ for $\card{\supp(R)}$.
A GMR $R$ is \emph{positive} if $R(\tup{t}) >
0$ for all $\tup{t} \in \supp(R)$.

The operations of GMR union ($R+S$), minus $(R - S)$, projection
($\proj_{\seq{z}} R$), natural join ($R \Join T$) and selection
($\sigma_P(R)$) are defined similarly as in relational algebra with
multiset semantics.
Figure~\ref{fig:semantics}
illustrates these operations.
We refer to \cite{dyn:2017,
  DBLP:journals/vldb/KochAKNNLS14} for a formal semantics. We abbreviate $\sigma_P(R\Join T)$ by 
  $R\Join_PT$ and, if $\seq{x}=\var(R)$, we abbreviate $\pi_{\seq{x}}(R\Join_P T)$ by $R\semijoin_PT$.

A \emph{database} over a set $\aset$ of atoms
is a function $\db$ that maps every atom $r(\seq{x}) \in \aset$ to a
positive \gmr $\db_{r(\seq{x})}$ over $\seq{x}$. Given a database
$\db$ over the atoms occurring in query $Q$, the evaluation of $Q$
over $\db$, denoted $Q(\db)$, is the GMR over $\seq{y}$ constructed in
the expected way: take the natural join of all GMRs in the
database, do a selection over the result w.r.t. each
predicate, and finally project on $\seq{y}$.

\begin{figure}[tbp]
  	\centering
	\[
	\begin{array}[t]{|lll|l|}
    	\multicolumn{4}{c}{R} \\ \hline
    	x & y & z & \mathbb{Z} \\ \hline
    	1 & 2 & 2 & 2\\
    	2 & 4 & 6 & 3\\
    	1 & 2 & 3 & 3\\\hline
	\end{array} 
	\qquad
	\begin{array}[t]{|ll|l|}
    	\multicolumn{3}{c}{S} \\ \hline
    	u & v & \mathbb{Z} \\ \hline
    	4 & 5 & 5\\
    	2 & 3 & 4 \\ 
    	1 & 4 & 2\\\hline
	\end{array} 
	\qquad
    \begin{array}[t]{|ll|r|}
        \multicolumn{3}{c}{T} \\ \hline
        u & v & \mathbb{Z} \\ \hline
        4 & 5 & -4\\
        2 & 1 & 6\\ 
    	1 & 4 & 3\\\hline
    \end{array}
    \qquad
	\begin{array}[t]{|ll|r|}
    	\multicolumn{3}{c}{S \Join T} \\ \hline
    	u & v & \mathbb{Z} \\ \hline
    	4 & 5 & -20\\
    	1 & 4 & 6\\\hline
	\end{array} 
	\]
	\[
	\begin{array}[t]{|r|r|}
    	\multicolumn{2}{c}{\proj_{y}(R)} \\ \hline
    	y & \mathbb{Z} \\ \hline
    	\phantom{\ }2 & 5\\
    	\phantom{\ }4 & 3\\\hline
	\end{array}
	\qquad 
	\begin{array}[t]{|ll|r|}
    	\multicolumn{3}{c}{S+T} \\ \hline
    	u & v & \mathbb{Z} \\ \hline
    	4 & 5 & 1\\
    	2 & 3 & 4 \\ 
        1 & 4 & 5 \\
    	2 & 1 & 6 \\ \hline
	\end{array} 
	\qquad 
	\begin{array}[t]{|ll|r|}
    	\multicolumn{3}{c}{S-T} \\ \hline
    	u & v & \mathbb{Z} \\ \hline
    	4 & 5 & 9\\
    	2 & 3 & 4 \\ 
        1 & 4 & -1 \\
    	2 & 1 & -6 \\ \hline
	\end{array} 
	\qquad 
	\begin{array}[t]{|lllll|r|}
    	\multicolumn{6}{c}{R \Join_{y<u} S} \\ \hline
    	x & y & z & u & v & \mathbb{Z} \\ \hline
    	1 & 2 & 2 & 4 & 5 & 10\\
        1 & 2 & 3 & 4 & 5 & 15\\ \hline
	\end{array} 
	\]
  	\label{fig:semantics}
  	\caption{Operations on GMRs}
\end{figure}

\paragraph*{\bf \bf Updates and deltas}
\label{sec:updates-deltas}

An \emph{update} to a \gmr $R$ is simply a \gmr $\delt R$ over the same variables as $R$. Applying update $\delt R$ to $R$ yields the GMR $R + \delt R$. An \emph{update to a database} $\db$ is a collection $u$ of (not necessarily positive) GMRs, one GMR $u_{r(\overline{x})}$ for every atom $r(\seq{x})$ of $\db$, such that $\db_{r(\seq{x})} + \upd_{r(\seq{x})}$ is positive. We write $\db + u$ for the database obtained by applying $u$ to each atom of $\db$, i.e., $(\db + u)_{r(\seq{x})} = \db_{r(\seq{x})} + u_{r(\seq{x})}$, for every atom $r(\seq{x})$ of $\db$. For every query $Q$, every database $\db$ and every update $\upd$ to $\db$, we define the delta query $\delt{Q}(\db,\upd)$ of $Q$ w.r.t. $\db$ and $\upd$ by
  $\sem{\delt Q}(\db, \upd) := \sem{Q}(\db+\upd) - \sem{Q}(\db)$.
As such, $\delt Q(\db,\upd)$ is the update that we need to apply to
$Q(\db)$ in order to obtain $Q(\db + \upd)$.

\paragraph*{\bf Enumeration with bounded delay}

A data structure $D$ supports \emph{enumeration} of a set $E$ if there
is a routine $\routenum$ such that $\routenum(D)$ outputs each element
of $E$ exactly once. Such enumeration occurs with delay $d$ if the
time until the first output; the time between any two consecutive
outputs; and the time between the last output and the termination of
$\routenum(D)$, are all bounded by $d$.  $D$ supports
enumeration of a GMR $R$ if it supports enumeration of the set $E_R =
\{ (\tup{t}, R(\tup{t})) \mid \tup{t} \in \supp(R)\}$.  When
evaluating a GCQ $Q$, we will be interested in representing the possible
outputs of $Q$ by means of a family ${\cal D}$ of data structures, one
data structure $D_{\db} \in {\cal D}$ for each possible input database $\db$. We say that
\emph{$Q$ can be enumerated from ${\cal D}$ with delay $f$}, if
for every input $\db$ we can enumerate $Q(\db)$ from $D_{\db}$ with
delay $O(f(D_{db}))$, where $f$ assigns a natural number to each
$D_{\db}$. Intuitively $f$ measures $D_{\db}$ in some way. In
particular, if $f$ is constant we say the results are generated from
the data structure with \emph{constant-delay enumeration} (CDE).

As a trivial example of \cde of a GMR $R$, assume that the pairs
$(\tup{t}, R(\tup{t}))$ of $E_R$ are stored in an array $A$ (without
duplicates). Then $A$ supports \cde of $R$: $\routenum(A)$ simply
iterates over each element in $A$, one by one, always outputting the
current element. Since array indexation is a $O(1)$ operation, this
gives constant delay. This example shows that \cde of the result
$Q(\db)$ of a query $Q$ on input database $\db$, can always be done
naively by materializing $Q(\db)$ in an in-memory array $A$.
Unfortunately, $A$ then requires memory proportional to
$\size{Q(\db)}$ which, depending on $Q$, can be of size polynomial in
$\size{\db}$. We hence search for other data structures that can
represent $Q(\db)$ using less space, while still allowing for efficient enumeration. Our experiments
in Section~\ref{sec:experiments} show that for the data structures
described in this paper, \cde is indeed competitive with enumeration
from an array while requiring much less space.

\paragraph*{\bf Computational Model}

It is important to note that we focus on dynamic query evaluation in
main memory.
Furthermore, we assume
a model of computation where the space used by tuple values and integers, the time of arithmetic operations on integers, and the time of memory lookups
are all $O(1)$. We also assume that every GMR $R$ can be
represented by a data structure that allows (1) enumeration of $R$
with constant delay; (2)
multiplicity lookups $R(\tup{t})$ in $O(1)$ time given $\tup{t}$; (3)
single-tuple insertions and deletions in $O(1)$ time; while (4) having
a size that is proportional to the number of tuples in the support of
$R$. Essentially, our assumptions amount to perfect hashing of linear
size~\cite{cormen2009introduction}. Although this is not realistic for
practical computers \cite{Papadimitriou:2003}, it is well known that
complexity results for this model can be translated, through amortized
analysis, to average complexity in real-life
implementations~\cite{cormen2009introduction}.


\section{Generalized Acyclicity}
\label{sec:acycl-aggr-join}

Join queries are \gcqs without projections that feature equality joins
only. The well-known subclass of acyclic join
queries~\cite{abiteboul1995foundations,DBLP:conf/vldb/Yannakakis81}, in contrast to the
entire class of join queries, can be evaluated in time $O(\size{\db} +
\size{Q(\db)})$, i.e., linear in both input and output. This result
relies on the fact that acyclic join queries admit a tree structure
that can be exploited during evaluation. In previous
work~\cite{dyn:2017}, we showed that this tree structure can also be
exploited for efficient processing of \cqs under updates.  In this
section, we therefore extend the tree structure and the notion of
acyclicity from join queries to \gcqs with both projections and
arbitrary $\theta$-joins.  We begin by defining this tree structure
and the related notion of acyclicity for full \gcqs. Then, we proceed
with the notion corresponding to \gcqs that feature projections, known
as free-connex acyclicity. 

\paragraph*{\bf Generalized Join Trees} To simplify notation, we
denote the set of all variables (resp. atoms, resp. predicates) that
occur in an object $X$ (such as a query) by $\var(X)$
(resp. $\atoms(X)$, resp. $\preds(X)$). In particular, if $X$ is
itself a set of variables, then $\var(X) = X$. We extend this notion
uniformly to labeled trees. E.g., if $n$ is a node in tree $T$, then
$\var_T(n)$ denotes the set of variables occurring in the label of
$n$, and similarly for edges and trees themselves. Finally, we write
$\child_T(n)$ for the set of children of $n$ in tree $T$. If $T$ is
clear from the context, we omit subscripts from our notation.

\begin{definition}[\gjt]
  \label{def:gjt}
  A \emph{Generalized Join Tree} (\gjt) is a node-labeled and
  edge-labeled directed tree $T = (V,E)$ such that:
    \begin{itemize}
    \item Every leaf is labeled by an atom.
    \item Every interior node $n$ is labeled by a hyperedge and has at
      least one child $c$ such that $\var(n)\subseteq\var(c)$.
    \item Whenever the same variable $x$ occurs in the label of two
      nodes $m$ and $n$ of $T$, then $x$ occurs in the label of
      each node on the unique path linking $m$ and
      $n$. This condition is called the \emph{connectedness condition}.
    \item Every edge $p \to c$ from parent $p$ to child $c$ in $T$ is
      labeled by a set $\predicates(p \to c)$ of predicates. It is
      required that for every predicate $\theta(\seq{z}) \in \predicates(p
      \to c)$ we have $\var(\theta) = \seq{z} \subseteq \var(p) \cup \var(c)$.
    \end{itemize}
\end{definition}

Let $n$ be a node in \gjt $T$. Every node $m$ with $\var(n) \subseteq
\var(m)$ is called a \emph{guard} of $n$. Observe that every interior
node must have a guard child by the second requirement above. Since
this child must itself have a guard child, which must itself have a
guard child, and so on, it holds that every interior node has at least
one guard descendant that is a leaf.

\begin{definition}
  \label{def:acyclicity}
  A \gjt $T$ is a \emph{\gjt for \gcq $Q$} if $\atoms(T) = \atoms(Q)$
  and the number of times that an atom occurs in $Q$ equals the number
  of times that it occurs as a label in $T$, and $\preds(T) =
  \preds(Q)$. A \gcq $Q$ is \emph{acyclic} if there is a \gjt for
  $Q$. It is \emph{cyclic} otherwise.
\end{definition}

\begin{example}\label{ex:running}
  The two trees depicted in Fig. \ref{fig:extrees} are \gjts for the following
  full \gcq $Q$, which is hence acyclic.
  \begin{equation*}\label{query:running-example}
    Q_1 =  \big(r(x,y)\Join s(y,z,w)\Join t(u,v) \mid x<z \wedge w<u\big)
  \end{equation*}
  In contrast, the query $r(x,y) \Join s(y,z) \Join t(x,z)$ (also
  known as the triangle query) is the prototypical cyclic join
  query. 
\end{example}

\begin{figure}
  \centering
    \begin{tikzpicture}[-,>=stealth',thick]
        \begin{scope}[xshift=-4cm]
            \node (w) [label={above:$(T_1)$}] { $\{y,w\}$ } [level distance=1cm,sibling distance=1.9cm]
            child {
                node (xu) { $\{y,z,w\}$ } [sibling distance=1.3cm]
                child {
                    node (xyz) { $r(x,y)$ } edge from parent node[left] {\scriptsize $x<z$}
                }
                child{
                    node (xut) { $s(y,z,w)$ }
                }
            }
            child{
                node (xu) { $t(u,v)$ } edge from parent node[right] {\scriptsize $w<u$}
            };
        \end{scope}

        \begin{scope}
            \node (w) [label={above:$(T_2)$}] { $\{y,w\}$ } [level distance=1cm,sibling distance=1.9cm]
            child {
                node (xu) { $\{y,z,w\}$ } [sibling distance=1.3cm]
                child {
                    node (xyz) { $r(x,y)$ } edge from parent node[left] {\scriptsize $x<z$}
                }
                child{
                    node (xut) { $s(y,z,w)$ }
                }
            }
            child{
                node (xu) { $\{u\}$ }
                child {
                    node (uv) { $t(u,v)$ }
                }
                edge from parent node[right] {\scriptsize $w<u$}
            };
        \end{scope}
    \end{tikzpicture}
  \caption{Two example \gjts.}
  \label{fig:extrees}
\end{figure}

If $Q$ does not contain any predicates, that is, if $Q$ is a \cq, then
the last condition of Definition~\ref{def:gjt} vacuously holds. In
that case, the definition corresponds to the definition of a
generalized join tree given in \cite{dyn:2017}, where it was also
shown that a \cq is acyclic under any of the traditional definitions
of acyclicity (e.g., \cite{abiteboul1995foundations}) if and only if the query has a
\gjt $T$ for $Q$ with $\preds(T) = \emptyset$. In this sense, Definition~\ref{def:acyclicity}
indeed generalizes acyclicity from \cqs to \gcqs. 

\paragraph*{Discussion}
\newcommand{\semijoinprog}{\ensuremath{\mathcal{S}}} The notion of
ayclicity for normal \cqs is well-studied in database
theory~\cite{abiteboul1995foundations} and has many equivalent
definitions, including a definition based on the 
existence of a \emph{full reducer}. Here, a full reducer for a \cq $Q$
is a program $\semijoinprog$ in the semijoin algebra (the variant of
relational algebra where joins are replaced by semijoins) that, given
a database $\db$ computes a new database $\semijoinprog(\db)$ with the
following properties. (1) $Q(\semijoinprog(\db)) = Q(\db)$; (2)
$\semijoinprog(\db)_{r(\seq{x})} \subseteq \db_{r(\seq{x})}$ for every
atom $r(\seq{x})$; and (3) no strict subset of
$\semijoinprog(\db)$ has $Q(\semijoinprog(\db)) = Q(\db)$. In other
words, $\semijoinprog$ selects a \emph{minimal} subset of $\db$
needed to answer $Q$.

Bernstein and Goodman~\cite{DBLP:journals/is/BernsteinG81} consider
conjunctive queries with inequalities and classify the class of such
queries that admit full reducers. As such, one can view this as a
definition of acyclicity for conjunctive queries with
inequalities. Bernstein and Goodman's notion of acyclicity is
incomparable to ours. On the on hand, our definition is more general:
Bernstein and Goodman consider only queries where for each pair of
atoms there is exactly one variable being compared by means of
equality or inequality. We, in contrast, allow an arbitrary number of
variables to be compared per pair of atoms. In particular, Bernstein
and Goodman's disallow queries like $(r(x,y), s(x,z) \mid y < z)$
since it compares $r.x$ with $s.x$ by means of equality and $r.y <
s.z$ by means of inequality, while this is trivially acyclic in our setting.

On the other hand, for this more restricted class of queries,
Bernstein and Goodman show that certain queries that we consider to be
cyclic have full reducers (and would be hence acyclic under their
notion). An example here is
\begin{multline*}
 r(x_r) \Join s(x_s,
y_s) \Join t(x_t, y_t) \Join u(y_u) \\ \mid x_s \leq x_r, x_t \leq x_r,
y_s \leq y_u, y_t \leq y_u 
\end{multline*}
The crucial reason that this query admits a full reducer is due to the
transitivity of $\leq$. Since our notion of acyclicity interprets
predicates abstractly and does hence not assume properties such as
transitivity on them, we must declare this query cyclic (as can be
checked by running the algorithm of Section~\ref{sec:gyo} on
it). It is an interesting direction for future work to incorporate
Bernstein and Goodman's notion of acyclicity in our framework.

\paragraph*{\bf Free-connex acyclicity} Acyclicity is actually a notion
for full \gcqs. Indeed, note that whether or not $Q$ is acyclic does not depend on the projections of $Q$ (if any). To also process
queries with projections efficiently, a related structural constraint
known as free-connex acyclicity is required. 

\begin{definition}[Connex, Frontier]
  Let $T = (V,E)$ be a \gjt. A \emph{connex subset} of $T$ is a set $N
  \subseteq V$ that includes the root of $T$ such that the subgraph of
  $T$ induced by $N$ is a tree. The \emph{frontier} of a connex set
  $N$ is the subset $F \subseteq N$ consisting of those nodes in $N$
  that are leaves in the subtree of $T$ induced by $N$.
\end{definition}
To illustrate, the set $\{\{y,w\},\{u\}, \{y,z,w\}\}$ is a connex subset of
  the tree $T_2$ shown in  Fig.~\ref{fig:extrees}. Its frontier is $\{ \{y,z,w\},
  \{u\}\}$. In contrast, $\{\{y,w\}, \{y,z,w\}, \allowbreak t(u,v) \}$  is not a
  connex subset of $T_2$.

  \begin{definition}[Compatible, Free-Connex Acyclic] A \emph{GJT
      pair} is a pair $(T, N)$ with $T$ a \gjt and $N$ a connex subset
    of $T$. A \gcq $Q$ is \emph{compatible with} $(T, N)$ if
    $T$ is a \gjt for $Q$ and $\var(N) = \free(Q)$. A \gcq is
    \emph{free-connex acyclic} if it has a compatible GJT pair.
\end{definition}

In particular, every full acyclic \gcq is free-connex acyclic since
the entire set of nodes $V$ of a \gjt $T$ for $Q$ is a connex set with
$\var(V) = \free(Q)$. Therefore, $(T, V)$ is a compatible \gjt pair for $Q$.

\begin{example}
  \label{ex:connexity}
  Let $Q_2 = \proj_{y,z,w,u}(Q_1)$ with $Q_1$ the \gcq from Example
  \ref{ex:running}. $Q_2$ is free-connex acyclic since it is
  compatible with the pair $(T_2, \{ \{y,w\}, \{y,z,w\}, \{u\}\})$
  with $T_2$ the \gjt from Fig.~\ref{fig:extrees}. By contrast,
  $Q_2$ is not compatible with any \gjt pair containing $T_1$, since
  any connex set of $T_1$ that includes a node with variable $u$ will
  also include variable $v$, which is not in $\free(Q_2)$.  Finally,
  it can be verified that no \gjt pair is compatible with
  $\proj_{x,u}(Q_1)$; this query is hence not free-connex
  acyclic.
\end{example}

In Section~\ref{sec:gyo} we show how to efficiently check
free-connex acyclicity and compute compatible \gjt pairs.

\paragraph*{\bf Binary \gjts and sibling-closed connex sets} As we
will see in Sections~\ref{sec:iedyn} and \ref{sec:gdyn}, a \gjt pair
$(T,N)$ essentially acts as query plan by which \edyn and \iedyn
process queries dynamically.  In particular, the \gjt $T$ specifies
the data structure to be maintained and drives the processing of
updates, while the connex set $N$ drives the enumeration of query
results.

In order to simplify the presentation of what follows, we will focus
exclusively on the class of \gjt pairs $(T,N)$ with $T$ a binary \gjt
and $N$ sibling-closed. 

\begin{definition}[Binary, Sibling-closed]
  A \gjt $T$ is binary if every node in it has at most two children. A
  connex subset $N$ of $T$ is \emph{sibling-closed} if for every node
  $n \in N$ with a sibling $m$ in $T$, $m$ is also in $N$.
\end{definition}

Our interest in limiting to sibling-closed connex sets is due to
the following property, which will prove useful
for enumerating query results, as explained in Section~\ref{sec:iedyn}.

\begin{lemma}
  \label{lem:sibling-closed-var-at-frontier}
  If $N$ is a sibling-closed connex subset, then $\var(N) = \var(F)$
  where $F$ is the frontier of $N$.
\end{lemma}
\begin{proof}
  Since $F \subseteq N$ the inclusion $\var(F) \subseteq \var(N)$ is
  immediate. It remains to prove $\var(N) \subseteq \var(F)$. To this
  end, let $n$ be an arbitrary but fixed node in $N$. We prove that
  $\var(n) \subseteq \var(F)$ by induction on the \emph{height} of $n$
  in $N$, which is defined as the length of the shortest path from $n$
  to a frontier node in $F$. The base case is where the height is
  zero, i.e., $n \in F$, in which case $\var(n) \subseteq \var(F)$
  trivially holds. For the induction step, assume that the height of
  $n$ is $k > 0$. In particular, $n$ is not a frontier node, and has
  at least one child in $N$. Because $N$ is sibling-closed, all
  children of $n$ are in $N$. In particular, the guard child $m$ of
  $n$ is in $N$ and has height at most $k-1$. By induction hypothesis,
  $\var(m) \subseteq \var(F)$. Then, because $m$ is a guard of $n$,
  $\var(n) \subseteq \var(m) \subseteq \var(F)$, as desired.
\end{proof}

Let us call a \gjt pair $(T,N)$ \emph{binary} if $T$ is binary, and
\emph{sibling-closed} if $N$ is sibling-closed. We say that two \gjt
pairs $(T,N)$ and $(T', N')$ are \emph{equivalent} if $T$ and $T'$ are
equivalent and $\var(N) = \var(N')$. Two GJTs $T$ and $T'$ are
equivalent if $\atoms(T) = \atoms(T')$, the number of times that an
atom appears as a label in $T$ equals the number of times that it
appears in $T'$, and $\preds(T) = \preds(T')$.

The following proposition shows that we can always convert an
arbitrary \gjt pair into an equivalent one that is binary and
sibling-closed. As such, we are assured that our focus on binary and
sibling-closed \gjt pairs  is without loss of
generality.  

\begin{proposition}
\label{prop:bin-sc-equivalent}
Every \gjt pair can be transformed in polynomial time into an
equivalent pair that is binary and sibling closed.
\end{proposition}

The rest of this section is devoted to proving
Proposition~\ref{prop:bin-sc-equivalent}. We do so in two
steps. First, we show that any pair $(T, N)$ can be transformed in
polynomial time into an equivalent sibling-closed pair. Next, we show
that any sibling-closed \gjt pair $(T, N)$  can be converted
in polynomial time into an equivalent binary and sibling-closed
pair. Proposition~\ref{prop:bin-sc-equivalent} hence follows by
composing these two transformations.

\begin{figure}[t]
  \centering
\tikzset{
    expand bubble/.style={
        preaction={draw,line width=1.4pt},
        gray!20,fill,draw,line width=2pt,
    },
}

    \begin{tikzpicture}[-,>=stealth',thick,every node/.style={transform shape},scale=0.98]
\begin{scope}[xshift=-4cm]
            \node (w) [label={above:$(T,N)$}] { $\{y,z,w\}$ } [level distance=1cm,sibling distance=1.3cm]
            child {
                node (xu) { $\{y,z,w\}$ } [sibling
                distance=1.3cm,level distance=1cm]
                child {
                    node (xyz) { $f(u,v)$ } edge from parent node[left] {\scriptsize $w<u$}
                }
                child{
                    node (xut) { $h(y,z,w,t)$ }
                }
            }
            child{
                    node (uv) { $r(x,y)$ }
            }
            child{
                    node (uw) { $s(y,z,m)$ }
                edge from parent node[right] {\scriptsize $m<w$}
            };
            \begin{pgfonlayer}{background}
              \path[expand bubble]plot [smooth cycle,tension=1]
              coordinates {(w.north) (xu.west) (uv.east)};
            \end{pgfonlayer}
        \end{scope}

        \begin{scope}
            \node (w) [label={above:$(T',N')$}] { $\{y,z,w\}$ } [level distance=1cm,sibling distance=1.9cm]
            child {
                node (xu) { $\{y,z,w\}$ } [sibling distance=1.3cm]
                child {
                    node (xyz) { $f(u,v)$ } edge from parent node[left] {\scriptsize $w<u$}
                }
                child{
                    node (xut) { $\{y,z,w\}$ }[sibling distance=1.5cm]
                    child{
                    	node(sss){$h(y,z,w,t)$}
                    }
                    child{
                    node (uw) { $s(y,z,m)$ }
                    edge from parent node[right] {\scriptsize $m<w$}
                    }
                }
            }
            child{
                    node (uv) { $r(x,y)$ }
            };
            \begin{pgfonlayer}{background}
              \path[fill=gray!20]plot [smooth cycle,tension=1]
              coordinates {(w.north) (xu.south west) (uv.south east)};
            \end{pgfonlayer}
        \end{scope}
    \end{tikzpicture}
  \caption{Illustration of the sibling-closed transform: removal of
    type-1 violator. The connex sets $N$ and $N'$ are indicated by the
  shaded areas.}
  \label{fig:sc-transform}
\end{figure}
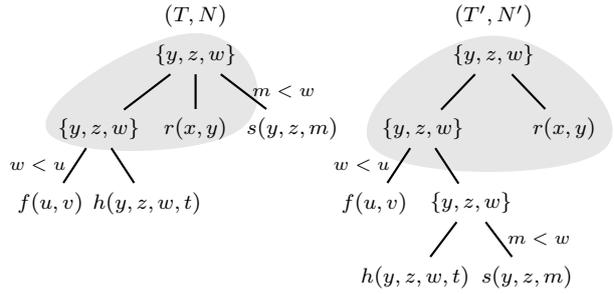

\paragraph*{\bf Sibling-closed transformation}
\label{sec:sc-transform}

We say that $n \in T$ is a \emph{violator} node in a \gjt pair
$(T,N)$  if $n \in N$ and some, but not all children of $n$ are in
$N$. A violator is of \emph{type 1} if some node in $\child(n) \cap N$
is a guard of $n$. It is of \emph{type 2} otherwise. We now define two
operations on $(T,N)$ that remove violators of type 1 and type 2,
respectively. The sibling-closed transformation is then obtained by
repeatedly applying these operators until all violators are removed.

The first operator is applicable when $n$ is a type 1
violator. It returns the pair $(T',N')$ obtained as follows:
\begin{itemize}
\item Since $n$ is a type 1 violator, some $g \in \child_T(n) \cap N$
  is a child guard of $n$ (i.e., $\var(n) \subseteq \var(g)$).
\item Because every node has a guard, there is some leaf node $l$ that
  is a descendant guard of $g$ (i.e. $\var(g)
  \subseteq \var(l)$). Possibly, $l$ is $g$ itself.
\item Now create a new node $p$ between node $l$ and its parent with
  label $\var(p)=\var(l)$. Since $l$ is a descendant guard of $n$ and $g$, $p$
  becomes a descendant guard of $n$ and $g$ as well. Detach all nodes in
  $\child(n) \setminus N$ from $n$ and attach them as children to $p$,
  preserving their edge labels. This effectively moves all subtrees
  rooted at nodes in $\child(n) \setminus N$ from $n$ to $p$. Denote
  by $T'$ the final result.
\item If $l$ was not in $N$, then $N' = N$. Otherwise, $N' = N
  \setminus \{l\} \cup \{p\}$.
\end{itemize}
We write $(T,N) \xrightarrow{1,n} (T',N')$ to indicate that $(T',N')$
can be obtained by applying the above-described operation on node
$n$. 

\begin{example}
  Consider the \gjt pair $(T,N)$ from Fig.~\ref{fig:sc-transform}
  where $N$ is indicated by the nodes in the shaded area. Let us
  denote the root node by $n$ and its guard child with label
  $\{y,z,w\}$ by $g$. The node $l = h(y,z,w,t)$ is a descendant guard
  of $g$. Since $s(y,z,m)$ is not in $N$, $n$ is violator of type
  1. After applying the operation 1 for the choice of guard node $g$
  and descendant guard node $l$, $(T',N')$ shows the resulting valid
  sibling-closed \gjt.
\end{example}

\begin{restatable}{lemma}{lemmaRemOne}
\label{lemma:op1}
  Let $n$ be a violator of type $1$ in $(T,N)$ and assume $(T,N)
  \xrightarrow{1,n} (T',N')$. Then $(T',N')$ is a \gjt pair and it is
  equivalent to $(T,N)$. Moreover, the number of violators in
  $(T',N')$ is strictly smaller than the number of violators in
  $(T,N)$.
\end{restatable}

We prove this lemma in Appendix~\ref{sec:app-acyclicity}.
The second operator is applicable when $n$ is a type 2 violator. When
applied to $n$ in $(T,N)$ it returns the pair $(T',N')$ obtained as
follows:
\begin{itemize}
\item Since $n$ is a type 2 violator, no node in $\child_T(n) \cap N$ is
  a guard of $n$. Since every node has a guard,
  there is some $g \in \child(n) \setminus N$ which is a guard of $n$.
\item Create a new child $p$ of $n$ with label $\var(p) = \var(n)$;
  detach all nodes in $\child(n)\setminus N$ (including $g$) from $N$,
  and add them as children of $p$, preserving their edge labels. This
  moves all subtrees rooted at nodes in $\child(n)
  \setminus N$ from $n$ to $p$. Denote by $T'$ the final result.
\item Set $N' = N \cup \{p\}$.
\end{itemize}
We write $(T,N) \xrightarrow{2,n} (T',N')$ to indicate that $(T',N')$
was obtained by applying this operation on $n$.

\begin{example}
  Consider the \gjt pair $(T,N)$ in
  Fig.~\ref{fig:sc2-transform}. Let us denote the root node by
  $n$. Since its guard child $h(y,z,w,t)$ is not in $N$, $n$ is
  violator of type 2. After applying operation 2 on $n$, $(T',N')$
  shows the resulting valid sibling-closed \gjt.
\end{example}

\begin{figure}[t]
  \centering
\tikzset{
    expand bubble/.style={
        preaction={draw,line width=1.4pt},
        gray!20,fill,draw,line width=2pt,
    },
}
    \begin{tikzpicture}[-,>=stealth',thick]
\begin{scope}[xshift=-4cm]
            \node (w) [label={above:$(T,N)$}] { $\{y,z,w\}$ } [level distance=1cm,sibling distance=1.3cm]
            child {
                node (xu) { $h(y,z,w,t)$ } 
            }
            child{
                    node (uv) { $r(x,y)$ }
            }
            child{
                    node (uw) { $s(y,z,m)$ }
                edge from parent node[right] {\scriptsize $m<w$}
            };
            \begin{pgfonlayer}{background}
              \path[expand bubble]plot [smooth cycle,tension=1]
              coordinates {(w.west) (uv.south) (w.east)};
            \end{pgfonlayer}
        \end{scope}

        \begin{scope}
            \node (w) [label={above:$(T',N')$}] { $\{y,z,w\}$ } [level distance=1cm,sibling distance=1.9cm]
            child {
                node (xu) { $\{y,z,w\}$ } [sibling distance=1.5cm]
                    child{
                    	node(sss){$h(y,z,w,t)$}
                    }
                    child{
                    node (uw) { $s(y,z,m)$ }
                    edge from parent node[right] {\scriptsize $m<w$}
                    }
            }
            child{
                    node (uv) { $r(x,y)$ }
            };
            \begin{pgfonlayer}{background}
              \path[fill=gray!20]plot [smooth cycle,tension=1]
              coordinates {(w.north) (xu.south west) (uv.south east)};
            \end{pgfonlayer}
        \end{scope}
    \end{tikzpicture}
  \caption{Illustration of the sibling-closed transform: removal of
    type-2 violator. The connex sets $N$ and $N'$ are indicated by the
  shaded areas.}
  \label{fig:sc2-transform}
\end{figure}
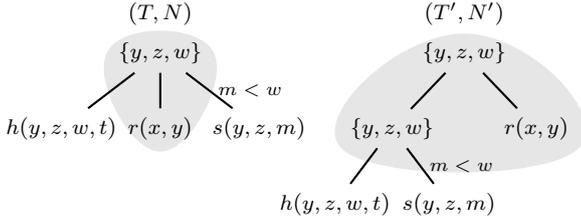

\tikzset{
itria/.style={
  draw,dashed,shape border uses incircle,
  isosceles triangle,shape border rotate=90,yshift=-0.80cm, minimum height=0.6cm,
        minimum width=0.4cm,
        shape border rotate=#1,
        isosceles triangle stretches,
        inner sep=0pt},
rtria/.style={
  draw,dashed,shape border uses incircle,
  isosceles triangle,isosceles triangle apex angle=90,
  shape border rotate=-45,yshift=0.2cm,xshift=0.5cm},
ritria/.style={
  draw,dashed,shape border uses incircle,
  isosceles triangle,isosceles triangle apex angle=110,
  shape border rotate=-55,yshift=0.1cm},
letria/.style={
  draw,dashed,shape border uses incircle,
  isosceles triangle,isosceles triangle apex angle=110,
  shape border rotate=235,yshift=0.1cm}
}
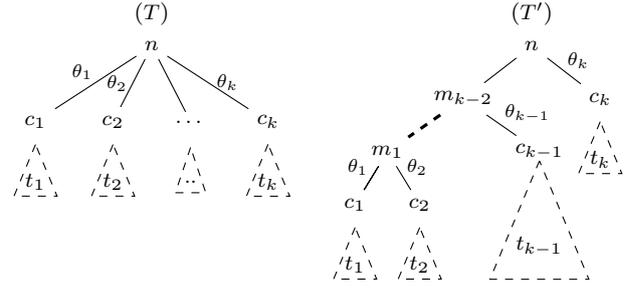
\begin{figure}[t]
\centering
\begin{tikzpicture}[sibling distance=1.02cm, level 2/.style={sibling distance =2cm},
emph/.style={edge from parent/.style={dashed,black,very thick,draw}},
norm/.style={edge from parent/.style={solid,black,thin,draw}}]
\begin{scope}[xshift=-5cm, level distance=1cm]
\node[label={above:$(T)$}] {$n$}
    child{ node[] {$c_1$}    
    { node[itria] {$t_1$} }
    edge from parent node[above left=-.1] {\scriptsize $\theta_1$}
    }
     child{ node[] {$c_2$}
    { node[itria] {$t_2$} }
    edge from parent node[left=-.04] {\scriptsize $\theta_2$}
    }
     child{ node[] {$\dots$}
     { node[itria] {$..$} }
    }
    child{ node[] {$c_k$}
            { node[itria] {$t_k$} }
    edge from parent node[right] {\scriptsize $\theta_k$}
        };
\end{scope}
\begin{scope}
\node[label={above:$(T')$}] {$n$}[sibling distance = 1.8cm, level distance=0.7cm]
    child[norm] { node[] {$m_{k-2}$}[sibling distance= 0.85cm]
                    child[emph] {
                        node[] {$m_1$}[sibling distance=0.85cm]
                            child[norm]{ node[] {$c_1$}{
                            node[itria] {$t_1$}
                            edge from parent node[above left=-.1] {\scriptsize $\theta_1$}
                            } }
                            child[norm] { node[] {$c_2$} 
                                     { node[itria] {$t_2$} }                         
                            edge from parent node[above right=-.1] {\scriptsize $\theta_2$}
                                }
                        }
                child[norm]{ node[] {$c_{k-1}$}
                        { node[itria,yshift=-13] {$t_{k-1}$} } 
                  edge from parent node[above right=-.1] {\scriptsize $\theta_{k-1}$}
                }
        }
    child[norm]{ node[] {$c_k$}
            { node[itria] {$t_k$} }      
                  edge from parent node[above right=-.1] {\scriptsize $\theta_{k}$}
        };
\end{scope}
\end{tikzpicture}
\caption{Binarizing a $k$-ary node $n$. }
\label{fig:binary}
\end{figure}

\begin{restatable}{lemma}{lemmaRemTwo}
\label{lemma:op2}
  Let $n$ be a violator of type $2$ in $(T,N)$ and assume $(T,N)
  \xrightarrow{2,n} (T',N')$. Then $(T',N')$ is a \gjt pair and it is
  equivalent to $(T,N)$. Moreover, the number of violators in
  $(T',N')$ is strictly smaller than the number of violators in
  $(T,N)$.
\end{restatable}

The proof can be found in Appendix~\ref{sec:app-acyclicity}.
\begin{proposition}
\label{prop:sctransform}
  Every GJT pair can be transformed in polynomial time into an equivalent
  sibling-closed pair.
\end{proposition}
\begin{proof}
  The two operations introduced above remove violators, one at a
  time. By repeatedly applying these operations until no violator
  remains we obtain an equivalent pair without violators, which must
  hence be sibling-closed. Since each operator can clearly be executed
  in polynomial time and the number of times that we must apply
  an operator is bounded by the number of nodes in the GJT
  pair, the removal takes polynomial time.
\end{proof}

\paragraph*{\bf Binary transformation} Next, we show how to transform
a sibling-closed pair $(T, N)$ into an equivalent binary and
sibling-closed pair $(T', N')$. The idea here is to ``binarize'' each
node $n$ with $k > 2$ children as shown in
Fig.~\ref{fig:binary}. There, we assume without loss of
generality that $c_1$ is a guard child of $n$. The binarization
introduces $k-2$ new intermediate nodes $m_1, \dots, m_{k-2}$, all
with $\var(m_i) = \var(n)$. Note that, since $c_1$ is a guard of $n$
and $\var(m_i) = \var(n)$, it is straightforward to see that $c_1$
will be a guard of $m_1$, which will be a guard of $m_2$, which will
be a guard of $m_3$, and so on. Finally, $m_{k-2}$ will be a guard of
$n$.  The connex set $N$ is updated as follows. If none of $n$'s
children are in $N$ i.e. $n$ is a frontier node, set
$N'=N$. Otherwise, since $N$ is sibling-closed, all children of $n$
are in $N$, and we set $N'=N\cup\{m_1,\dots,m_{k-2}\}$. Clearly, $N'$
remains a sibling-closed connex subset of $T'$ and $\var(N') =
\var(N)$. We may hence conclude:

\begin{lemma}
  By binarizing a single node in a sibling-closed GJT pair $(T, N)$ as
  shown in Fig.~\ref{fig:binary}, we obtain an equivalent GJT pair
  $(T', N')$ that has strictly fewer non-binary nodes than $(T,N)$.
\end{lemma}

Binarizing a single node is a polynomial-time
operation. Then, by iteratively binarizing non-binary nodes until all
nodes have become binary we hence obtain:
\begin{proposition}
\label{prop:bintransform}
  Every sibling-closed GJT pair can be transformed in polynomial time into an equivalent,
  binary and sibling-closed pair.
\end{proposition}


\section{Dynamic joins with equalities and inequalities: an example}
\label{sec:iedyn}

In this section we illustrate how to dynamically process free-connex acyclic
\gcqs when all predicates are inequalities $(\leq, <, \geq, >)$. We do
so by means of an extensive example that shows the
indexing structures and GMRs. The definitions and algorithms (that apply to arbitrary $\theta$-joins) will be formally presented in
Section~\ref{sec:gdyn}.

Throughout this section we consider the following query $Q$, which is free-connex acyclic (see Example~\ref{ex:connexity}):
\[  \pi_{y,z,w,u}\big(r(x,y)\Join s(y,z,w)\Join t(u,v) \mid x<z \wedge w<u\big).\]
Let $T_2$ be the \gjt from Fig.~\ref{fig:extrees}.
We process $Q$ based on a $T_2$\emph{-reduct}, a data structure that
succinctly represents the output of $Q$. For every node $n$, define $\predicates(n)$ as the set of all predicates on outgoing edges of $n$, i.e.
$\predicates(n)=\bigcup_{c \text{ child of } n} \predicates(n \to c).$

\begin{definition}[$T$-reduct]
  Let $T$ be a \gjt for a query $Q$ and let $\db$ be a database over $\atoms(Q)$. The $T$-reduct (or \emph{semi-join reduction}) of $\db$ is a collection $\reduc$ of
  GMRs, one GMR $\reduc_n$ for each node $n \in T$, defined inductively as follows:
  \begin{compactitem}[-]
  \item if $n = r(x)$ is an atom, then $\reduc_n = \db_{r(x)}$
  \item if $n$ has a single child $c$, then
      $\reduc_n = \pi_{\var(n)}\sigma_{\predicates(n)}\reduc_c$
  \item otherwise, $n$ has two children $c_1$ and
      $c_2$. In this case we have $\reduc_n = \pi_{\var(n)}\left(\reduc_{c_1} \Join_{\predicates(n)}
    \reduc_{c_2}\right)$.
  \end{compactitem}
\end{definition}

\begin{figure}[t]
  \centering
  \addtolength{\tabcolsep}{-2pt} 
  \begin{tikzpicture}[-,>=stealth',thick, level distance=11mm]
    \tikzstyle{detailnode}=[anchor=north,inner sep=2pt,text width=2cm, text centered]
    \node [detailnode] (yw) { %
      \scriptsize
      $\reduc_{\{y,w\}}$ \\
      \begin{tabular}{|l|r|}\hline
        \rowcolor{brown!60}
        $y$ $w$ & \# \\ \hline
        1 3 & 84 \\ \hline
      \end{tabular} %
    }
    [child anchor=north, sibling distance=4cm]
    child {
      node [detailnode] (yzw) { %
        \scriptsize
        $\reduc_{\{y,z,w\}}$ \\
        \begin{tabular}{|l|r|}\hline
          \rowcolor{brown!60}
          $y$ $z$ $w$ & \# \\ \hline
          1 3 3 & 12 \\ \hline
          2 4 6 & 15 \\ \hline
        \end{tabular} 
      }
      [sibling distance=2cm]
      child {
        node [detailnode] (r) { %
          \scriptsize
          $\reduc_r$ \\
          \begin{tabular}{|l|l|}\hline
            \rowcolor{brown!60}
            $x$ $y$ & \# \\ \hline
            2 2 & 2 \\ \hline
            3 2 & 3 \\ \hline
            2 1 & 4 \\ \hline
          \end{tabular} %
        }
        edge from parent node[left] {\scriptsize $x<z$}
      }
      child {
        node [detailnode] (s) { %
          \scriptsize
          $\reduc_{s}$\\
          \begin{tabular}{|l|l|}\hline
            \rowcolor{brown!60}
            $y$ $z$ $w$ & \# \\ \hline
            1 2 2 & 2 \\ \hline
            1 3 3 & 3 \\ \hline
            2 4 6 & 3 \\ \hline
          \end{tabular} %
        }
      }
    }
    child {
      node [detailnode] (u){ %
        \scriptsize
        $\reduc_{\{u\}}$\\
        \begin{tabular}{|l|l|}\hline
          \rowcolor{brown!60}
          $u$ & \# \\ \hline
          4 & 7 \\ \hline
          2 & 4 \\ \hline
        \end{tabular} %
      } 
      child {
        node [detailnode] (t){ %
          \scriptsize
          $\reduc_{t}$\\
          \begin{tabular}{|l|l|} \hline
            \rowcolor{brown!60}
            $u$ $v$ & \# \\ \hline
            2 3 & 4 \\ \hline
            4 6 & 2 \\ \hline
            4 5 & 5 \\ \hline
          \end{tabular} %
        }
      }
      edge from parent node[right] {\ \scriptsize $w<u$}
    };
    \node (ywlab) at ([shift={(2.2,-0.15)}]yw.north) {\scriptsize $=\pi_{y,w}(\reduc_{\{y,z,w\}} \Join_{w<u} \reduc_{\{u\}})$};
    \node (yzwlab) at ([shift={(1.9,-0.15)}]yzw.north) {\scriptsize $=\pi_{y,z,w}(\reduc_r \Join_{x<z} \reduc_s)$};
    \node (ulab) at ([shift={(0.7,-0.15)}]u.north) {\scriptsize $=\pi_{u}\reduc_t$};
    
    \node (db) at (0,1) { \scriptsize   \begin{tabular}[t]{|l|l|}
        \multicolumn{2}{c}{$\db_{t(u,v)}$} \\ \hline
        $u$  $v$ & $\mathbb{Z}$ \\ \hline
        2  3 & 4 \\
        4  6 & 2 \\
        4  5 & 5 \\\hline
    \end{tabular}
    \qquad
    	\begin{tabular}[t]{|l|l|}
    	\multicolumn{2}{c}{$\db_{s(y,z,w)}$} \\ \hline
    	$y$ $z$ $w$ & $\mathbb{Z}$ \\ \hline
    	1  2  2 & 2\\
    	1  3  3 & 3\\
    	2  4  6 & 3\\\hline
	\end{tabular} 
        \qquad 
	\begin{tabular}[t]{|l|l|}
    	\multicolumn{2}{c}{$\db_{r(x,y)}$} \\ \hline
    	$x$ $y$ & $\mathbb{Z}$ \\ \hline
    	2  2 & 2 \\
    	3  2 & 3 \\ 
    	2  1 & 4 \\\hline
	\end{tabular} 
};
  \end{tikzpicture}
  \caption{Example database and its $T_2$-reduct.}
  \label{fig:ex-reduct}
\end{figure}
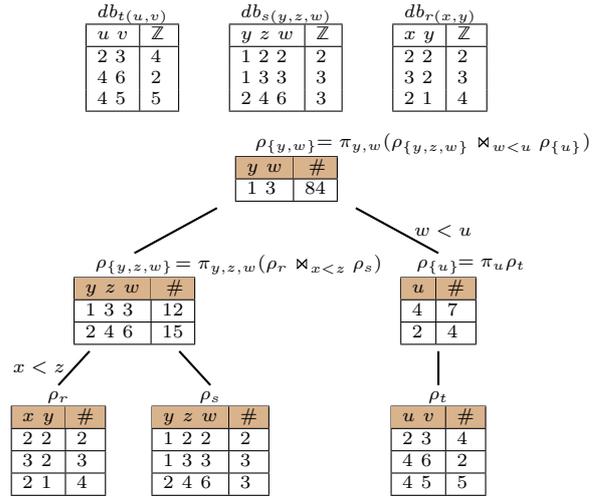

Fig.~\ref{fig:ex-reduct} depicts an
example database (top) and its $T_2$-reduct $\rho$ (bottom). Note, for
example, that the only tuple in the GMR at the root $\rho_{\{y,w\}}$ is
the join of $\rho_{\{y,z,w\}}$ and $\rho_{\{u\}}$ restricted to $w<y$
and projected over $\{y,w\}$. 

It is important to observe that the size of a $T$-reduct of a database
$\db$ can be at most linear in the size of $\db$. The reason is that, as illustrated in Fig.~\ref{fig:ex-reduct}, for each node $n$
there is some descendant atom $\alpha$ (possibly $n$ itself) such that
$\supp(\rho_n) \subseteq \supp(\proj_{\var(n)}\db_{\alpha})$. Note
that $Q(\db)$, in contrast, can easily become polynomial in the size of $\db$ in
the worst case.

\paragraph*{\bf Enumeration} 
From a $T$-reduct we can enumerate the result
$Q(\db)$ rather naively simply by recomputing the query results, in
particular because we have access to the complete database in the
leaves of $T$. We would like, however, to make the enumeration as
efficient as possible. To this end, we equip $T$-reducts with a set of
indices. To avoid the space cost of materialization, we do not
want the indices to use more space than the $T$-reduct itself (i.e.,
linear in $\db$). We illustrate these ideas in our running example by introducing a simple set of indices that allow for efficient enumeration.

Let $N = \{ \{y,w\}, \{y,z,w\}, \{u\} \}$ be the
connex subset of $T_2$ satisfying $\var(N) = \free(Q) = \{
y,z,w,u\}$. $(T_2,N)$ is compatible with $Q$, binary and
sibling-closed. We rely on the sibling-closed property of $N$ to
enumerate query results, and can do so without loss of generality by Proposition~\ref{prop:bin-sc-equivalent}. To enumerate the query results, we will traverse top-down the nodes in $N$. The
traversal works as follows: for each tuple $\tup{t_1}$ in $\rho_{\{y,w\}}$,
we consider all tuples $\tup{t_2}$ in $\rho_{\{y,z,w\}}$ that are compatible
with $\tup{t_1}$, and all tuples $\tup{t_3}\in\rho_{\{u\}}$ that are compatible
with $\tup{t_1}$. Compatibility here means that the corresponding equalities
and inequalities are satisfied. Then, for each pair $(\tup{t_2}, \tup{t_3})$, we
output the tuple $\tup{t_2}\cup \tup{t_3}$ with multiplicity
$\rho_{\{y,z,w\}}(\tup{t_2})\times\rho_{\{u\}}(\tup{t_3})$. A crucial difference here
with naive recomputation is that, since $\rho_{\{y,w\}}$ is already a
join between $\rho_{\{y,z,w\}}$ and $\rho_{\{u\}}$, we will only
iterate over \emph{relevant} tuples: each tuple that we iterate over
will produce a new output tuple. For example, we will never look at
the tuple $\langle y:2,z:4,w:6\rangle$ in $\rho_{\{y,z,w\}}$ because
it does not have a compatible tuple at the root.

To implement this enumeration strategy efficiently, we desire
index structures on $\reduc_{\{y,z,w\}}$ and $\reduc_{\{u\}}$ that
allow to enumerate, for a given tuple $\tup{t_1}$ in $\rho_{\{y,w\}}$, all
compatible tuples $\tup{t_2} \in \rho_{\{y,z,w\}}$ (resp. $\tup{t_3} \in
\rho_{\{u\}}$) with constant delay. In the case of $\rho_{\{u\}}$ this
is achieved simply by keeping $\rho_{\{u\}}$ sorted decreasingly on
variable $u$. Given tuple $\tup{t_1}$, we can enumerate the compatible
tuples from $\reduc_{\{u\}}$ by iterating over its tuples one by one in a decreasing manner, starting from the
largest value of $u$, and stopping whenever the current $u$ value is
smaller or equal than the $w$ value in $\tup{t_1}$. 
For
indexing $\rho_{\{y,z,w\}}$ we use a more standard index. Since we need to enumerate all tuples that have the same $y$ and
$w$ value as $\tup{t_1}$, CDE can be achieved by using a hash-based
index on $y$ and $w$. This index is depicted as $I_{\rho_{\{y,z,w\}}}$
in Fig.~\ref{fig:ex-trep}. We can see that, since the
described indices provide \cde of the compatible tuples given $\tup{t_1}$,
our strategy provides enumeration of
$Q(\db)$ with constant delay if we assume the query to be fixed (i.e. in data complexity~\cite{Vardi:1982}). 

\paragraph*{\bf Updates}
Next we illustrate how to process updates. The objective here is to
transform the $T_2$-reduct of $\db$ into a $T_2$-reduct of $\db+u$,
where $u$ is the received update. To do this efficiently we use additional indexes on $\rho$. We present the intuitions
behind these indices with an update consisting of two
insertions: $\langle y\colon 2, z\colon 3, w\colon 6\rangle$ with multiplicity $2$ and
$\langle u\colon 4, v\colon 9\rangle$ with multiplicity
$3$. Fig.~\ref{fig:ex-trep} depicts the update
process highlighting the modifications caused by the update.

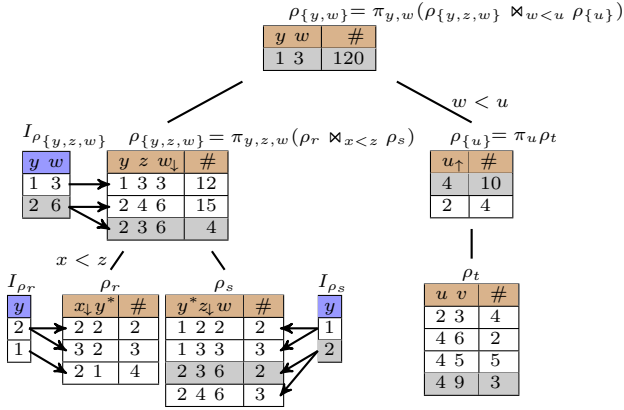
\begin{figure}[t!]
	\centering
	\addtolength{\tabcolsep}{-2pt} 
	\begin{tikzpicture}[-,>=stealth',thick, level distance=11mm]

		\tikzstyle{detailnode}=[anchor=north,inner sep=3.6pt,text width=2.5cm, text centered]

		\node [detailnode] (yw) { %
		\scriptsize
		$\reduc_{\{y,w\}}$ \\
		\begin{tabular}{|l|r|}\hline
		  \rowcolor{brown!60}
		  $y$ $w$ & \# \\ \hline
		  \rowcolor{gray!40}
		  1 3 & 120 \\ \hline
		\end{tabular} %
		}
		[child anchor=north, sibling distance=4cm]
		child {
		node [detailnode] (yzw) { %
		  \scriptsize
		  $\reduc_{\{y,z,w\}}$ \\
		    \begin{tabular}{|l|r|}\hline
		      \rowcolor{brown!60}
		      $y$ $z$ $w_{\hspace{-.04cm}\downarrow}$\hspace{-.1cm} & \# \\ \hline
		      1 3 3 & 12 \\ \hline
		      2 4 6 & 15 \\ \hline
		      \rowcolor{gray!40}
		      2 3 6 & 4 \\ \hline
		    \end{tabular} 
		}
		[sibling distance=1.5cm]
		child {
		  node [detailnode] (r) { %
		    \scriptsize
		    $\reduc_r$ \\
		    \begin{tabular}{|l|l|}\hline
		      \rowcolor{brown!60}
		      $x_{\hspace{-.04cm}\downarrow}$\hspace{-.1cm} $y^*$\hspace{-.13cm}\text{}& \# \\ \hline
		      2 2 & 2 \\ \hline
		      3 2 & 3 \\ \hline
		      2 1 & 4 \\ \hline
		      \end{tabular} %
		    }
		    edge from parent node[left] {\scriptsize $x<z$}
		}
		child {
		  node [detailnode] (s) { %
		    \scriptsize
		    $\reduc_{s}$\\
		    \begin{tabular}{|l|l|}\hline
		        \rowcolor{brown!60}
		        $y^*$\hspace{-.13cm} $z_{\hspace{-.04cm}\downarrow}$\hspace{-.1cm} $w$ & \# \\ \hline
		        1 2 2 & 2 \\ \hline
		        1 3 3 & 3 \\ \hline
		        \rowcolor{gray!40}
		        2 3 6 & 2 \\ \hline
		        2 4 6 & 3 \\ \hline
		      \end{tabular} %
		    }
		}
		}
		child {
		node [detailnode] (u){ %
		  \scriptsize
		  $\reduc_{\{u\}}$\\
		  \begin{tabular}{|l|l|}\hline
		    \rowcolor{brown!60}
		    $u_\uparrow$\hspace{-.1cm} & \# \\ \hline
		    \rowcolor{gray!40}
		    4 & 10 \\ \hline
		    2 & 4 \\ \hline
		  \end{tabular} %
		} 
		child {
		  node [detailnode] (t){ %
		    \scriptsize
		    $\reduc_t$\\
		    \begin{tabular}{|l|l|} \hline
		      \rowcolor{brown!60}
		      $u$ $v$ & $\#$ \\ \hline
		      2 3 & 4 \\ \hline
		      4 6 & 2 \\ \hline
		      4 5 & 5 \\ \hline
		      \rowcolor{gray!40}
		      4 9 & 3 \\ \hline
		    \end{tabular} %
		  }
		}
		edge from parent node[right] {\ \scriptsize $w<u$}
		};

		\addtolength{\tabcolsep}{-2pt}

		\node [anchor=south east, inner sep=0pt, text width=1cm]
		(In1) at (-2.9,-2.95) {%
			\scriptsize
			$I_{\rho_{\{y,z,w\}}}$\text{}\\
			\begin{tabular}{|l|}\hline
			    \rowcolor{blue!40}
			    $y$ $w$ \\ \hline
			    1\ \,3  \\ \hline
			    \rowcolor{gray!40}
			    2\ \,6  \\ \hline
			\end{tabular}
		};

		\node [anchor=south east, text width=1cm, inner sep=0pt] 
		(Is) at (0.99,-4.86) {%
			\scriptsize
			$I_{\rho_s}$\text{}\\
			\begin{tabular}{|l|}\hline
			    \rowcolor{blue!40}
			    $y$  \\ \hline
			    1  \\ \hline
			    \rowcolor{gray!40}
			    2  \\ \hline
			\end{tabular}
		};

		\node [anchor=south east, text width=1cm, inner sep=0pt] 
		(Ir) at (-3.1,-4.86) {%
		\scriptsize
		$I_{\rho_r}$\text{}\\
		  \begin{tabular}{|l|}\hline
		    \rowcolor{blue!40}
		    $y$  \\ \hline
		    2  \\ \hline
		    1  \\ \hline
		  \end{tabular}
		};

		\draw[->] (Ir) ++(-0.21cm,-.13) -- +(0.5,0);
		\draw[->] (Ir) ++(-0.21cm,-.43) -- +(0.5,-0.3);
		\draw[->] (Ir) ++(-0.21cm,-.13) -- +(0.5,-0.3);

		\draw[->] (In1) ++(.11cm,-.17) -- +(0.55,0);
		\draw[->] (In1) ++(.11cm,-.47cm) -- +(0.55,0);
		\draw[->] (In1) ++(.11cm,-.47cm) -- +(0.55,-.3);

		\draw[<-] (Is) ++(-1cm,-.13) -- +(0.5,0);
		\draw[<-] (Is) ++(-1cm,-.43) -- +(0.5,.3);
		\draw[<-] (Is) ++(-1cm,-.73) -- +(0.5,.3);
		\draw[<-] (Is) ++(-1cm,-1.03) -- +(0.5,.6);

    \node (ywlab) at ([shift={(2.2,-0.2)}]yw.north) {\scriptsize $=\pi_{y,w}(\reduc_{\{y,z,w\}} \Join_{w<u} \reduc_{\{u\}})$};
    \node (yzwlab) at ([shift={(1.9,-0.2)}]yzw.north) {\scriptsize $=\pi_{y,z,w}(\reduc_r \Join_{x<z} \reduc_s)$};
    \node (ulab) at ([shift={(0.7,-0.2)}]u.north) {\scriptsize $=\pi_{u}\reduc_t$};

	\end{tikzpicture}
	\caption{\dtree{$T_2$} of $\db+u$ with $T_1$ and $\db$ as in
          Fig.~\protect \ref{fig:ex-reduct}
          and $u$ the update
          containing  $(\langle y\colon 2, z\colon 5,
          w\colon 6\rangle,2)$ and
          $(\langle u\colon 4, v\colon 9\rangle,3)$. }
	\label{fig:ex-trep}
\end{figure}

Let us first discuss how to process the tuple $\tup{t_1} = \langle y\colon 2,
z\colon 3, w\colon 6\rangle$. We proceed bottom-up, starting at
$\reduc_{s}$ which is itself affected by the insertion of
$\tup{t_1}$. Subsequently, we need to \emph{propagate} the modification of
$\reduc_{s}$ to its ancestors $\reduc_{\{y,z,w\}}$ and
$\reduc_{\{y,w\}}$. Concretely, from the definition of $T$-reduction,
it follows that we need to add some modifications to
$\reduc_s$, $\reduc_{\{y,z,w\}}$, and $\reduc_{\{y,w\}}$ on $\tup{t_1}$:
\begin{compactitem}[]
 \item $\Delta \reduc_s = [\tup{t_1} \mapsto 2]$,
 \item $\Delta \reduc_{\{y,z,w\}} = \proj_{y,z,w}\left( \reduc_r \Join_{x<z}
    \Delta \reduc_s\right)$,
 \item $\Delta \reduc_{\{y,w\}} = \proj_{y,w} \left( \Delta
    \reduc_{\{y,z,w\}} \Join_{w < u} \reduc_{\{u\}} \right)$.
\end{compactitem}
To compute the joins on the right-hand sides efficiently, we create a
number of additional indexes on $\reduc_r, \reduc_s$, and
$\reduc_{\{y,z,w\}}$. Concretely, in order to efficiently compute
$\proj_{y,z,w}\left( \reduc_r \Join_{x<z} \Delta \reduc_s\right)$,
we \emph{group} tuples in the GMR $\rho_r$ by the
variables that $\reduc_r$ has in common with $\reduc_s$ (in this case
$y$) and then, per group, \emph{sort} tuples ascending on variable
$x$. We mark grouping variables in Fig.~\ref{fig:ex-trep} with $*$
(e.g. $y^*$), and sorting by $\downarrow$ (for ascending, e.g.,
$x_\downarrow$) and $\uparrow$ (for descending). A hash index on the
grouping variables (denoted $I_{\reduc_{r}}$ in
Fig.~\ref{fig:ex-trep}) then allows to find the group given a $y$
value. The join can then be processed by means of a hybrid form of
sort-merge and index nested loop join. Sort $\Delta \reduc_s$
ascendingly on $y$ and $z$. For each $y$-group in $\Delta \reduc_s$
find the corresponding group in $\reduc_r$ by passing the $y$ value to
the index $I_{\rho_r}$. Let $\tup{t'}$ be the first tuple in the $\Delta
\reduc_s$ group. Then iterate over the tuples of the $\reduc_r$ group
in the given order and sum up their multiplicities until $x$ becomes
larger than $\tup{t'}(z)$. Add $\tup{t'}$ to the result with its original
multiplicity multiplied by the found sum (provided it is
non-zero). Then consider the next tuple in the $\Delta \reduc_s$
group, and continue summing from the current tuple in the $\reduc_r$
group until $x$ becomes again larger than $z$, and add the result
tuple with the correct multiplicity. Continue repeating this process
for each tuple in the $\Delta \reduc_s$ group, and for each group in
$\Delta \reduc_s$.
In our case, there is only one group in $\Delta \reduc_s$ (given by $y=2$) and we will
only iterate over the tuple $\langle x\colon 2, y\colon 2\rangle$ in $\reduc_r$,
obtaining a total multiplicity of 2, and therefore compute $\Delta
\reduc_{\{y,z,w\}} = [ \tup{t_1} \to 4]$. In order to compute the join
$\proj_{y,w} \left( \Delta \reduc_{\{y,z,w\}} \Join_{w < u}
  \reduc_{\{u\}} \right)$ efficiently, we proceed similarly. Here,
however, there are no grouping variables on $\reduc_{\{u\}}$ and it
hence suffices to sort $\reduc_{\{u\}}$ descendingly on $u$. Note that
this was actually already required for efficient enumeration. Also
note that $\Delta \reduc_{\{y,w\}}$ is empty.

Now we discuss how to process $\tup{t_2} = \langle u\colon 4,
v\colon 9\rangle$. First, we insert $\tup{t_2}$ into $\rho_t$. We need to
propagate this change to the parent $\reduc_{\{u\}}$ by calculating
$\Delta \reduc_{\{u\}} = \proj_u \Delta \reduc_{t}$. This is done
by a simple hash-based aggregation. 
Finally, we need to propagate $\Delta \reduc_{\{u\}}$ to the root by
computing $\Delta \reduc_{\{y,w\}} =
\pi_{y,w}(\reduc_{\{y,z,w\}} \Join_{w<u} \Delta \reduc_{\{u\}})$. To process this join efficiently we proceed as before. Again, there are no grouping variable on $\reduc_{\{y,z,w\}}$ (since
it has no variables in common with $\reduc_{\{u\}}$) and it hence
suffices to sort $\reduc_{\{y,z,w,\}}$ ascending on $w$. The only tuple that we iterate over during the hybrid join is $\langle
y\colon 1, z\colon 3, w\colon 3\rangle$ wich has multiplicity 12. Hence, we have $\Delta
\reduc_{y,w} = [\langle y\colon 1, w\colon 3\rangle \mapsto 36 \ ]$, concluding
the example.


\section{Dynamic Yannakakis Over GCQs}
\label{sec:gdyn}

Dynamic Yannakakis (\dyn) is an algorithm to efficiently evaluate free-connex acyclic
aggregate-equijoin queries under updates~\cite{dyn:2017}. This algorithm matches two
important theoretical lower bounds (for q-hierarchical \cqs
\cite{Berkholz:2017} and free-connex acyclic \cqs \cite{Bagan:2007}),
and is highly efficient in practice. In this section we
present a generalization of \dyn, called \edyn, to dynamically
process free-connex acyclic \gcqs. Since predicates in a \gcq can be arbitrary, our
approach is purely algorithmic; the  efficiency by which \edyn process
updates and produces results will depend entirely on the efficiency
of the underlying data structures. Here we only describe the properties
that those data structures should satisfy and present the general (worst-case)
complexity of the algorithm.
The techniques and indices presented in the previous section provide a practical instantiation of \edyn to a \gcq with equalities and inequalities, and throughout this section we make a parallel between that instantiation and the more abstract definitions of \edyn.

In this section we assume that $Q$ is a free-connex acyclic
\gcq and that $(T,N)$ is a binary and sibling-closed \gjt pair
compatible with $Q$. Like in the case of equalities and inequalities,
the dynamic processing of $Q$ will be based on a $T$-reduct of
the current database $\db$. A set of indices will be added to
optimize the enumeration of query results and maintenance of the
$T$-reduct under updates. We formalize the notion of index as follows:

\begin{definition}[Index]
  Let $R$ be a GMR over $\seq{x}$, let $\seq{y}$ be a hyperedge, let
  $\seq{w}$ be a hyperedge  satisfying
  $\seq{w}\subseteq\seq{x} \cup \seq{y}$, and let $\theta(\seq{z})$ be a predicate with $\seq{z}
  \subseteq \seq{x}\cup\seq{y}$. An \emph{index} on $R$ by $(\theta,\seq{y},\seq{w})$ with delay $f$ is a data structure $I$ that provides, for any given GMR $R_{\seq{y}}$ over $\seq{y}$, enumeration of $\proj_{\seq{w}}(R \bowtie_\theta R_{\seq{y}})$ with delay
  $O(f(\card{R}+\card{R_{\seq{y}}}))$. The update time of index $I$ is the time required
  to update $I$ to an index on $R + \Delta{R}$ (by $(\theta,\seq{y},\seq{w})$) given update
  $\Delta{R}$ to $R$.
\end{definition}

For example, $I_{\rho_{r}}$ in Fig.~\ref{fig:ex-trep} is used as an index on $\rho_{r}$ by $(x<z, \{y,z,w\}, \{y,z,w\})$. Indeed, in the previous section we precisely discussed how $I_{\reduc_r}$ allows to efficiently compute $\pi_{y,z,w}(\rho_{r}\Join_{x<z}\Delta{\rho_{s}})$ for an update $\Delta{\rho_{s}}$ to $\rho_{s}$. Having the notion of index, we discuss how \edyn enumerates query results and processes updates.

\paragraph*{\bf Enumeration} Let $\db$ be the current database. To enumerate $Q(\db)$
from a $T$-reduct $\reduc$ of $\db$ we can iterate over the reductions
$\rho_n$ with $n \in N$ in a nested fashion, starting at the root and
proceeding top-down.  When $n$ is the root, we iterate over all tuples
in $\rho_n$. For every such tuple $\tup{t}$, we iterate only over the
tuples in the children $c$ of $n$ that are compatible with $\tup{t}$
(i.e., tuples in $\rho_c$ that join with $\tup{t}$ and satisfy
$\predicates(n \to c)$).  This procedure continues until we reach
nodes in the frontier of $N$ at which time the output tuple can be
constructed. The pseudocode is given in
Algorithm~\ref{alg:enumeration-st}, where the tuples that are
compatible with $\tup{t}$ are computed by $\rho_c
\semijoin_{\predicates(n\to c)} \tup{t}$.

\begin{algorithm}[tb]
  \caption {Enumerate $Q(\db)$ given $T$-reduct $\reduc$ of $\db$.}
  \label{alg:enumeration-st}
  \begin{algorithmic}[1]
    \Function{$\routenum_{T,N}$}{$\reduc$}
    \LineFor{\textbf{each} $\tup{t}\in \reduc_{\treeroot(T)}$}
        {$\routenum_{T,N}(root(T), \tup{t}, \reduc)$}\label{line:rootloop}
    \EndFunction
    \medskip
    \Function{$\routenum_{T,N}$}{$n, \tup{t},\reduc$}\label{line:subenum}
    \If{$n$ is in the frontier of $N$}{\textbf{ yield} $(\tup{t}, \reduc_n(\tup{t}))$}\label{line:frontiertuple}
    \ElsIf{$n$ has one child $c$}
      \LineFor{\textbf{each} $\tup{s}\in \reduc_c \semijoin_{\predicates(n
          \to c)} \tup{t}$}{$\routenum_{T,N}(c,\tup{s},\reduc)$}
    \Else{ $n$ has two children $c_1$ and $c_2$} \label{line:twochildren-st}
      \For{\textbf{each} $\tup{t_1}\in \rho_{c_1} \semijoin_{\predicates(n
          \to c_1)} \tup{t}$}
        \For{\textbf{each} $\tup{t_2}\in \rho_{c_2} \semijoin_{\predicates(n \to c_2)} \tup{t}$}
          \For{\textbf{each} ($\tup{s_1}, \mu)\in \routenum_{T,N}(c_1,\tup{t_1},\reduc)$}
            \For{\textbf{each} ($\tup{s_2}, \nu) \in \routenum_{T,N}(c_2, \tup{t_2},\reduc)$}
              \State \textbf{yield} ($\tup{s_1}\cup \tup{s_2}$, $\mu\times\nu$)
            \EndFor
          \EndFor
        \EndFor
      \EndFor
   \EndIf
    \EndFunction
  \end{algorithmic}
\end{algorithm}

Now we show the correctness of the enumeration algorithm, for which we need to introduce some further notation. Let $Q$, $T$ and $N$ be as above. Given a node $n\in T$ we denote the sub-tree of $T$ rooted at $n$ by $T_n$, and define the query induced by $T_n$ as
$$Q_{n}=(\Join_{r(\seq{x}) \in\atoms(T_n)}r(\seq{x}) \mid\pred(T_n))$$
where $\atoms(T_n)$ and $\pred(T_n)$ are the sets of all atoms and
predicates occurring in $T_n$, respectively.

\begin{restatable}{lemma}{subtreeQuery}
\label{lem:subtree-query}
    Let $Q$, $T$, $N$, and $n$ be defined as above, and let $\rho$ be a $T$-reduct for $Q$. Then, $\rho_n=\pi_{\var(n)}Q_n(\db)$.
\end{restatable}
The proof by induction is detailed in Appendix~\ref{sec:app-algorithm}.
To show correctness of enumeration, we need the following additional
lemma regarding the subroutine of Algorithm~\ref{alg:enumeration-st}
(Line~\ref{line:subenum}). The proof is again by induction and
detailed in Appendix~\ref{sec:app-algorithm}.

\begin{restatable}{lemma}{enumSubRoutine}
  \label{lem:enum_subroutine}
    Let $Q$, $T$, and $N$ be as above. If $\rho$ is a $T$-reduct of
    $\db$, then for every node $n\in N$ and every tuple $\tup{t}$ in
    $\rho_{n}$, $\routenum_{T,N}(n, \tup{t},\rho)$ correctly
    enumerates $\pi_{\var(N) \cap \var(Q_n)}Q_n(\db)\semijoin \tup{t}$.
\end{restatable}

\begin{proposition}
  Let $Q$, $T$, $N$ and $\rho$ be as above. Then $\routenum_{T,N}(\rho)$ enumerates $Q(\db)$.
\end{proposition}

\begin{proof}
    Let $r$ be the root of $T$. By Lemma~\ref{lem:subtree-query} we have $\rho_n=\pi_{\var(r)}Q_r(\db)=\pi_{\var(r)}Q(\db)$, and therefore $\rho_n$ is a projection of $Q(\db)$. This implies that $Q(\db)=Q(\db)\semijoin\rho_r$, which is equivalent to the disjoint union $\bigcup_{\tup{t}\in\rho_r}Q(\db)\semijoin\tup{t}$. By Lemma~\ref{lem:enum_subroutine}, it is clear that this is exactly what $\routenum_{T,N}(\rho)$ enumerates.
\end{proof}

We now analyze the complexity of $\routenum_{T,N}$. First, observe
that by definition of $T$-reducts, compatible tuples will
exist at every node. Hence, every tuple that we iterate over will
eventually produce a new output tuple.  This ensures that we do not
risk wasting time in iterating over tuples that in the end yield no
output. As such, the time needed for $\routenum_{T,N}(\rho)$ to
produce a single new tuple is determined by the time taken to
enumerate the tuples in $\rho_n \semijoin_{\predicates(p \to n)}\tup{t}$, where $p$ is the parent of $n$. Since this is equivalent to $\proj_{\var(n)}(\rho_n \Join_{\predicates(p \to n)}\tup{t})$ we can do this efficiently by creating an index on $\rho_n$ by $(\predicates(p\rightarrow n),\var(p),\var(n))$.
For example, in Section~\ref{sec:iedyn} we defined hash-maps and group-sorted GMRs so that given one tuple from a parent we could enumerate the compatible tuples in the child with constant delay. In general, the efficiency of enumeration will depend on the delay provided by the indices.

\begin{proposition}\label{prop:enum-complexity}
  Assume that for every $n \in N$ we have an index on $\rho_n$ by $(\predicates(p \to n), \var(p), \var(n))$ with delay $f$, where $p$ is the parent of $n$ and $f$ is a monotone function. Then, using these indices, $\routenum_{T,N}(\rho)$ correctly enumerates $Q(\db)$ with delay $O(|N| \times f(M))$ where $M$ is given by $\max_{n\in N}(\card{\rho_n})$. Thus, the total time required to execute $\routenum_{T,N}(\rho)$ is $O(\card{Q(\db)}\cdot f(M)\cdot \card{N})$.
\end{proposition}
\begin{proof}
    We show that for every $n\in N$ and $\tup{t}\in\rho_n$, the call $\routenum_{T,N}(n,\tup{t},\rho)$ enumerates $\pi_{\var(N)}Q_n(\db)$ with delay $O(|N\cap T_n|\times f(M))$. We proceed by induction in $|N|$. If $|N|=1$ then $N=\treeroot(T)$ and the delay is clearly constant as the algorithm will only yield $\tup{t}$. Now assume that $|N|>1$. If $n$ has a single child $c$, the index on $\rho_c$ by $(\pred(n), \var(n), \var(c))$ allows us to iterate over $\rho_c\semijoin_{\pred(n)} \tup{t}$ with delay $O(f(|\rho_c|))$ and therefore delay $O(f(M))$. For each element $\tup{s}$ of this enumeration, the algorithm calls $\routenum_{T, N}(c, \tup{s}, \rho)$, which by induction hypothesis enumerates $\pi_{\var(N)}Q_c(\db)\semijoin \tup{s}$ with delay $O(|N\cap T_c|\times f(M))$. Then, the maximum delay between two outputs is $O(f(|\rho_c|))+O(|N\cap T_c|\times f(M))$, and since $|\rho_c|\leq M$ this is in 
    $$O\left((|N\cap T_c| + 1)\times f(M)\right)=O\left(|N\cap T_n|\times f(M)\right).$$
    The final observation is that the sets $\pi_{\var(N)}Q_c(\db)\semijoin\tup{s}$ are disjoint for different values of $\tup{s}$, and thus the enumeration does not produce repeated values.

    For the case in which $n$ has two children $c_1$ and $c_2$, by similar reasoning it is easy to show that the maximum delay between two outputs is 
    \begin{multline*}
        O(f(|\rho_{c_1}|))+O(|N\cap T_{c_1}|\times f(M)) +\\
        O(f(|\rho_{c_2}|))+O(|N\cap T_{c_2}|\times f(M))
    \end{multline*}
    $$=O((|N\cap T_{c_1}| + |N\cap T_{c_2}| + 2)\times f(M))$$
    $$=O(|N\cap T_n|\times f(M)).$$
    It is also important to mention that the sets enumerated by $\routenum_{T, N}(c_i,\tup{t_i}, \rho)$ are disjoint for each $\tup{t_i}\in\rho_{c_i}\semijoin_{\pred(n\rightarrow c_i)}\tup{t}$ ($i\in\{1,2\}$), and that for each $(\tup{s_1},\mu)\in \routenum_{T, N}(c_1,\tup{t_1}, \rho)$  and $(\tup{s_2},\mu)\in \routenum_{T, N}(c_2,\tup{t_2}, \rho)$, it is the case that $\tup{s_1}$ and $\tup{s_2}$ are compatible, thus producing outputs in every iteration.
\end{proof}

In particular, if $f$ is constant we enumerate $\card{Q(\db)}$ with delay $O(N)$ (i.e. constant in data complexity).
\paragraph*{\bf Update processing} To allow enumeration of
$Q(\db)$ under updates to $\db$ we need to
maintain the $T$-reduct $\reduc$ (and, if present, its indexes)
up to date. As illustrated in the previous section, it suffices to traverse the nodes of
$T$ in a bottom-up fashion. At each node $n$ we have to compute the delta of $\rho_n$. For leaf nodes, this delta is given by the update $\upd$ itself. For interior nodes, the delta can be computed from the delta and original reduct of its children. Algorithm~\ref{alg:update-processing-st} gives the pseudocode.

The fundamental part of Algorithm~\ref{alg:update-processing-st} is to compute joins and
produce delta GMRs (Line~\ref{line:childrenjoinupdate}), propagating updates from each node to its parent. When there is an update
$\Delta_n$ to a node $n$ with sibling $m$ and parent $p$, we need to
compute $\proj_{\var(p)} \left( \reduc_m \Join_{\predicates(p)}
  \Delta_n \right)$. To do this efficiently, we naturally store an
index on $\rho_m$ by $(\predicates(p),\var(n),\var(p))$. For example,
we discussed how the hash-map $I_{\rho_r}$ in Fig.~\ref{fig:ex-trep}
plus the sorting on $x$ of $\rho_r$ allowed us to efficiently compute
$\proj_{y,z,w}(\rho_r\Join_{x<z}\Delta\rho_s)$.

\begin{algorithm}[tbp]
  \caption {Update($\reduc$, $\upd$)}
  \label{alg:update-processing-st}
  \begin{algorithmic}[1]
    \State \textbf{Input:} A $T$-reduct $\reduc$ for $\db$ and an update $\upd$.
    \State \textbf{Result:} Transforming $\reduc$ to a $T$-reduc for $\db+\upd$.

    \For{\textbf{each} $n\in\leafs(T)$ labeled by $r(\seq{x})$}
      \State $\Delta_n\leftarrow\upd_{r(\seq{x})}$
    \EndFor
  	\For{\textbf{each} $n\in\nodes(T)\setminus\leafs(T)$}
  		\State $\Delta_n\leftarrow$ empty GMR over $\var(n)$
    \EndFor

    \For{\textbf{each} $n\in\nodes(T)$, traversed bottom-up}\label{line:updateiteration}
    	\State $\reduc_n += \Delta_n$
      \If{$n$ has a parent $p$ and a sibling $m$}
        \State $\Delta_{p} += \proj_{\var(p)} \left( \reduc_m \Join_{\predicates(p)} \Delta_n \right)$\label{line:childrenjoinupdate}
      \ElsIf{$n$ has parent $p$}
        \State $\Delta_{p} += \proj_{\var(p)}\sigma_{\predicates(p)} \Delta_n$\label{line:singlechildupdate}
    	\EndIf
    \EndFor
  \end{algorithmic}
\end{algorithm}
\smallskip

Summarizing, to efficiently enumerate query results and process updates we need to store a $T$-reduct plus a set of indices on its GMRs. The data structure containing these elements is called a \emph{$(T,N)$-representation}.

\begin{definition}[$(T,N)$-representation]
  \label{def:trep}
  Let $\db$ be a database. A $(T,N)$-representation (\dtree{$(T,N)$} for short) of $\db$ is composed by a $T$-reduct of $\db$ and, for each node $n$
  with parent $p$, the following set of indices:
  \begin{compactitem}[-]
    \item If $n$ belongs to $N$, then we store an index $P_n$ on $\rho_n$ by $(\predicates(p\rightarrow n),\var(p),\var(n))$.
    \item If $n$ is a node with a sibling $m$, then we store an index $S_n$ on $\rho_n$ by $(\predicates(p), \var(m), \var(p))$.
  \end{compactitem}
\end{definition}
Together with the notion of \dtree{$(T,N)$}, Algorithms~\ref{alg:enumeration-st} and~\ref{alg:update-processing-st} provide a framework for dynamic query evaluation. By constructing the
$T$-reduct and set of indices (and their update procedures) one can process free-connex
acyclic \gcqs under updates. Naturally, to implement such framework
one needs to devise indices for a particular set of predicates. For
example, \dyn is an instantiation to the class of
\cqs, and in the previous section we showed how to
instantiate this framework for a \gcq based on equalities and
inequalities. Next, we present the general set of indices
required to process free-connex acyclic \gcqs with equalities and
inequalities.

\vspace{-1ex}
\paragraph*{\bf IEDyn}
For queries that have only inequality predicates, the instantiation of
a $(T,N)$-representation of $\db$ contains a $T$-reduct of $\db$ and, for
each node $n$ with parent $p$, the following data structures:
  \begin{compactitem}[-]
  \item If $n \in N$, the index $P_n$ on $\reduc_n$ from
    Definition~\ref{def:trep} is obtained by doing two things. (1)
    First, group $\reduc_n$ according to the variables in $\var(n)
    \cap \var(p)$. Then, per group, sort the tuples according to
    the variables of $\var(n)$ mentioned in $\predicates(p \to n)$ (if
    any).
    (2) Create a hash table that maps each tuple $\tup{t} \in
    \proj_{\var(n) \cap \var(p)}(\reduc_n)$ to its corresponding group
    in $\reduc_n$. If $\var(n) \cap \var(p)$ is empty this
    hash table is omitted.
  \item If $n$ has a sibling $m$, the index $S_n$ of
    Definition~\ref{def:trep} is obtained by doing two things. (1)
    First, group $\reduc_n$ according to the variables in $\var(n)
    \cap \var(m)$. Then, per group, sort the tuples according to
    the variables of $\var(n)$ mentioned in $\predicates(p)$ (if any).
    (2) Create a hash table mapping each $\tup{t} \in
    \proj_{\var(n) \cap \var(m)}(\reduc_n)$ to the corresponding group
    in $\tup{s} \in \reduc_n$. If $\var(n) \cap \var(m)$ is
    empty this hash table is omitted.
\end{compactitem}
In Wection~\ref{sec:iedyn} we illustrated how use these data structures. Effectively, in
Figure~\ref{fig:ex-trep} $I_{\rho_r}$ and $I_{\rho_s}$ are examples
of $S_n$, used for update propagation, while $I_{\reduc_{\{y,z,w\}}}$ is an example of $P_n$, used for enumeration.\par

Note that the example query from Section~\ref{sec:iedyn} has at most
one inequality between each pair of atoms. This causes each edge in
$T$ to consist of at most inequality. As such, when creating the index
$P_n$ for a node $n \in N$, the reduct $\reduc_n$ will be sorted per group
according to at most one variable. This is important for enumeration
delay because, as exemplified in Section~\ref{sec:iedyn}, we can then
find compatible tuples by first the corresponding group and then
iterating over the sorted group from the start and stopping when the
first non-compatible tuple is found. When there are multiple
inequalities per pair of atoms then we will need to sort according to
multiple variables under some lexicographic order. This causes
enumeration delay to become logarithmic since then compatible tuples
will intermingle with non-compatible tuples, and a binary search is
necessary to find the next batch of compatible tuples in the group.

We call \iedyn the algorithm for processing free-connex acyclic
\gcqs with equalities and inequalities.

\begin{theorem}\label{thm:iedyn-efficiency}
    Let $Q$ be a \gcq in which all predicates are equalities and
    inequalities. Let $(T,N)$ be a binary and sibling-closed \gjt pair
    compatible with $Q$. Given a database $\db$ over $\atoms(Q)$, a \dtree{$(T,N)$} $\drep$ of $\db$, under \iedyn Algorithm~\ref{alg:enumeration-st} enumerates $Q(\db)$ with delay $O(|N|\cdot\log(\card{\db}))$. Also, given an update $\upd$ under \iedyn Algorithm~\ref{alg:update-processing-st} transforms $\drep$ into a \dtree{$(T,N)$} of $\db+u$ in time $O(|T|\cdot M^2\cdot\log(M))$, where $M=\card{\db}+\card{\upd}$.
\end{theorem}
\begin{proof}

    Let us first prove the enumeration bounds. It is immediate to see that for every node $n\in T$ the GMR $\rho_n$ satisfies $|\rho_n|\leq |\db|$, given that $\rho_n$ is defined as a series of semi-joins based on $\db$ (or, equivalently, because every internal node has a guard). Therefore, according to Proposition~\ref{prop:enum-complexity} the enumeration delay is $O(|N|\cdot f(|\db|))$ where $N$ is the connex subset of $T$ and $f$ is the delay provided by the index $P_n$. Now, from the description of \iedyn these indices are implemented as hash tables that map each tuple $\tup{t}$ in $\pi_{\var(p)\cap\var(n)}\rho_n$ to a lexicographically sorted set containing $\rho_n\semijoin_{\pred(p\rightarrow n)} \tup{t}$, where $(p,n)$ is a parent-child pair. Therefore, given a tuple $\tup{t}\in\rho_p$ we can enumerate $\rho_n\semijoin_{\pred(p\rightarrow n)}\tup{t}$ by first projecting $\tup{t}$ over $\var(n)$ and then iterating over all tuples satisfying $\pred(p\rightarrow n)$. Since these predicates are only inequalities, each group can be kept sorted lexicographically and, as mentioned earlier, enumeration can be achieved with logarithmic delay. It follows from Prop.~\ref{prop:enum-complexity} that the enumeration delay is $O(|N|\cdot \log(|db|))$.

Now we discuss update time. As can be seen in
Algorithm~\ref{alg:update-processing-st}, for each parent-child pair
$(p,n)\in T$ we need to compute either
$\pi_{\var(p)}(\rho_m\Join_{\pred(p\rightarrow n)}\Delta_n)$ or
$\pi_{\var(p)}\sigma_{\pred(p)}(\Delta_n)$, depending on whether or
not $n$ has a sibling $m$. If $n$ does not have a sibling, computing
$\pi_{\var(p)}\sigma_{\pred(p)}(\Delta_n)$ can be done directly by
sorting $\Delta_n$ lexicographically, enumerating those tuples
satisfying $\pred(p)$ (with logarithmic delay), and finally projecting
over $\var(p)$. This takes time in
$O(|\Delta_n|\cdot\log(|\Delta_n|)$, which is clearly contained in
$O(M^2\cdot \log(M))$ since $|\Delta_n|\leq M$. The more involved case
is when $n$ has a sibling $m$ and we need to compute
$\pi_{\var(p)}(\rho_m\Join_{\pred(p)}\Delta_n)$. Here we first sort $\Delta_n$ lexicographically. Then, for every
tuple $\tup{t}$ in $\pi_{\var(p)}\rho_m$ compute
$\pi_{\var(p)}(\tup{t}\Join_{\pred(p)} \Delta_n)$. Note that this can
be done in time $O(|\Delta_n|\cdot\log(|\Delta_n|))$ since from the
constructed data structures we can enumerate
$\Delta_n\semijoin_{\pred(p)}\tup{t}$ with logarithmic delay. Because
the previous procedure needs to be performed for each
$\tup{t}\in\rho_n$, this can be done in time
$O(|\rho_n|\cdot|\Delta_n|\cdot\log(|\Delta_n|))$ and therefore in
time $O(M^2\cdot\log(M))$. Note that here we ignore the sorting steps
as well as the maintenance of the corresponding GMRs as those steps
are clearly $O(M\cdot\log(M))$. Finally, since we need to perform the
procedure described above once per each parent-child pair, the entire
routine takes at most $O(|T|\cdot M^2\cdot\log(M)$). \qedhere
\end{proof}

From the previous result we can see that for the general case of equalities and inequalities we already have a procedure that can be quadratic in the size of the database.\footnote{In the conference version of this paper~\cite{DBLP:journals/pvldb/IdrisUVVL18} there was an incorrect claim: we stated that updates could be processed in time $O(M\cdot\log(M))$ in data complexity. We then found a bug in our algorithm and we currently do not know if this bound can be achieved.} However, if we restrict the use of inequalities in a particular way, we can speed up both update processing and enumeration delay.

\begin{theorem}\label{thm:iedyn-single-efficiency}
    Let $Q$, $T$ and $N$ be defined as in Theorem~\ref{thm:iedyn-efficiency}, and assume that for each $p\in T$ it is the case that $|\pred(p)|\leq 1$. Given a database $\db$ over $\atoms(Q)$, a \dtree{$(T,N)$} $\drep$ of $\db$, under \iedyn Algorithm~\ref{alg:enumeration-st} enumerates $Q(\db)$ with delay $O(|N|)$. Also, given an update $\upd$ under \iedyn Algorithm~\ref{alg:update-processing-st} transforms $\drep$ into a \dtree{$(T,N)$} of $\db+u$ in time $O(|T|\cdot M\cdot\log(M))$, where $M=\card{\db}+\card{\upd}$.
\end{theorem}

\begin{proof}
    The main observation to prove this result is that when there is a single predicate, a lexicographically sorted set is totally sorted by a single attribute. Regarding enumeration, this implies that given a parent-child pair $(p,n)$ and a tuple $\tup{t}\in\pi_{\var(n)}\rho_p$, we can enumerate $\rho_n\semijoin_{\pred(P)} \tup{t}$ with constant delay. The reason behind this is that the index $P_n$ maps $\tup{t}$ to a totally sorted set, and therefore we can start from the largest/smallest value of the relevant attribute, and iterate over all tuples decreasingly/increasingly until we find a tuple that does not satisfy the inequality. At that point we are certain that we have visited all tuples satisfying the inequality.

    The update processing can also be improved by a similar argument, although the modification is slightly more involved. Assume again that we have a parent-child pair $(p,n)$ and want to compute $\pi_{\var(p)}(\rho_m\Join_{\pred(p)}\Delta_n)$, where $m$ is the sibling of $n$. We do so efficiently as follows. Recall that the index $S_m$ groups $\rho_m$ by $\var(n)\cap\var(m)$ and sorts each group by the variables involved in $\pred(p)$. We construct an index over $\Delta_n$ with the same characteristics, which is achieved by a vanilla implementation in $O(|\Delta_n|\cdot\log(|\Delta_n|))$. Again, since $\pred(p)$ contains at most a single inequality, each group will be sorted by a single variable and hence totally sorted.
    Assume now that $m$ is a guard of $p$. Since by definition $\rho_m\Join_{\pred(p)}\Delta_n=\sigma_{\pred(p)}(\rho_m\Join\Delta_n)$, to compute this join it is sufficient to find for each tuple $\tup{t}$ in $\rho_m$ the matching tuples in the corresponding group of $\Delta_n$. However, a naive implementation would take $O(M^2)$, since for such $\tup{t}$ we might iterate over a potentially linear set of tuples in $\Delta_m$. This can be avoided by considering the following two observations:
    \begin{enumerate}
    \item Given a tuple $\tup{t}$ in $\rho_m$, since $m$ is a guard of $p$ we only need to compute the multiplicity associated to $\tup{t}$ in $\sigma_{\pred(p)}(\pi_{\var(p)}(\rho_m\Join\Delta_n))$, which can be computed as $\rho_m(\tup{t})\cdot\sum_{\tup{s}\in\Delta_n\semijoin_{\pred(p)} \tup{t}} \Delta_n(\tup{s})$.
    \item Let $\tup{t_1}$ and $\tup{t}_2$ be two tuples belonging to the same group in $\rho_m$. Assume $\pred(p)=a<b$, with $a\in\var(n)$ and $b\in\var(m)$. Then, if $\tup{t}_1(a)<\tup{t}_2(a)$ we have that $\Delta_n\semijoin_{\pred(p)} \tup{t_2}$ is a subset of $\Delta_n\semijoin_{\pred(p)} \tup{t_1}$.        
    \end{enumerate}
    By these two facts, if we iterate in order over the tuples $\tup{t}$ of each group of $\rho_n$, and we iterate simultaneously in order over the tuples $\tup{s}$ in the group of $\Delta_n$ corresponding to $\tup{t}$ (which can be done with constant delay), we can compute the corresponding multiplicities incrementally, visiting each tuple in $\Delta_n$ only once. Therefore, this join can be computed in linear time in $M$ and the most expensive part of this procedure is to actually construct and maintain the sorted groups, an $O(M\cdot\log(M))$ procedure. It is easy to see that this can be generalized to any inequality, and that in the case in which $n$ is a guard of $p$ it suffices to swap the roles of $\rho_m$ and $\Delta_n$. We conclude that in this case \iedyn updates the corresponding $(T,N)$-representation in $O(M\cdot\log(M))$.
\end{proof}


\section{Computing GJTs}
\label{sec:gyo}
In this section, we discuss how to check acyclicity and free-connex
acyclicity for \gcqs, and give an algorithm to compute a compatible
\gjt pair for a given \gcq. 

The canonical algorithm for checking acyclicity of normal conjunctive
queries is the GYO algorithm~\cite{abiteboul1995foundations}. Our algorithm is
a generalisation of the GYO algorithm that checks free-connex
acyclicity in addition to normal acyclicity and  deals with \gcqs
featuring $\theta$-join predicates instead of \cqs that have equality
joins only.

\subsection{Classical GYO}
\label{sec:classical}

The GYO algorithm operates on \emph{hypergraphs}.  A \emph{hypergraph}
$H$ is a set of non-empty hyperedges. Recall from
Section~\ref{sec:preliminaries} that a hyperedge is just a finite set
of variables. Every \gcq is associated to a hypergraph as
follows.  
\begin{definition}
	\label{def:hypergraph}
	Let $Q$ be a \gcq. The hypergraph of $Q$, denoted $\hypergraph(Q)$, is
	the hypergraph
    $$\hypergraph(Q) = \{ \seq{x} \mid r(\seq{x}) \text{ is an atom of $Q$ with } \seq{x} \not = \emptyset \}.$$
\end{definition}

The GYO algorithm checks acyclicity of a normal conjunctive query $Q$
by constructing $\hypergraph(Q)$ and repeatedly removing \emph{ears}
from this hypergraph. If ears can be removed
until only the empty hypergraph remains, then the query is acyclic;
otherwise it is cyclic.

An \emph{ear} in a hypergraph $H$ is a hyperedge $e$ for which
we can divide its variables into two groups: (1) those that appear
exclusively in $e$, and (2) those that are contained in another
hyperedge $\ell$ of $H$.  A variable that appears exclusively in a
single hyperedge is also called an \emph{isolated variable}. Thus, ear
removal corresponds to executing the following two reduction
operations.
\begin{itemize}
	\item Remove isolated variables: select a hyperedge $e$ in $H$ and
	remove isolated variables from it; if $e$ becomes empty, remove
	$e$ it altogether from $H$.
	\item Subset elimination: remove hyperedge $e$ from $H$ if there
	exists another hyperedge $\ell$ for which $e \subseteq \ell$.
\end{itemize}
The \emph{GYO reduction} of a hypergraph is the hypergraph that is
obtained by executing these operations until no further operation is
applicable. The following result is standard; see e.g.,
\cite{abiteboul1995foundations} for a proof.

\begin{proposition}
	A \cq $Q$ is acyclic if and only if the
	GYO-reduction of $\hypergraph(Q)$ is the empty hypergraph.
\end{proposition}

\subsection{GYO-reduction for \gcqs}
\label{sec:gyo-reduction-gcqs}

In order to extend the GYO-reduction to check
free-connex acyclicity (not simply acyclicity) of \gcqs (not simply
standard \cqs), we will: (1) Redefine the notion of being an ear to
take into account the predicates; and (2)
transform the GYO-reduction into a two-stage procedure. The first
stage allows to check that a connex set with exactly $\free(Q)$ can
exist while the first and second stage combined check that the query
is acyclic.

Our algorithm operates on \emph{hypergraph
	triplets} instead of hypergraphs, which are defined as follows.

\begin{definition}
	A \emph{hypergraph triplet} is a triple $\trip H = (\hypergraph(\trip H), \free(\trip H), \preds(\trip H))$
	with $\hypergraph(\trip H)$ a hypergraph, $\free(\trip H)$ a
	hyperedge, and $\preds(\trip H)$ a set of predicates.
\end{definition}

Intuitively, the variables in $\free(\trip H)$ will correspond to the
output variables of a query and the set $\preds(\trip H)$ will contain 
predicates that need to be taken into account when removing
ears. Every \gcq is therefore naturally associated to a hypergraph
triplet as follows.

\begin{definition}
	The hypergraph triplet of a \gcq $Q$, denoted $\hypertrip(Q)$, is
	the triplet $(\hypergraph(Q), \free(Q), \preds(Q))$.
\end{definition}

In order to extend the notion of an ear, we require the following
definitions. Let $\trip H$ be a hypergraph
triplet. Variables that occur in $\free(\trip H)$ or in at least two
hyperedges in $\hypergraph(\trip H)$ are called \emph{equijoin
	variables} of $\trip{H}$. We denote the set of all equijoin
variables of $\trip H$ by $\equijoinvars(\trip H)$ and abbreviate
$\equijoinvars_{\trip{H}}(e) = e \cap \equijoinvars(\trip H)$.  A
variable $x$ is \emph{isolated} in $\trip{H}$ if it is not an equijoin
variable and is not mentioned in any predicate, i.e., if $x \not \in
\equijoinvars(\trip H)$ and $x \not \in \var(\preds(\trip H))$. We
denote the set of isolated variables of $\trip H$ by $\isolated(\trip
H)$ and abbreviate $\isolated_{\trip{H}}(e) = e \cap \isolated(\trip
H)$. The \emph{extended variables} of hyperedge $e$ in $\trip H$,
denoted $\ext_{\trip{H}}(e)$ is the set of all variables of predicates
that mention some variable in $e$, except the variables in $e$
themselves:
\[ \ext_{\trip{H}}(e) = \bigcup \{ \var(\theta) \mid \theta \in
\preds(\trip H), \var(\theta) \cap e \not = \emptyset\} \setminus e.\]
Finally, a hyperedge $e$ is a \emph{conditional subset} of hyperedge
$\ell$ w.r.t. $\trip{H}$, denoted $e \cse_{\trip H} \ell$, if
$\equijoinvars_{\trip{H}}(e) \subseteq \ell$ and $\ext_{\trip{H}}(e
\setminus \ell) \subseteq \ell$. We omit subscripts from our notation if the
triplet is clear from the context.

\tikzset{
    expand bubble/.style={
        preaction={draw,line width=1.4pt},
        gray!20,fill,draw,line width=2pt,
    },
    vertex/.style={
      minimum height=13pt,minimum width=13pt, node distance=23pt
    },
    hyperedge/.style={
      rounded corners, fill, opacity=0.35, line join=round
    },
    pred/.style={
      style=dashed,color=black,very thick
    },
    rewr/.style={
      ->,thick,decorate,
      decoration={snake,amplitude=.5mm,segment length=2mm, post length=1mm}
    }
}

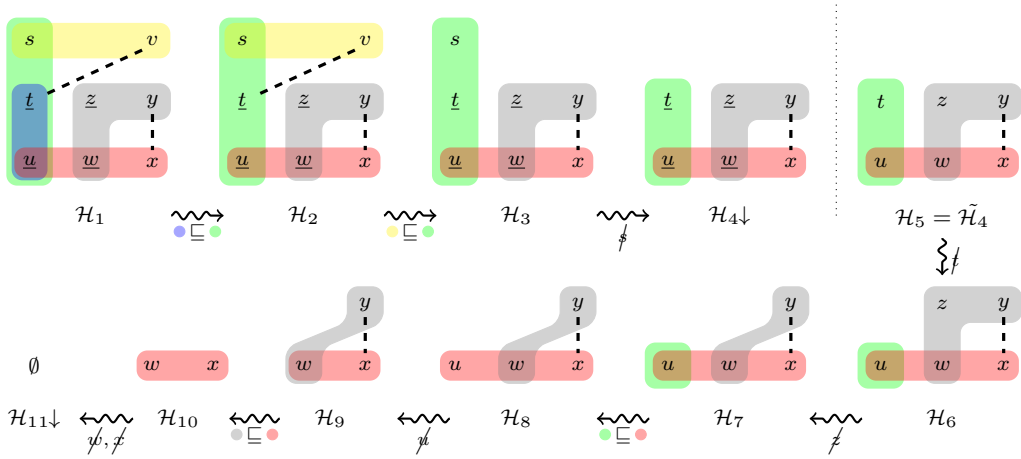
\begin{figure*}[t]
\centering
 \begin{tikzpicture}
   \begin{scope}
     \node[vertex,          ]  at(0,0) (s) {$s$};
     \node[vertex,below of=s] (t) {$\underline{t}$};
     \node[vertex,below of=t] (u) {$\underline{u}$};
     \node[vertex,right of=u] (w) {$\underline{w}$};
     \node[vertex,right of=w] (x) {$x$};
     \node[vertex,right of=t] (z) {$\underline{z}$};
     \node[vertex,right of=z] (y) {$y$};
     \node[vertex,above of=y] (v) {$v$};

     \node[below of=w, node distance=20pt] (H1) {$\trip H_1$};


     \begin{pgfonlayer}{background}
       \path[hyperedge, color=green] 
       ($ (s.north east) + (2pt,2pt)$) rectangle ($ (u.south west) - (2pt, 2pt)$);

       \path[hyperedge, color=yellow] 
       (s.north west) rectangle (v.south east);

       \path[hyperedge, color=blue] 
       (t.north east) rectangle (u.south west);

       \path[hyperedge, color=red] 
       ($(u.north west) + (1pt,-1pt)$) rectangle ($ (x.south east) + (-1pt,1pt)$);

       \draw[hyperedge, color=gray] 
       (w.south west) -- (z.north west) -- (y.north east) -- (y.south east) --
       (z.south east) -- (w.south east) -- cycle;

       \draw[pred] (t) -- ($ (v.south) + (-2pt,4pt)$);
       \draw[pred] ($ (y.south) + (0pt,2pt)$) -- ($ (x.north) + (0pt,-2pt) $);
     \end{pgfonlayer}
   \end{scope}

   \begin{scope}[xshift=2.8cm]
     \node[vertex,          ]  at(0,0) (s) {$s$};
     \node[vertex,below of=s] (t) {$\underline{t}$};
     \node[vertex,below of=t] (u) {$\underline{u}$};
     \node[vertex,right of=u] (w) {$\underline{w}$};
     \node[vertex,right of=w] (x) {$x$};
     \node[vertex,right of=t] (z) {$\underline{z}$};
     \node[vertex,right of=z] (y) {$y$};
     \node[vertex,above of=y] (v) {$v$};

     \node[below of=w, node distance=20pt] (H2) {$\trip H_2$};


     \begin{pgfonlayer}{background}
       \path[hyperedge, color=green] 
       ($ (s.north east) + (2pt,2pt)$) rectangle ($ (u.south west) - (2pt, 2pt)$);

       \path[hyperedge, color=yellow] 
       (s.north west) rectangle (v.south east);


       \path[hyperedge, color=red] 
       ($(u.north west) + (1pt,-1pt)$) rectangle ($ (x.south east) + (-1pt,1pt)$);

       \draw[hyperedge, color=gray] 
       (w.south west) -- (z.north west) -- (y.north east) -- (y.south east) --
       (z.south east) -- (w.south east) -- cycle;

       \draw[pred] (t) -- ($ (v.south) + (-2pt,4pt)$);
       \draw[pred] ($ (y.south) + (0pt,2pt)$) -- ($ (x.north) + (0pt,-2pt) $);
     \end{pgfonlayer}
   \end{scope}

   \begin{scope}[xshift=5.6cm]
     \node[vertex,          ]  at(0,0) (s) {$s$};
     \node[vertex,below of=s] (t) {$\underline{t}$};
     \node[vertex,below of=t] (u) {$\underline{u}$};
     \node[vertex,right of=u] (w) {$\underline{w}$};
     \node[vertex,right of=w] (x) {$x$};
     \node[vertex,right of=t] (z) {$\underline{z}$};
     \node[vertex,right of=z] (y) {$y$};
     \node[coordinate,above of=y] (v) {};

     \node[below of=w, node distance=20pt] (H3) {$\trip H_3$};


     \begin{pgfonlayer}{background}
       \path[hyperedge, color=green] 
       ($ (s.north east) + (2pt,2pt)$) rectangle ($ (u.south west) - (2pt, 2pt)$);



       \path[hyperedge, color=red] 
       ($(u.north west) + (1pt,-1pt)$) rectangle ($ (x.south east) + (-1pt,1pt)$);

       \draw[hyperedge, color=gray] 
       (w.south west) -- (z.north west) -- (y.north east) -- (y.south east) --
       (z.south east) -- (w.south east) -- cycle;

       \draw[pred] ($ (y.south) + (0pt,2pt)$) -- ($ (x.north) + (0pt,-2pt) $);
     \end{pgfonlayer}
   \end{scope}

   \begin{scope}[xshift=8.4cm]
     \node[coordinate       ]  at(0,0) (s) {};
     \node[vertex,below of=s] (t) {$\underline{t}$};
     \node[vertex,below of=t] (u) {$\underline{u}$};
     \node[vertex,right of=u] (w) {$\underline{w}$};
     \node[vertex,right of=w] (x) {$x$};
     \node[vertex,right of=t] (z) {$\underline{z}$};
     \node[vertex,right of=z] (y) {$y$};
     \node[coordinate,above of=y] (v) {};

     \node[below of=w, node distance=20pt] (H4) {$\nf{\trip H_4}$};


     \begin{pgfonlayer}{background}
       \path[hyperedge, color=green] 
       ($ (t.north east) + (2pt,2pt)$) rectangle ($ (u.south west) - (2pt, 2pt)$);



       \path[hyperedge, color=red] 
       ($(u.north west) + (1pt,-1pt)$) rectangle ($ (x.south east) + (-1pt,1pt)$);

       \draw[hyperedge, color=gray] 
       (w.south west) -- (z.north west) -- (y.north east) -- (y.south east) --
       (z.south east) -- (w.south east) -- cycle;

       \draw[pred] ($ (y.south) + (0pt,2pt)$) -- ($ (x.north) + (0pt,-2pt) $);
     \end{pgfonlayer}
   \end{scope}

   \begin{scope}[xshift=11.2cm]
     \node[coordinate       ]  at(0,0) (s) {};
     \node[vertex,below of=s] (t) {$t$};
     \node[vertex,below of=t] (u) {$u$};
     \node[vertex,right of=u] (w) {$w$};
     \node[vertex,right of=w] (x) {$x$};
     \node[vertex,right of=t] (z) {$z$};
     \node[vertex,right of=z] (y) {$y$};
     \node[coordinate,above of=y] (v) {};

     \node[below of =w, node distance=20pt] (H5) {$\trip H_5=\tilde{\trip H_4}$};


     \begin{pgfonlayer}{background}
       \path[hyperedge, color=green] 
       ($ (t.north east) + (2pt,2pt)$) rectangle ($ (u.south west) - (2pt, 2pt)$);



       \path[hyperedge, color=red] 
       ($(u.north west) + (1pt,-1pt)$) rectangle ($ (x.south east) + (-1pt,1pt)$);

       \draw[hyperedge, color=gray] 
       (w.south west) -- (z.north west) -- (y.north east) -- (y.south east) --
       (z.south east) -- (w.south east) -- cycle;

       \draw[pred] ($ (y.south) + (0pt,2pt)$) -- ($ (x.north) + (0pt,-2pt) $);
     \end{pgfonlayer}
   \end{scope}

   \begin{scope}[xshift=11.2cm,yshift=-2.5cm]
     \node[coordinate       ]  at(0,0) (s) {};
     \node[coordinate,below of=s] (t) {};
     \node[vertex,below of=t] (u) {$u$};
     \node[vertex,right of=u] (w) {$w$};
     \node[vertex,right of=w] (x) {$x$};
     \node[vertex,right of=t] (z) {$z$};
     \node[vertex,right of=z] (y) {$y$};
     \node[coordinate,above of=y] (v) {};

     \node[below of =w, node distance=20pt] (H6) {$\trip H_6$};


     \begin{pgfonlayer}{background}
       \path[hyperedge, color=green] 
       ($ (u.north east) + (2pt,2pt)$) rectangle ($ (u.south west) - (2pt, 2pt)$);



       \path[hyperedge, color=red] 
       ($(u.north west) + (1pt,-1pt)$) rectangle ($ (x.south east) + (-1pt,1pt)$);

       \draw[hyperedge, color=gray] 
       (w.south west) -- (z.north west) -- (y.north east) -- (y.south east) --
       (z.south east) -- (w.south east) -- cycle;

       \draw[pred] ($ (y.south) + (0pt,2pt)$) -- ($ (x.north) + (0pt,-2pt) $);
     \end{pgfonlayer}
   \end{scope}

   \begin{scope}[xshift=8.4cm,yshift=-2.5cm]
     \node[coordinate       ]  at(0,0) (s) {};
     \node[coordinate,below of=s] (t) {};
     \node[vertex,below of=t] (u) {$u$};
     \node[vertex,right of=u] (w) {$w$};
     \node[vertex,right of=w] (x) {$x$};
     \node[coordinate,right of=t] (z) {};
     \node[vertex,above of=x] (y) {$y$};
     \node[coordinate,above of=y] (v) {};

     \node[below of =w, node distance=20pt] (H7) {$\trip H_7$};


     \begin{pgfonlayer}{background}
       \path[hyperedge, color=green] 
       ($ (u.north east) + (2pt,2pt)$) rectangle ($ (u.south west) - (2pt, 2pt)$);



       \path[hyperedge, color=red] 
       ($(u.north west) + (1pt,-1pt)$) rectangle ($ (x.south east) + (-1pt,1pt)$);

       \draw[hyperedge, color=gray] 
       (w.south west) -- (w.north west) -- (y.south west) -- 
       (y.north west) -- (y.north east) -- (y.south east) -- 
       (w.north east) -- (w.south east) -- cycle;

       \draw[pred] ($ (y.south) + (0pt,2pt)$) -- ($ (x.north) + (0pt,-2pt) $);
     \end{pgfonlayer}
   \end{scope}

   \begin{scope}[xshift=5.6cm,yshift=-2.5cm]
     \node[coordinate       ]  at(0,0) (s) {};
     \node[coordinate,below of=s] (t) {};
     \node[vertex,below of=t] (u) {$u$};
     \node[vertex,right of=u] (w) {$w$};
     \node[vertex,right of=w] (x) {$x$};
     \node[coordinate,right of=t] (z) {};
     \node[vertex,above of=x] (y) {$y$};
     \node[coordinate,above of=y] (v) {};

     \node[below of =w, node distance=20pt] (H8) {$\trip H_8$};


     \begin{pgfonlayer}{background}
       ;



       \path[hyperedge, color=red] 
       ($(u.north west) + (1pt,-1pt)$) rectangle ($ (x.south east) + (-1pt,1pt)$);

       \draw[hyperedge, color=gray] 
       (w.south west) -- (w.north west) -- (y.south west) -- 
       (y.north west) -- (y.north east) -- (y.south east) -- 
       (w.north east) -- (w.south east) -- cycle;

       \draw[pred] ($ (y.south) + (0pt,2pt)$) -- ($ (x.north) + (0pt,-2pt) $);
     \end{pgfonlayer}
   \end{scope}

   \begin{scope}[xshift=2.8cm,yshift=-2.5cm]
     \node[coordinate       ]  at(0,0) (s) {};
     \node[coordinate,below of=s] (t) {};
     \node[vertex,below of=t] (u) {};
     \node[vertex,right of=u] (w) {$w$};
     \node[vertex,right of=w] (x) {$x$};
     \node[coordinate,right of=t] (z) {};
     \node[vertex,above of=x] (y) {$y$};
     \node[coordinate,above of=y] (v) {};

     \node[coordinate,below of=w, node distance=20pt] (lbl) {};
     \node (H9) at ($ (lbl) + (10pt,0pt)$) {$\trip H_{9}$};


     \begin{pgfonlayer}{background}
       ;



       \path[hyperedge, color=red] 
       ($(w.north west) + (1pt,-1pt)$) rectangle ($ (x.south east) + (-1pt,1pt)$);

       \draw[hyperedge, color=gray] 
       (w.south west) -- (w.north west) -- (y.south west) -- 
       (y.north west) -- (y.north east) -- (y.south east) -- 
       (w.north east) -- (w.south east) -- cycle;

       \draw[pred] ($ (y.south) + (0pt,2pt)$) -- ($ (x.north) + (0pt,-2pt) $);
     \end{pgfonlayer}
   \end{scope}

   \begin{scope}[xshift=0.8cm,yshift=-2.5cm]
     \node[coordinate       ]  at(0,0) (s) {};
     \node[coordinate,below of=s] (t) {};
     \node[vertex,below of=t] (u) {};
     \node[vertex,right of=u] (w) {$w$};
     \node[vertex,right of=w] (x) {$x$};
     \node[coordinate,right of=t] (z) {};
     \node[vertex,above of=x] (y) {};
     \node[coordinate,above of=y] (v) {};

     \node[coordinate,below of=w, node distance=20pt] (lbl) {};
     \node (H10) at ($ (lbl) + (10pt,0pt)$) {$\trip H_{10}$};


     \begin{pgfonlayer}{background}
       ;



       \path[hyperedge, color=red] 
       ($(w.north west) + (1pt,-1pt)$) rectangle ($ (x.south east) + (-1pt,1pt)$);


     \end{pgfonlayer}
   \end{scope}

   \node[left of=H10, node distance=1.9cm] (H11) {$\nf{\trip H_{11}}$};
   \node[above of=H11, node distance=20pt] (empty) {$\emptyset$};

   \draw[rewr] ($(H1) + (30pt,0pt)$) -- 
               node[below] {%
\scriptsize ${{\tikz \draw[hyperedge,color=blue] (0,0) circle
    [radius=2pt];} \cse {\tikz \draw[hyperedge,color=green] (0,0) circle [radius=2pt];}}$}
              ($(H2) - (30pt,0)$);

   \draw[rewr] ($(H2) + (30pt,0pt)$) -- 
               node[below] {%
\scriptsize ${{\tikz \draw[hyperedge,color=yellow] (0,0) circle [radius=2pt];} \cse {\tikz \draw[hyperedge,color=green] (0,0) circle [radius=2pt];}}$}
             ($(H3) - (30pt,0)$);

   \draw[rewr] ($(H3) + (30pt,0pt)$) -- 
               node[below] {\scriptsize $\cancel{s}$}
               ($(H4) - (30pt,0)$);

   \draw[dotted] ($(H4) + (40pt,0pt)$) -- ($(H4) + (40pt,80pt)$);
   \draw[rewr] ($(H5.south) + (0pt,-1pt)$) -- 
               node[right] {\scriptsize $\cancel{t}$}
               ($(H5.south) - (0pt,15pt)$);

   \draw[rewr] ($(H6) - (30pt,0pt)$) -- 
               node[below] {\scriptsize $\cancel{z}$}
               ($(H7) + (30pt,0)$);

   \draw[rewr] ($(H7) - (30pt,0pt)$) -- 
               node[below] {%
\scriptsize ${{\tikz \draw[hyperedge,color=green] (0,0) circle [radius=2pt];} \cse {\tikz \draw[hyperedge,color=red] (0,0) circle [radius=2pt];}}$}
               ($(H8) + (30pt,0)$);

   \draw[rewr] ($(H8) - (25pt,0pt)$) -- 
               node[below] {\scriptsize $\cancel{u}$}
               ($(H9) + (25pt,0)$);

   \draw[rewr] ($(H9) - (19pt,0pt)$) -- 
               node[below] {%
\scriptsize ${{\tikz \draw[hyperedge,color=gray] (0,0) circle [radius=2pt];} \cse {\tikz \draw[hyperedge,color=red] (0,0) circle [radius=2pt];}}$}
               ($(H10) + (19pt,0)$);

   \draw[rewr] ($(H10) - (17pt,0pt)$) --node[below] {\scriptsize $\cancel{w},\cancel{x}$} ($(H11) + (17pt,0)$);

\end{tikzpicture}
 \caption{Illustration of GYO-reduction for GCQs. Colored regions
   depict hyperedges. Variables in $\free$ are underlined. Variables
   occurring in the same predicate are connected by dashed lines.}
 \label{fig:hypergraphEx2}
\end{figure*}

\begin{example}
  \label{ex:hypergraphEx2}
  In Fig.~\ref{fig:hypergraphEx2} we depict several hypergraph
  triplets. There, hyperedges in $\trip H$ are depicted by colored
  regions and variables in $\free(\trip H)$ are underlined.  We use
  dashed lines to connect variables that appear together in a
  predicate. So, in $\trip H_1$, we have predicates $\theta_1,
  \theta_2$ with $\var(\theta_1) = \{t,v\}$ and $\var(\theta_2) =
  \{x,y\}$. Now consider triplet $\trip H_1$ in particular. It is the
  hypergraph triplet $\hypertrip(Q)$ for the following \gcq $Q$:
  \begin{multline*}
  Q= \pi_{t,u,z,w}(r_1(s, t, u) \Join r_2(t, u) \Join r_3(u, w, x)
  \Join \\ r_4(s, v)
  \Join r_5(w, z, y) \mid  t < v \wedge x < y).
  \end{multline*}
   Moreover, $\equijoinvars(\trip H_1) =\{s,t,u,w,z\} $
  and $\isolated(\trip H_1) = \emptyset$. Furthermore, $\ext_{\trip
    H_1}(\{v\}) = \{t\}$ since $\theta_1 = t < v$ shares
  variables with $\{v\}$. Finally $\equijoinvars_{\trip H_1}(\{s,v\}) =
  \{s\} \subseteq \{s,t,u\}$ and $\ext_{\trip H_1}(\{s,v\} \setminus
  \{s,t,u\}) = \ext_{\trip H_1}(\{v\}) = \{t\} \subseteq
  \{s,t,u\}$. Therefore, $\{s, v\} \cse_{\trip H_1} \{s,t,u\}$. Similarly,
  $\{t,u\} \cse_{\trip H_1} \{s,t,u\}$.
\end{example}


We define ears in our context as follows. 
\begin{definition}
	A hyperedge $e$ is an ear in a hypergraph triplet $\trip H$ if $e \in
	\hypergraph(\trip H)$ and either
	\begin{enumerate}
		\item we can divide its variables into
		two: (a) those that are isolated and (b) those that form a
		conditional subset of another hyperedge $\ell \in \hypergraph(\trip H)
		\setminus \{e\}$; or
		\item $e$ consists only of non-join variables, i.e., $\equijoinvars(e)
		= \emptyset$ and $\ext(e) = \emptyset$.
	\end{enumerate}
\end{definition}

Note that case (2) allows for $\theta \in \preds(\trip H)$ with
$\var(\theta) \subseteq e$. We call predicates that are covered by a
hyperedge in this sense \emph{filters} because they correspond to
filtering a single GMR instead of $\theta$-joining two
GMRs. If, in case (2), there is no filter $\theta$ with $\var(\theta)
\subseteq e$, then $e = \isolated_{\trip H}(e)$.
Similar to the classical GYO reduction, we can view ear removal
as a rewriting process on triplets, where we consider the
following reduction operations. 
\begin{itemize}[-]
	\item (ISO) Remove isolated variables: select a hyperedge $e \in
	\hypergraph(\trip H)$
	and remove a non-empty set $X \subseteq \isolated_{\trip{H}}(e)$
	from it. If $e$ becomes empty, remove it from
	$\hypergraph(\trip H)$. 
	\item (CSE) Conditional subset elimination: remove hyperedge $e$ from
	$\hypergraph(\trip H)$ if it is a conditional subset of another
	hyperedge $f$ in $\hypergraph(\trip H)$. Also update $\preds(\trip
	H)$ by removing all predicates $\theta$ with $\var(\theta) \cap (e
	\setminus f) \not = \emptyset$.
	\item (FLT) Filter elimination: select  $e \in
	\hypergraph(\trip H)$ and a non-empty subset of predicates $\Theta
	\subseteq \preds(\trip H)$ with $\var(\Theta) \subseteq e$.
	Remove all predicates in $\Theta$ from $\preds(\trip H)$.
\end{itemize}
We write $\trip{H} \rewr \trip{I}$ to denote that triplet $\trip{I}$
is obtained from triplet $\trip{H}$ by applying a single such
operation, and $\trip{H} \rewr^* \trip{I}$ to denote that $\trip{I}$
is obtained by a sequence of zero or more of such operations. 

\begin{example} 
\label{ex:reductions}
For the hypergraph triplets illustrated in
Fig.~\ref{fig:hypergraphEx2} we have $\trip H_1 \rewr \trip H_2
\rewr \trip H_3 \rewr \trip H_4$ and $\trip H_5 \rewr \allowbreak
\trip H_6 \allowbreak \rewr \trip H_7 \allowbreak \rewr \trip H_8 \rewr \trip H_9
\rewr \trip H_{10} \rewr \trip H_{11}$.  For each reduction, it is
illustrated in the figure which set of isolated variables is removed,
or which conditional subset is removed.
\end{example}

We write
$\nf{\trip{H}}$ to denote $\trip{H}$ is in \emph{normal form}, i.e.,
that no operation is applicable on triplet $\trip{H}$. Note that,
because each operation removes at least one variable, hyperedge, or
predicate, we will always reach a normal form after a finite number of
operations. Furthermore, while multiple different reduction steps may
be applicable on a given triplet $\trip{H}$, the order in which we
apply them does not matter:
\begin{proposition}[Confluence]
	\label{prop:confluence}
	Whenever $\trip{H} \rewr^* \trip{I}_1$ and $\trip{H} \rewr^* \trip{I}_2$,
	there exists $\trip{J}$ such that $\trip{I}_1 \rewr^* \trip{J}$
	and $\trip{I}_2 \rewr^* \trip{J}$. 
\end{proposition}
Because the proof is technical but not overly enlightning, we defer it
to Appendix~\ref{sec:proof-prop-confluence}. A direct consequence is that normal forms
are unique: if $\trip{H} \rewr^* \nf{\trip{I}_1}$ and $\trip{H}
\rewr^* \nf{\trip{I}_2}$ then $\trip{I}_1 = \trip{I}_2$.

Let $\trip{H}$ be a triplet. The residual of $\trip{H}$, denoted
$\tilde{\trip{H}}$, is the triplet $(\hypergraph(\trip H),
\emptyset, \preds(\trip H))$, i.e., the triplet where $\free(\trip H)$
is set to $\emptyset$. A triplet is \emph{empty} if it equals
$(\emptyset, \emptyset, \emptyset)$.

Our main result in this section states that to
check whether a \gcq $Q$ is free-connex acyclic it suffices to start
from $\hypertrip(Q)$ and do a two stage reduction: the first from
$\hypertrip(Q)$ until a normal form $\nf{\trip{I}}$ is reached, and
the second from the residual of $\nf{\trip{I}}$, until
another normal form $\trip{J}$ is reached.\footnote{Note that because we set
	$\free(\trip I) = \emptyset$ on the residual, new variables may
	become isolated and therefore more reductions steps may be possible on
	the normal form of $\trip{I}$.}

\begin{theorem}
	\label{thm:gyo}
	Let $Q$ be a \gcq. Assume $\hypertrip(Q) \rewr^* \nf{\trip{I}}$
	and $\tilde{\trip{I}} \rewr^* \nf{\trip{J}}$. Then the following hold.
	\begin{enumerate}
		\item $Q$ is acyclic if, and only if, $\trip{J}$ is the empty
		triplet.
		\item $Q$ is free-connex acyclic if, and only if, $\trip{J}$ is
		the empty triplet and $\var(\hypergraph(\trip{I})) = \free(Q)$.
		\item For every \gjt $T$ of $Q$ and every connex subset $N$ of $T$
		it holds that $\var(\hypergraph(\trip{I})) \subseteq \var(N)$.
	\end{enumerate}
\end{theorem}
We devote Section~\ref{sec:gyo-correctness} to the proof.

\begin{example} 
	\label{ex:hyp}
	Fig.~\ref{fig:hypergraphEx2} illustrates the two-stage sequence of
	reductions starting from $\hypertrip(Q)$ with $Q$ the \gcq of
	Example~\ref{ex:hypergraphEx2}. Note that $\hypertrip(Q) = \trip H_1$
	and $\trip H_5$ is the residual of $\trip H_4$. Because we end with
	the empty triplet, $Q$ is acyclic but not free-connex since $\free(Q) \subsetneq \var(\trip H_4)$.
\end{example}

Theorem~\ref{thm:gyo} gives us a decision procedure for checking
free-connex acyclicity of \gcq $Q$. From its proof in
Section~\ref{sec:gyo-correctness}, we can actually derive an algorithm for
constructing a compatible \gjt pair for $Q$. At its essence, this
algorithm starts with the set of atoms appearing in $Q$, and
subsequently uses the sequence of reduction steps from
Theorem~\ref{thm:gyo} to construct a \gjt from it, at the same time
checking free-connex acyclicity. Every reduction step causes new nodes
to be added to the partial \gjt constructed so far. We will refer to
such partial \gjts as \emph{Generalized Join Forests} (\gjf).

\begin{definition}[\gjf]
	A \emph{Generalized Join Forest }is a set $F$ of pairwise disjoint
	\gjts s.t. for distinct trees $T_1, T_2 \in F$ we have
	$\var(T_1) \cap \var(T_2) = \var(n_1) \cap \var(n_2)$ where $n_1$
	and $n_2$ are the roots of $T_1$ and $T_2$.
\end{definition}

Every \gjf encodes a hypergraph as follows.

\begin{definition}
	The hypergraph $\hypergraph(F)$ associated to \gjf $F$ is the
	hypergraph that has one hyperedge for every non-empty root node in
	$F$,
	$$\hypergraph(F) = \{ \var(n) \mid n \text{ root  node in } F, \var(n) \not =
	\emptyset\}.$$
\end{definition}


The \gjt construction algorithm does not manipulate hypergraph
triplets directly. Instead, it manipulates \emph{\gjf triplets}. A
\gjf triplet is defined like a hypergraph triplet, except that it has
a \gjf instead of a hypergraph.

\begin{definition}
	A \emph{\gjf triplet} is a triple $\ftrip F = (\forest(\ftrip F), \allowbreak
	\free(\ftrip F), \fpreds{F})$ with $\forest(\ftrip F)$ a \gjf,
	$\free(\ftrip F)$ a hyperedge, and $\Theta_{\ftrip F}$ a set of
	predicates. Every \gjf triplet $\ftrip{F}$ induces a hypergraph
	triplet $\hypertrip(\ftrip{F}) = (\hypergraph(\forest(\ftrip F)),
	\free(\ftrip F), \allowbreak \fpreds{F})$.
\end{definition}

The algorithm for constructing a \gjt pair compatible with a given
\gcq $Q$ is now shown in Algorithm~\ref{alg:gjt}.  It starts in
line~\ref{alg:gjt-1} by initializing the \gjf triplet $\ftrip{F}$ to
$\ftrip{F} = (\forest(Q), \free(Q), \pred(Q)$. Here, $\forest(Q)$ is
the \gjf obtained by creating, for every atom $r(\overline{x})$ that
occurs $k > 0$ times in $Q$, $k$ corresponding leaf nodes labeled by
$r(\overline{x})$. In Lines~\ref{alg:gjt-2}--\ref{alg:gjt-3},
Algorithm~\ref{alg:gjt} then performs the first phase of reduction
steps of Theorem~\ref{thm:gyo}.  To this end, it checks whether a
reduction operation is applicable to $\hypertrip(\ftrip{F})$ and, if
so, \emph{enacts} this operation by modifying $\ftrip{F}$ as follows.
\begin{itemize}[-]
	\item (ISO). If the reduction operation on the hypergraph triplet
	$\hypertrip(\ftrip{F})$ were to remove a non-empty subset $X$ of
	isolated variables from hyperedge $e$, then $\ftrip{F}$ is modified
	as follows. Let $n_1,\dots, n_k$ be all the root nodes in
	$\forest(\ftrip{F})$ that are labeled by $e$. Merge the
	corresponding trees into one tree by creating a new node $n$ with
	$\var(n) = e$ and attaching $n_1, \dots, n_k$ as children to it with
	$\pred(n \to n_i) = \emptyset$ for $1 \leq i \leq k$. Then, enact
	the removal of $X$ by creating a new node $p$ with $\var(p) = e
	\setminus X$ and attaching $n$ as child to it with $\pred(p \to n) =
	\emptyset$. 
	\item (CSE) If the reduction operation on $\hypertrip(\ftrip{F})$ were
	to remove a hyperedge $e$ because it is a conditional subset of
	another hyperedge $\ell$, then $\ftrip{F}$ is modified as follows. Let
	$n_1,\dots, n_k$ (resp. $m_1, \dots, m_l$) be all the root nodes in
	$\forest(\ftrip{F})$ that are labeled by $e$ (resp. $\ell$), and let
	$T_1,\dots, T_k$ (resp. $U_1,\dots, U_l$) be their corresponding
	trees. Similar to the previous case, merge the $T_i$ (resp. $U_j$)
	into a single tree with new root $n$ labeled by $e$ (resp. $m$
	labeled by $\ell$). Then enact the removal of $e$ by creating a new
	node $p$ with $\var(p) = \ell$ and attaching $n$ and $m$ as children
	with $\pred(p \to n) = \{ \theta \in \preds(\ftrip F) \mid
	\var(\theta) \cap (e \setminus \ell) \not = \emptyset\}$ and $\pred(p
	\to m) =
	\emptyset$. 
	\item (FLT) If the reduction operation on $\hypertrip(\ftrip{F})$ were
	to remove non-empty set of predicates $\Theta$ because there exists
	a hyperedge $e$ with $\var(\Theta) \subseteq e$, then $\ftrip{F}$ is
	modified as follows. Let $n_1,\dots, n_k$ be all the root nodes in
	$\forest(\ftrip{F})$ that are labeled by $e$. Merge the
	corresponding trees into one tree by creating a new root $n$ labeled
	by $e$, and attaching $n_1, \dots, n_k$ as children with $\pred(n
	\to n_i) = \Theta$. Enact the removal of $\Theta$ by removing all
	$\theta \in \Theta$ from $\Theta(\ftrip F)$.
\end{itemize}

\begin{algorithm}[tb]
  \caption {Compute a \gjt pair}
  \label{alg:gjt}
  \begin{algorithmic}[1]
    \State \textbf{Input:} A \gcq $Q$.
    \medskip
    
    \State $\ftrip{F} \leftarrow (\forest(Q), \free(Q), \preds(Q))$
    \label{alg:gjt-1} 
    \While{ a reduction step is applicable to $\hypertrip(\ftrip{F})$}
    \label{alg:gjt-2}  
    \State enact the reduction on $\ftrip{F}$
    \EndWhile
    \label{alg:gjt-3}
    \State $X \leftarrow$ set of all root nodes in $\ftrip{F}$
    \label{alg:gjt-4}
    \State set $\preds(\ftrip F) := \emptyset$
    \label{alg:gjt-5}
    \While{ a reduction step is applicable to $\hypertrip(\ftrip{F})$}
    \label{alg:gjt-6}
    \State enact the reduction on $\ftrip{F}$
    \EndWhile
    \label{alg:gjt-7}
    \If{ $\hypertrip(\ftrip{F})$ is not the empty triplet}
    \label{alg:gjt-8}
    \State \textbf{error} ``$Q$ is not acyclic''
    \label{alg:gjt-9}
    \Else
    \State $T \leftarrow$ tree obtained by connecting all root nodes
    of $\ftrip{F}$'s forest to a new root, labeled by $\emptyset$
    \label{alg:gjt-10}
    \State $N \leftarrow$ all nodes in $X$ and their ancestors in $T$
    \label{alg:gjt-11}
    \State \textbf{return} $(T,N)$
    \label{alg:gjt-12}
    \EndIf
  \end{algorithmic}
\end{algorithm}

It is straightforward to check that these modifications of the forest
triplet $\ftrip{F}$ faithfully enact the corresponding operations on
$\hypertrip(\ftrip{F})$, in the following sense.

\begin{lemma}
	\label{lem:forest-enactment}
	Let $\ftrip{F}$ be a forest triplet and assume
	$\hypertrip(\ftrip{F}) \rewr \trip I$. Let $\ftrip{G}$ be the result
	of enacting this reduction operation on $\ftrip{F}$.  Then
	$\ftrip{G}$ is a valid forest triplet and $\hypertrip(\ftrip{G}) =
	\trip I$. 
\end{lemma}

We continue the explanation of Algorithm~\ref{alg:gjt}.
In line \ref{alg:gjt-4}, Algorithm~\ref{alg:gjt} records the set of
root nodes obtained after the first stage of reductions. It then
sets $\free(\ftrip F) = \emptyset$ in
line~\ref{alg:gjt-5} and continues with the second stage of reductions
in lines~\ref{alg:gjt-6}--\ref{alg:gjt-7}. It then employs
Theorem~\ref{thm:gyo} to check acyclicity of $Q$. If $Q$ is not
acyclic, it reports this in lines \ref{alg:gjt-8}--\ref{alg:gjt-9}. If
$Q$ is acyclic, then we know by Theorem~\ref{thm:gyo} that
$\hypertrip(\ftrip{F})$ has become the empty triplet. Note that
$\hypertrip(\ftrip{F})$ can be empty only if all the roots of
$\ftrip{F}$'s join forest are labeled by the empty set of
variables. As such, we can transform this forest into a join tree $T$
by linking all of these roots to a new unique root, also labeled
$\emptyset$. This is done in line~\ref{alg:gjt-10}. In
line~\ref{alg:gjt-11}, the set of nodes $N$ is computed, and consists
of all nodes identified at the end of the first stage
(line~\ref{alg:gjt-4}) plus all of their parents in $T$.

We will prove in Section~\ref{sec:gyo-correctness} that
Algorithm~\ref{alg:gjt} is correct, in the following sense.

\begin{theorem}
  \label{thm:construction-correct}
  Given a \gcq $Q$, Algorithm~\ref{alg:gjt} reports an error if $Q$ is
  cyclic. Otherwise, it returns a sibling-closed \gjt pair $(T,N)$
  with $T$ a \gjt for $Q$. If $Q$ is free-connex acyclic, then $(T,N)$
  is compatible with $Q$. Otherwise, $\free(Q) \subsetneq \var(N)$,
  but $\var(N)$ is minimal in the sense that for every other \gjt pair
  $(T', N')$ with $T'$ a \gjt for $Q$ we have $\var(N) \subseteq
  \var(N')$.
\end{theorem}
It is straightforward to check that this algorithm runs in polynomial
time in the size of $Q$.

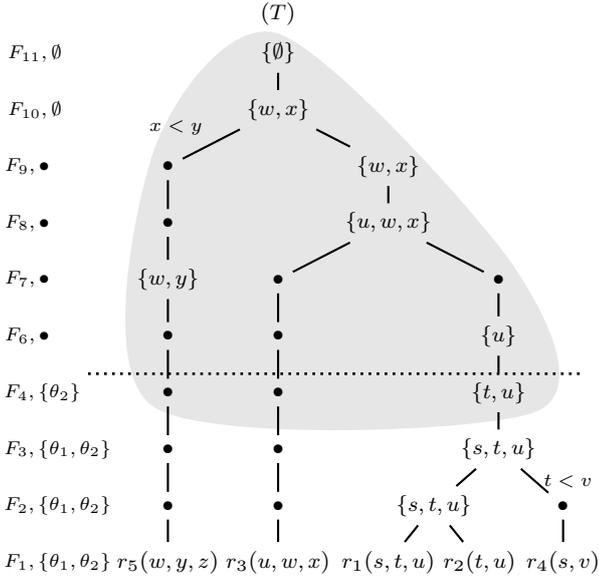
\begin{figure}[t]
\tikzset{
    expand bubble/.style={
        preaction={draw,line width=1.4pt},
        gray!20,fill,draw,line width=2pt,
    },
}
	\centering
	\begin{tikzpicture}[-,>=stealth',thick]
	\begin{scope}[xshift=-4cm]
	\draw[style=dotted, color=black, line width=1pt] (-2.5,-4.25) -- (4,-4.25);
	\node[color=white][label=center:{\scriptsize$F_1, \{\theta_1,\theta_2\}$}] at(-2.9, -6.75)(s){};
	\node[color=white][label=center:{\scriptsize$F_2, \{\theta_1,\theta_2\}$}] at(-2.9, -6.0)(s){};
	\node[color=white][label=center:{\scriptsize$F_3, \{\theta_1,\theta_2\}$}] at(-2.9, -5.25)(s1){};
	\node[color=white][label=center:{\scriptsize$F_4, \{\theta_2\}$}] at(-3.1, -4.50)(s2){};
	\node[color=white][label=center:{\scriptsize$F_6, \bullet$}] at(-3.3, -3.75)(s3){};
	\node[color=white][label=center:{\scriptsize$F_7, \bullet$}] at(-3.3, -3.0)(s4){};
	\node[color=white][label=center:{\scriptsize$F_8, \bullet$}] at(-3.3, -2.25)(s5){};
	\node[color=white][label=center:{\scriptsize$F_9, \bullet$}] at(-3.3, -1.50)(s6){};
	\node[color=white][label=center:{\scriptsize$F_{10},\emptyset$}] at(-3.2, -0.75)(s7){};
	\node[color=white][label=center:{\scriptsize$F_{11},\emptyset$}] at(-3.2, 0.0)(s8){};

	\node (w) [label={above:$(T)$}] { $\{\emptyset\}$ }[level distance=0.75cm,sibling distance=1.9cm]
	child{
		node(c8) {$\{w,x\}$}[sibling distance=2.9cm]
		child{
			node(c9) {$\bullet$} 
			child{
				node(c10) {$\bullet$}
				child{
					node(c11) {$\{w,y\}$}
					child{
                                          node (c112) { $\bullet$}
                                          child {
						node(c12) {$\bullet$}
						child{
							node(c13) {$\bullet$}
							child{
								node(c14) {$\bullet$}
								child{
									node(c15) {$r_5(w,y,z)$}
								}
							}
						}
					}
				}
                              }
			}edge from parent node[left,anchor=south east] {\scriptsize $x<y$}
		}
		child{
			node(c7) {$\{w,x\}$}
			child{
				node(c5) {$\{u,w,x\}$}[sibling distance=2.9cm]
				child{
                                   node (c15a) {$\bullet$}
                                    child {
					node(c15) {$\bullet$}
					child{
						node(c16) {$\bullet$}
						child{
							node(c17) {$\bullet$}
							child{
								node(c18) {$\bullet$}
								child{
									node(l4) {$r_3(u,w,x)$}
								}
							}
						}
					}
                                      }
				}
				child{
                                  node (c4a) {$\bullet$}
                                   child {
					node(c4){$\{u\}$}
					child{
						node(c3){ $\{t,u\}$ } 
						child{
							node (c1) { $\{s,t,u\}$ } [sibling distance=1.7cm]
							child{
								node (c2){$\{s,t,u\}$}[sibling distance=1.2cm]
								child{ node(l1) {$r_1(s,t,u)$}}
								child{ node(l2) {$r_2(t,u)$}}
							}
							child{
								node (c22){$\bullet$}
								child{
									node(l3){$r_4(s,v)$}
								}edge from parent node[right] {\scriptsize $t<v$}
							}
						}
					}
                                      }
				}
			}
		}
	};
	\begin{pgfonlayer}{background}
	\path[fill=gray!20]plot [smooth cycle,tension=0.5]
	coordinates {(w.north) (c9.north west) (c12.south west)
          (c3.south east)           (c4a.north east) };
	\end{pgfonlayer}
	
	\end{scope}
	\end{tikzpicture}
	\caption{\gjt Construction by  GYO-reduction.}
	\label{fig:gjtconstruction}
\end{figure}

\begin{example}
	\label{ex:forest}
	
	In Fig.~\ref{fig:gjtconstruction}, we show a \gjt $T$ and use
        this $\gjt$ to illustrate a number of \gjfs $F_1,\dots,
        F_{10}$ in the following way: let level $1$ be the leaf nodes,
        level $2$ the parents of the leaves, and so on. Then we take
        \gjf $F_i$ to be the set of all trees rooted at nodes at level
        $i$, for $1 \leq i \leq 10$, and with each level $i$, we
        mention the set of remaining predicates $\theta_i$ for $1 \leq
        i \leq k$ where $k$ is the number of predicates in $Q$.  Nodes
        (resp. predicates with each $F_i$) labeled by ``$\bullet$'' in
        Fig.~\ref{fig:gjtconstruction} indicates that the node (and
        hence tree, resp. predicates) was already present in $F_{i-1}$
        and did not change. These should hence not be interpreted as
        new nodes (resp. predicates changed). With this coding of
        forests, it is easy to see that for all $1 \leq i \leq 9$,
        $F_i = \hypergraph(\trip H_i)$ with $\trip H_i$ illustrated in
        Fig.~\ref{fig:hypergraphEx2} (note here that the hypergraph
        of residual of $\trip H_4$ i.e. $\trip H_5$ is the same as
        $\trip H_4$, hence we do not show the corresponding $F_5$).
        Furthermore, $\preds(F_i) = \preds(Q) \setminus \preds(\trip
        H_i)$ with $Q$ the \gcq from
        Example~\ref{ex:hypergraphEx2}. As such, the tree illustrates
        the sequence of \gjf triplets that is obtained by enacting the
        hypergraph reductions illustrated in
        Fig.~\ref{fig:hypergraphEx2}. For example, let $\ftrip F_1 =
        (F_1, \free(Q), \preds(Q)$. After enacting the removal of
        hyperedge $\{t,u\}$ from $\trip H_1$ to obtain $\trip H_2$ we
        obtain $\ftrip F_2 = (F_2, \free(Q), \preds(Q))$. Here, $F_2$
        is obtained by merging the single-node trees (i.e. labelled by
        the atoms in $Q$) $\{s,t,u\}$ and $\{t,u\}$ in to a single
        tree with root $\{s,t,u\}$. The shaded area illustrate the
        nodes in the connex subset $N$ computed by
        Algorithm~\ref{alg:gjt}.
\end{example}

We stress that Algorithm~\ref{alg:gjt} is non-deterministic in the
sense that the pair $(T,N)$ returned depends on the order in which the
reduction operations are performed. 

\subsection{Correctness}
\label{sec:gyo-correctness}
To prove theorems~\ref{thm:gyo} and \ref{thm:construction-correct} we show some propositions. 

\begin{proposition}
  \label{prop:alg-gjt-soundess}
  Let $Q$ be a \gcq. Assume $\hypertrip(Q) \rewr^*
  \nf{\trip{I}}$ and $\tilde{\trip{I}} \rewr^* \nf{\trip{J}}$.  If
  $\trip{J}$ is the empty triplet, then, when run on $Q$,
  Algorithm~\ref{alg:gjt} returns a pair $(T,N)$ s.t. $T$ is
  a \gjt for $Q$, $N$ is sibling-closed, and $\var(N) =
  \var(\hypergraph(\trip I))$.
\end{proposition}
\begin{proof}
  Assume that $\trip J$ is the empty triplet.  Algorithm~\ref{alg:gjt}
  starts in line~\ref{alg:gjt-2} by initializing $\ftrip{F} =
  (\forest(Q),\allowbreak \free(Q), \allowbreak \preds(Q))$. Clearly,
  $\hypertrip(\ftrip F) = \hypertrip(Q)$ at this
  point. Algorithm~\ref{alg:gjt} subsequently modifies $\ftrip F$
  throughout its execution. Let $\ftrip H$ denote the initial version
  of $\ftrip F$; let $\ftrip I$ denote the version of $\ftrip F$ when
  executing line~\ref{alg:gjt-4}; let $\tilde{\ftrip I}$ denote the
  version of $\ftrip F$ after executing line \ref{alg:gjt-5} and let
  $\ftrip J$ denote the version of $\ftrip F$ when executing line
  \ref{alg:gjt-8}. By repeated application of
  Lemma~\ref{lem:forest-enactment} we know that $\hypertrip(Q) =
  \hypertrip(\ftrip H) \rewr^* \hypertrip(\ftrip I)$. Furthermore,
  $\hypertrip(\ftrip I)$ is in normal form. Since also $\hypertrip(Q)
  \rewr^* \nf{\trip I}$ and normal forms are unique,
  $\hypertrip(\ftrip I) = \trip I$. Therefore,
  $\hypertrip(\tilde{\ftrip I}) = \tilde{\trip I}$. Again by repeated
  application of Lemma~\ref{lem:forest-enactment} we know that
  $\tilde{\trip I} = \hypertrip(\tilde{\ftrip I}) \rewr^*
  \hypertrip(\ftrip J)$. Moreover, $\hypertrip(\ftrip J)$ is in normal
  form. Since also $\tilde{\trip I} \rewr^* \nf{\trip J}$ and normal
  forms are unique, $\hypertrip(\ftrip J) = \trip J$. As $\trip J$
  is empty, we will execute lines~\ref{alg:gjt-10}--\ref{alg:gjt-12}.
  Since $\trip J$ is the empty hypergraph triplet, every root of every
  tree in $\forest(\ftrip J)$ must be labeled by $\emptyset$. By
  definition of join forests, no two distinct trees in $\forest(\ftrip
  J)$ hence share variables. As such, the tree $T$ obtained in line
  \ref{alg:gjt-10} by linking all of these roots to a new unique root,
  also labeled $\emptyset$, is a valid \gjt.

  We claim that $T$ is a \gjt for $Q$. Indeed, observe that $\atoms(T)
  = \atoms(Q)$ and the number of times that an atom occurs in $Q$
  equals the number of times that it occurs as a label in $T$. This is
  because initially $\forest(\ftrip H) = \forest(Q)$ and by enacting
  reduction steps we never remove nor add nodes labeled by
  atoms. Furthermore $\preds(T) = \preds(Q)$. This is because
  initially $\preds(\ftrip H) = \preds(Q)$ yet $\fpreds{J}$ is empty.
  This means that, for every $\theta \in \preds(Q)$, there was some
  reduction step that removed $\theta$ from the set of predicates of
  the current \gjf triplet $\ftrip F$. However, when enacting
  reduction steps we only remove predicates after we have added them
  to $\forest(\ftrip F)$. Therefore, every predicate in $\preds(Q)$
  must occur in $T$. Conversely, during enactment of reduction steps
  we never add predicates to $\forest(\ftrip F)$ that are not in
  $\fpreds F$, so all predicates in $T$ are also in $\preds(Q)$. Thus,
  $T$ is a \gjt for $Q$.

  It remains to show that $N$ is a sibling-closed connex subset of $T$
  and $\var(\hypergraph(\trip I)) = \var(N)$.  To this end, let $X$ be
  the set of all root nodes of $\forest(\ftrip I)$, as computed in
  Line~\ref{alg:gjt-4}. Since $\ftrip{J}$ is obtained from
  $\tilde{\ftrip I}$ by a sequence of reduction enactments, and since
  such enactments only add new nodes and never delete them, $M$ is a
  subset of nodes of $\forest(\ftrip J)$ and therefore also of $T$.
  As computed in Line~\ref{alg:gjt-11}, $N$ consists of $X$ and all
  ancestors of nodes of $X$ in $T$. Then $N$ is a connex subset of $T$
  by definition. Moreover, since enactments of reduction steps can
  only merge existing trees or add new parent nodes (never new child
  nodes), $N$ must also be sibling-closed. Furthermore, since
  $\hypertrip(\ftrip I) = \trip I$, $\hypergraph(\forest(\ftrip I)) =
  \hypergraph(\trip I)$. Thus, $\var(X) = \var(\hypergraph(\ftrip
  I)) = \var(\hypergraph(\trip I))$. Then, since $X$ is the frontier
  of $N$ and $N$ is sibling-closed we have $\var(N) = \var(X) =
  \var(\hypergraph(\trip I))$ by Lemma~\ref{lem:sibling-closed-var-at-frontier}.
\end{proof}

\begin{corollary}[Soundness]
	\label{prop:soundness}
	Let $Q$ be a \gcq and assume that $\hypertrip(Q) \rewr^* \nf{\trip{I}}$
	and $\tilde{\trip{I}} \rewr^* \nf{\trip{J}}$. Then:
	\begin{enumerate}
		\item If $\trip{J}$ is the empty triplet then $Q$ is acyclic.
		\item If $\trip{J}$ is the empty triplet and
		$\var(\hypergraph(\trip{I})) = \free(Q)$ then $Q$ is free-connex
		acyclic.
	\end{enumerate}
\end{corollary}

To also show completeness, we will interpret a \gjt $T$ for a \gcq $Q$
as a ``parse tree'' that specifies the two-stage sequence of reduction
steps that can be done on $\hypertrip(Q)$ to reach the empty
triplet. Not all \gjts will allows us to do so easily, however, and we
will therefore restrict our attention to those \gjts that are
\emph{canonical}.

\begin{definition}[Canonical]
	A \gjt $T$ is \emph{canonical} if:
	\begin{enumerate}
		\item its root is labeled by $\emptyset$;
		\item every leaf node $n$ is the child of an internal node $m$ with
		$\var(n) = \var(m)$;
		\item for all internal nodes $n$ and $m$ with $n \not = m$ we have
		$\var(n) \not = \var(m)$; and
		\item for every edge $m \to n$ and all $\theta \in \preds(m \to n)$
		we have $\var(\theta) \cap (\var(n) \setminus \var(m)) \not = \emptyset$.
	\end{enumerate}
	A connex subset $N$ of $T$ is \emph{canonical} if every node in it
	is interior in $T$.  A \gjt pair $(T, N)$ is canonical if both $T$
	and $N$ are canonical.
\end{definition}

The following proposition, proven in Appendix~\ref{sec:app-gyo}, shows
that we may restrict our attention to canonical \gjt pairs without
loss of generality.
\begin{restatable}{proposition}{canonicalPair}
	\label{prop:canonical}
	For every \gjt pair there exists an equivalent canonical pair.
\end{restatable}

We also require the following auxiliary notions and insights.
First, if $(T,N)$ is a \gjt pair, then define the \emph{hypergraph
	associated to $(T,N)$}, denoted $\hypergraph(T,N)$, to be the
hypergraph formed by node labels in $N$,
$$\hypergraph(T,N) = \{ \var_T(n) \mid n \in N, \var_T(n) \not =
\emptyset \}.$$
Further, define $\preds(T,N)$ to be the set of all
predicates occurring on edges between nodes in $N$. For a hyperedge
$\seq{z}$, define the \emph{hypergraph triplet} of $(T,N)$
w.r.t. $\seq{z}$, denoted $\hypertrip(T,N,\seq{z})$ to be the
hypergraph triplet $(\hypergraph(T,N), \allowbreak \seq{z},
\preds(T,N))$. 

The following technical Lemma shows that we can use canonical pairs as
``parse'' trees to derive a sequence of reduction steps. Its proof can be found in Appendix~\ref{sec:app-gyo}.
\begin{restatable}{lemma}{progress}
	\label{lem:progress}
	Let $(T,N_1)$ and $(T,N_2)$ be canonical \gjt pairs with $N_2 \subseteq
	N_1$. Then $\hypertrip(T,N_1, \seq{z}) \rewr^* \hypertrip(T,N_2,
	\seq{z})$ for every $\seq{z} \subseteq \var(N_2)$.
\end{restatable}

We require the following additional lemma, proven in
Appendix~\ref{sec:app-gyo}:
\begin{restatable}{lemma}{subsetRemoval}
  \label{lem:subset-removal}
  Let $H_1$ and $H_2$ be two hypergraphs such that for all $e \in H_2$
  there exists $\ell \in H_1$ such that $e \subseteq \ell$. Then $(H_1 \cup
  H_2, \seq{z}, \Theta) \rewr^* (H_1, \seq{z}, \Theta)$, for every
  hyperedge $\seq{z}$ and set of predicates $\Theta$.
\end{restatable}

We these tools in hand we can prove completeness.
\begin{proposition}
  \label{prop:completeness-core}
  Let $Q$ be a \gcq, let $T$ be a \gjt for $Q$ and let $N$ be a connex
  subset of $T$ with $\free(Q) \subseteq \var(N)$. Assume that
  $\hypertrip(Q) \rewr^* \nf{\trip{I}}$ and $\tilde{\trip{I}} \rewr^*
  \nf{\trip{J}}$. Then $\trip J$ is the empty triplet and
  $\var(\hypergraph(\trip{I})) \subseteq \var(N)$. 
\end{proposition}
\begin{proof}
	By Proposition~\ref{prop:canonical} we may assume without loss of
	generality that $(T, N)$ is a canonical \gjt pair.  Let $A$ be the
	set of all of $T$'s interior nodes. Clearly, $A$ is a connex subset
	of $T$ and $\var(A) \subseteq \var(Q)$. Furthermore, because for
	every atom $r(\overline{x})$ in $Q$ there is a leaf node $l$ in $T$
	labeled by $r(\overline{x})$ (as $T$ is a \gjt for $Q$), which has a
	parent interior node $n_l$ labeled $\overline{x}$ (because $T$ is
	canonical), also $\var(Q) \subseteq \var(A)$. Therefore, $\var(A) =
	\var(Q)$. By the same reasoning, $\hypergraph(Q) \subseteq
	\hypergraph(T,A)$. Therefore, $\hypergraph(T, A) = \hypergraph(T,A)
	\cup \hypergraph(Q)$. Furthermore, because every interior node in a
	\gjt has a guard descendant, and the leaves of $T$ are all labeled
	by atoms in $Q$, we know that for every node $n \in A$ there exists
	some hyperedge $f \in \hypergraph(Q)$ such that $\var(n) \subseteq
	\var(f)$.  In addition, we claim that $\preds(T,A) =
	\preds(Q)$. Indeed, $\preds(T,A) \subseteq \preds(Q)$ since $T$ is a
	\gjt for $Q$. The converse inclusion follows from canonicality
	properties (2) and (4): because leaf nodes in a canonical \gjt have
	a parent labeled by the same hyperedge, there can be no predicates
	on edges to leaf nodes in $T$. Thus, all predicates in $T$ are on
	edges between interior nodes, i.e., in $\preds(T,A)$. Then, because
	every predicate in $Q$ appears somewhere in $T$ (since $T$ is a \gjt
	for $Q$), we have $\preds(Q) \subseteq \preds(T,A)$.
	From all of the observations made so far and
	Lemma~\ref{lem:subset-removal}, we obtain:
	\begin{align*}
	& \hypertrip(T, A, \free(Q)) \\
	& = (\hypergraph(T, A), \free(Q), \preds(T, A)) \\ & =
	(\hypergraph(T, A) \cup \hypergraph(Q), \free(Q), \preds(T, A)) \\
	& \rewr^* (\hypergraph(Q),
	\free(Q), \preds(T,A)) \\
	& = (\hypergraph(Q), \free(Q), \preds(Q)) = \hypertrip(Q)
	\end{align*}
	Thus $\hypertrip(T, A, \free(Q)) \rewr^* \hypertrip(Q) \rewr^*
	\trip{I}$. Furthermore, because $(T,N)$ is also canonical with $N
	\subseteq A$ and $\free(Q) \subseteq \var(N)$ we have $\hypertrip(T,
	A, \free(Q)) \rewr^* \hypertrip(T, N, \free(Q))$ by
	Lemma~\ref{lem:progress}. Then, because reduction is confluent
	(Proposition~\ref{prop:confluence}) we obtain that $\hypertrip(T,
	\allowbreak N,\allowbreak \free(Q))$ and $\trip{I}$ can be reduced
	to the same triplet. Because $\trip I$ is in normal form,
	necessarily $\hypertrip(T, N, \free(Q)) \rewr^* \trip I$. Since
	reduction steps can only remove nodes and hyperedges (and never add
	them), $\var(\hypergraph(\trip{I})) \subseteq \var(N)$.
	
	It remains to show that $\trip{J}$ is the empty triplet. Hereto,
	first verify the following. For any hypergraph triplets $\trip U$
	and $\trip V$, if $\trip U \rewr^* \trip V$ then also $\tilde{\trip
		U} \rewr^* \tilde{\trip V}$. From this, $\hypertrip(T, A,
	\free(Q)) \rewr^* \trip{I}$, and the fact that $\hypertrip(T,A,
	\emptyset)$ is the residual of $\hypertrip(T,A, \free(Q))$ we
	conclude $\hypertrip(T,A,\emptyset) \allowbreak \rewr^*
	\tilde{\trip{I}}$. Then, because $\tilde{\trip I} \rewr^* \trip{J}$,
	it follows that $\hypertrip(T, A, \emptyset) \allowbreak \rewr^*
	\trip J$.  Let $r$ be $T$'s root node, which is labeled by
	$\emptyset$ since $T$ in canonical. Then $\{r\}$ is a connex
	subset of $T$. By Lemma~\ref{lem:progress}, $\hypertrip(T, A,
	\emptyset) \rewr^* \hypertrip(T, \{r\}, \emptyset)$. Now observe
	that the hypergraph of $\hypertrip(T, \{r\}, \emptyset)$ is empty,
	and its predicate set is also empty. Therefore, $\hypertrip(T,
	\{r\}, \emptyset)$ is the empty hypergraph triplet. In particular,
	it is in normal form.  But, since $\trip J$ is also in normal form
	and normal forms are unique, $\trip J$ must also be the empty
	triplet.
\end{proof}

\begin{corollary}[Completeness]
	\label{prop:completeness}
	Let $Q$ be a \gcq. Assume that $\hypertrip(Q) \rewr^* \nf{\trip{I}}$
	and $\tilde{\trip{I}} \rewr^* \nf{\trip{J}}$. 
	\begin{enumerate}
		\item If $Q$ is acyclic, then $\trip{J}$ is the empty triplet.
		\item If $Q$ is free-connex acyclic, then $\trip{J}$ is the empty
		triplet and $\var(\hypergraph(\trip{I})) = \free(Q)$.
		\item For every \gjt $T$ of $Q$ and every connex subset $N$ of $T$
		it holds that $\var(\hypergraph(\trip{I})) \subseteq \var(N)$.
	\end{enumerate}
\end{corollary}
\begin{proof}
	(1) Since $Q$ is acyclic, there exists a \gjt $T$ for $Q$. Let $N$
	be the set of all of $T$'s nodes. Then $N$ is a connex subset of $T$
	and $\free(Q) \subseteq \var(N) = \var(Q)$. The result then follows
	from Proposition~\ref{prop:completeness-core}.
	
	(2) Since $Q$ is free-connex acyclic, there exists a \gjt pair
	$(T,N)$ compatible with $Q$. In particular, $\var(N) = \free(Q)$. By
	Proposition~\ref{prop:completeness-core}, $\trip{J}$ is the empty
	triplet, and $\var(\hypergraph(\trip{I})) \subseteq \var(N) =
	\free(Q)$. It remains to show $\free(Q) \subseteq
	\var(\hypergraph(\trip{I}))$. First verify the following: A
	reduction step on a hypergraph triplet $\trip{H}$ never removes any
	variable in $\free(\trip H)$ from $\hypergraph(\trip H)$, nor does
	it 
	modify $\free(\trip H)$. Then, since $\free(\hypertrip(Q)) =
	\free(Q) \subseteq \var(Q) \subseteq \var(\hypergraph(\hypertrip(Q))))$, and
	$\hypertrip(Q) \rewr^* \trip I$ we obtain $\free(Q) \subseteq \allowbreak
	\var(\hypergraph(\trip I))$.
	
	(3) Follows directly from Proposition~\ref{prop:completeness-core}. 
\end{proof}

Theorem~\ref{thm:gyo} follows directly from
Corollaries~\ref{prop:soundness} and
\ref{prop:completeness}. Theorem~\ref{thm:construction-correct}
follows from Theorem~\ref{thm:gyo} and
Proposition~\ref{prop:alg-gjt-soundess}.



\section{Experimental Setup}
\label{sec:experiments}

In this section, we present the setup of our experimental evaluation,
whose results are discussed in Section~\ref{sec:exper-eval}. We first
present our practical implementation of \iedyn, then 
show the queries and update stream used for evaluation, and finally
discuss the competing systems. 

\vspace{-1ex}
\paragraph*{\bf Practical Implementation} We have implemented \\\iedyn as a
query compiler that generates executable code in the Scala programming language. The generated code instantiates a $T$-rep and defines \emph{trigger functions} that are used for maintaining the \dtree{$T$} under updates. Our implementation is basic in the
sense that we use Scala off-the-shelf collection libraries (notably \texttt{MutableTreeMap}) to implement the required
indices. Faster implementations with specialized code for the index structures are certainly possible.

Our implementation supports two modes of operation: \emph{push-based}
and \emph{pull-based}. In both modes, the system maintains the $T$-rep
under updates. In the \emph{push-based mode} the system generates,
on its output stream, the delta result $\delt Q(\db,\upd)$ after each
single-tuple update $\upd$. To do so, it uses a modified version of enumeration (Algorithm~\ref{alg:enumeration-st}) that we
call \emph{delta enumeration}. Similarly to how
Algorithm~\ref{alg:enumeration-st} enumerates $Q(\db)$, delta
enumeration enumerates $\Delta Q(\db,\upd)$ with constant delay (if
$Q$ has at most one inequality per pair of atoms) resp. logarithmic
delay (otherwise). To do so, it uses both (1) the $T$-reduct GMRs
$\rho_n$ and (2) the delta GMRs $\Delta \reduc_n$ that are computed by
Algorithm~\ref{alg:update-processing-st} when processing $u$. In this
case, however, one also needs to index the $\Delta \reduc_n$ similarly
to $\reduc_n$.  In the \emph{pull-based mode}, in contrast, the system
only maintains the $T$-rep under updates but does not generate any
output stream. Nevertheless, at any time a user can call the
enumeration (Algorithm \ref{alg:enumeration-st}) procedure to obtain the current output.

We have described in Section~\ref{sec:gdyn} how \iedyn can process
free-connex acyclic \gcqs under updates. It should be noted that our
implementation also supports the processing of general acyclic \gcqs
that are not necessarily free-connex. This is done using the following simple
strategy. Let $Q$ be acyclic but not free-connex. First, compute a
free-connex acyclic approximation $Q_F$ of $Q$. $Q_F$ can always be
obtained from $Q$ by extending the set of output variables of $Q$. In
the worst case, we need to add all variables, and $Q_F$ becomes the
full join underlying $Q$. Then, use \iedyn to maintain a \dtree{$T$}
for $Q_F$.  When operating in push-based mode, for each update $\upd$,
we use the $T$-representation to delta-enumerate $\delt{Q_F}(\db, \upd)$ and project
each resulting tuple to materialize $\delt Q(\db,\upd)$ in an array.
Subsequently, we copy this array to the output. Note that the
materialization of $\delt Q(\db, \upd)$ here is necessary since the
delta enumeration on $T$ can produce duplicate tuples after
projection.  When operating in pull-based mode, we materialize
$Q(\db)$ in an array, and use delta enumeration of $Q_F$ to maintain
the array under updates. Of course, under this strategy, we require
$\Omega(\size{Q(\db)})$ space in the worst case, just like (H)IVM would, but we avoid
the (partial) materialization of delta queries. Note the
distinction between the two modes: in push-based mode $\delt Q(\db,
\upd)$ is materialized (and discarded once the output is generated),
while in pull-based mode $Q(\db)$ is materialized upon requests.

\paragraph*{\bf Queries and Streams}
In contrast to the setting for equi-join queries where systems can be
compared based on industry-strength benchmarks such as TPC-H and
TPC-DS, there is no established benchmark suite
for inequality-join queries.

We evaluate \iedyn on the \gcq queries listed in
table~\ref{tab:queries}. Here, queries $Q_1$--$Q_6$ are full join
queries (i.e., queries without projections). Among these, $Q_1$,
$Q_3$ and $Q_4$ are cross products with inequality predicates, while
$Q_2$, $Q_5$ and $Q_6$ have at least one equality in addition to the
inequality predicates. Queries $Q_1$ and $Q_2$ are binary join queries,
while $Q_3$--$Q_6$ are multi-way join queries.
Queries $Q_7$--$Q_{12}$ project over the result of queries
$Q_4$--$Q_6$. Among these, $Q_7$--$Q_9$ are free-connex acyclic
while $Q_{10}$--$Q_{12}$ acyclic but not free-connex.

\begin{table}[t]
\small
\begin{tabular}[l]{|p{0.038\textwidth}|p{0.38\textwidth}|}
\hline
 Query & Expression \\\hline 
 $Q_1$ & $R(a,b,c)  \Join S(d,e,f)                   | a<d$  \\\hline 
 $Q_2$ & $R(a,b,c,k)\Join S(d,e,f,k)                 | a<d$\\ \hline
 $Q_3$ & $R(a,b,c)  \Join S(d,e,f)  \Join T(g,h,i)   | a<d \wedge e<g $ \\ \hline
 $Q_4$ & $R(a,b,c)  \Join S(d,e,f)  \Join T(g,h,i)   | a<d \wedge d<g $ \\\hline 
 $Q_5$ & $R(a,b,c,k)\Join S(d,e,f,k)\Join T(g,h,i)   | a<d \wedge d<g$ \\ \hline
 $Q_6$ & $R(a,b,c)  \Join S(d,e,f,k)\Join T(g,h,i,k) |  a<d \wedge d<g$ \\ \hline
 $Q_7$ & $\pi_{a,b,d,e,f,g,h}  (Q_4) $\\\hline 
 $Q_8$ & $\pi_{a,d,e,f,g,h,k}(Q_5)$ \\ \hline
 $Q_9$ & $\pi_{d,e,f,g,h,k}(Q_6)$ \\ \hline
 $Q_{10}$ & $\pi_{b,c,e,f,h,i}  (Q_4) $ \\\hline 
 $Q_{11}$ & $\pi_{b,c,e,f,h,i}  (Q_5)$ \\ \hline 
  $Q_{12}$ & $\pi_{b,c,e,f,h,i}  (Q_6)$ \\ \hline 
\end{tabular}
\centering
\caption{Queries for experimental evaluation.}
\label{tab:queries}
\vspace{-2ex}
\end{table}
We evaluate these queries on streams of
updates where each update consists of a single tuple insertion. The
database is always empty when we start processing the update stream.
We synthetically generate two kinds of update streams:
\emph{randomly-ordered} and \emph{temporally-ordered}
update streams. In \emph{randomly-ordered} update streams, insertions can
occur in any order. In contrast, \emph{temporally-ordered}
update streams guarantee that any attribute that
participates in an inequality in the query has a larger
value than the same attribute in any of the previously inserted tuples. Randomly-ordered update streams are useful for comparing against systems that allow processing
of out-of-order tuples; temporally-ordered update streams are useful
for comparison against systems that assume
events arrive always with increasing timestamp values. Examples of systems that process temporally-ordered streams are automaton-based \cer systems.

A random update stream of size $N$ for a query with $k$ relations is
generated as follows. First, we generate $N/k$ tuples with random
attribute values for each relation. Then, we insert tuples in the
update stream by uniformly and randomly selecting them without
repetitions. This ensures that there are $N/k$ insertions from each
relation in the stream. To utilize the same update stream for
evaluating each system we compare to, each stream is stored in a
file. We choose the values for equality join attributes uniformly at
random from $1$ to $200$, except for the scalability and selectivity
experiments in Section~\ref{sec:exper-eval} where the interval depends
on the stream size.

Temporally-ordered streams are generated similarly, but when a new insertion tuple is chosen, a new value is inserted in the
attributes that are compared through inequalities. This value is
larger than the corresponding values of previously inserted tuples. All attributes hold integer values, except for attributes $c$ and $i$ which contain string values.

\paragraph*{\bf Competitors}

We compare \iedyn with DBToaster
(DBT) \cite{DBLP:journals/vldb/KochAKNNLS14}, Esper (E)~\cite{esper}, SASE
(SE)\cite{DBLP:conf/sigmod/WuDR06,Sase2014,DBLP:conf/sigmod/AgrawalDGI08},
Tesla (T)\cite{DBLP:conf/debs/CugolaM10,DBLP:journals/jss/CugolaM12},
and ZStream (Z)~\cite{DBLP:conf/sigmod/MeiM09} using \emph{memory
  footprint}, \emph{update processing time}, and \emph{enumeration
  delay} as comparison metrics. The competing systems differ in their
mode of operation (push-based vs pull-based) and some of them only support temporally-ordered streams.

\uline{DBToaster} is a state-of-the-art
implementation of \hivm.  It operates in pull-based mode, and can deal
with randomly-ordered update streams. DBToaster is particularly meticulous in that
it materializes only useful views, and therefore it is an interesting
implementation for comparison. DBToaster has been extensively tested
on equi-join queries and has proven to be more efficient than a
commercial database management system, a commercial stream processing
system and an IVM
implementation~\cite{DBLP:journals/vldb/KochAKNNLS14}.  DBToaster
compiles given SQL statements into executable trigger programs in different
programming languages. We compare against those generated in Scala
from the DBToaster Release
2.2\footnote{https://dbtoaster.github.io/}, and it uses
actors\footnote{https://doc.akka.io/docs/akka/2.5/} to generate events
from the input files. During our experiments, however, we have found
that this creates unnecessary memory overhead. For a fair memory-wise comparison, we have
therefore removed these actors.

\uline{Esper} is a \cer engine with a relational model based on Stanford
STREAM~\cite{DBLP:books/sp/16/ArasuBBCDIMSW16}. It is
push-based, and can deal with randomly-ordered update streams. We use the
Java-based open source\footnote{http://www.espertech.com/esper/esper-downloads/} for
our comparisons. Esper processes queries expressed in the Esper event
processing language (EPL).

\uline{SASE} is an automaton-based \cer system. It
operates in push-based mode, and can deal with temporally-ordered
update streams only. We use the publicly available Java-based
implementation of SASE\footnote{https://github.com/haopeng/sase}. This implementation does not support
projections. Furthermore, since SASE requires queries to specify a
match semantics (any match, next match, partition contiguity) but does
not allow combinations of such semantics, we can only express queries
$Q_1$, $Q_2$, and $Q_4$ in SASE. Hence, we compare against SASE for
these queries only.  To be coherent with our semantics, the corresponding SASE expressions use the any match semantics~\cite{DBLP:conf/sigmod/AgrawalDGI08}.

\uline{Tesla/T-Rex} is also an automaton-based
\cer system. It operates in push-based mode only, and supports
temporally-ordered update streams only. We use the publicly available
C-based implementation\footnote{https://github.com/deib-polimi/TRex}. This implementation operates in a
publish-subscribe model where events are published by clients to the
server, known as TRexServer. Clients can subscribe to receive recognized
composite events.  Tesla cannot deal with queries involving inequalities on multiple attributes e.g. $Q_3$, therefore, we do not show results for $Q_3$. Since Tesla works in a decentralized manner, we measure the update processing time by
logging the time at the Tesla TRexServer from the stream start 
until the end.

\uline{ZStream} is a \cer system based on a
relational internal architecture. It operates in push-based mode, and
can deal with temporally-ordered update streams only.  ZStream is not
available publicly. Hence, we have created our own implementation following the lazy evaluation algorithm of ZStream described
in their original paper~\cite{DBLP:conf/sigmod/MeiM09}. This paper does not
describe how to treat projections, and as such we compare against
ZStream only for full join queries $Q_1$--$Q_6$.

Due to space limitations, we omit the query expressions used for
Esper (in EPL), SASE, and Tesla/TRex (rules) in this paper, but they
are available at~\cite{exps}.

\paragraph*{\bf Setup} Our experiments are run on an 8-core
3.07 GHz machine running Ubuntu with GNU/Linux 3.13.0-57-generic. To compile
the different systems or generated trigger programs, we have used GCC
version 4.8.2, Java 1.8.0\_101, and Scala version 2.12.4. Each query is evaluated 10
times to measure update processing delay, and two times to measure
memory footprint. We present the average over those runs. Each time
a query is evaluated, 20 GB of main memory are freshly allocated to
the program. To measure the memory footprint for Scala/Java based
systems, we invoke the JVM system calls every 10 updates and consider the
maximum value. For C/C++ based systems we use the GNU/Linux \emph{time}
command to measure memory usage. Experiments that measure memory
footprint are always run separately of the experiments that measure
processing time.
\section{Experimental Evaluation}
\label{sec:exper-eval}

Before presenting experimental results we make some remarks. First, when we
compare against another system we run \iedyn in the
operation mode supported by the competitor. For push-based
systems we report the time required to
both process the entire update stream, and generate the changes to the output after each update. When comparing
against a pull-based system, the measured time includes only
processing the entire update stream. We later report the speed with which the result can be generated from the underlying representation of the output (a $T$-representation in the case of \iedyn). When comparing against a
system that supports randomly-ordered update streams, we only report
comparisons using streams of this type. We have also
looked at temporally-ordered streams for these systems, but the
throughput of the competing systems is similar (fluctuating between 3\% and 12\%) while that of \iedyn significantly
improves (fluctuating between 35\% and 50\%) because insertions to sorted lists become constant instead of logarithmic. We omit these experiments due to lack of space.

It is also important to remark that some executions of the competing systems failed either because they
required more than 20GB of main memory or they took more than 1500
seconds. If an execution requires more than 20GB, we report the processing time elapsed until the exception was
raised. If an execution is still running after 1500 seconds, we stop it and report its maximum memory usage while running.

\begin{figure*}[t]
\centering
\includegraphics[width=1\textwidth]{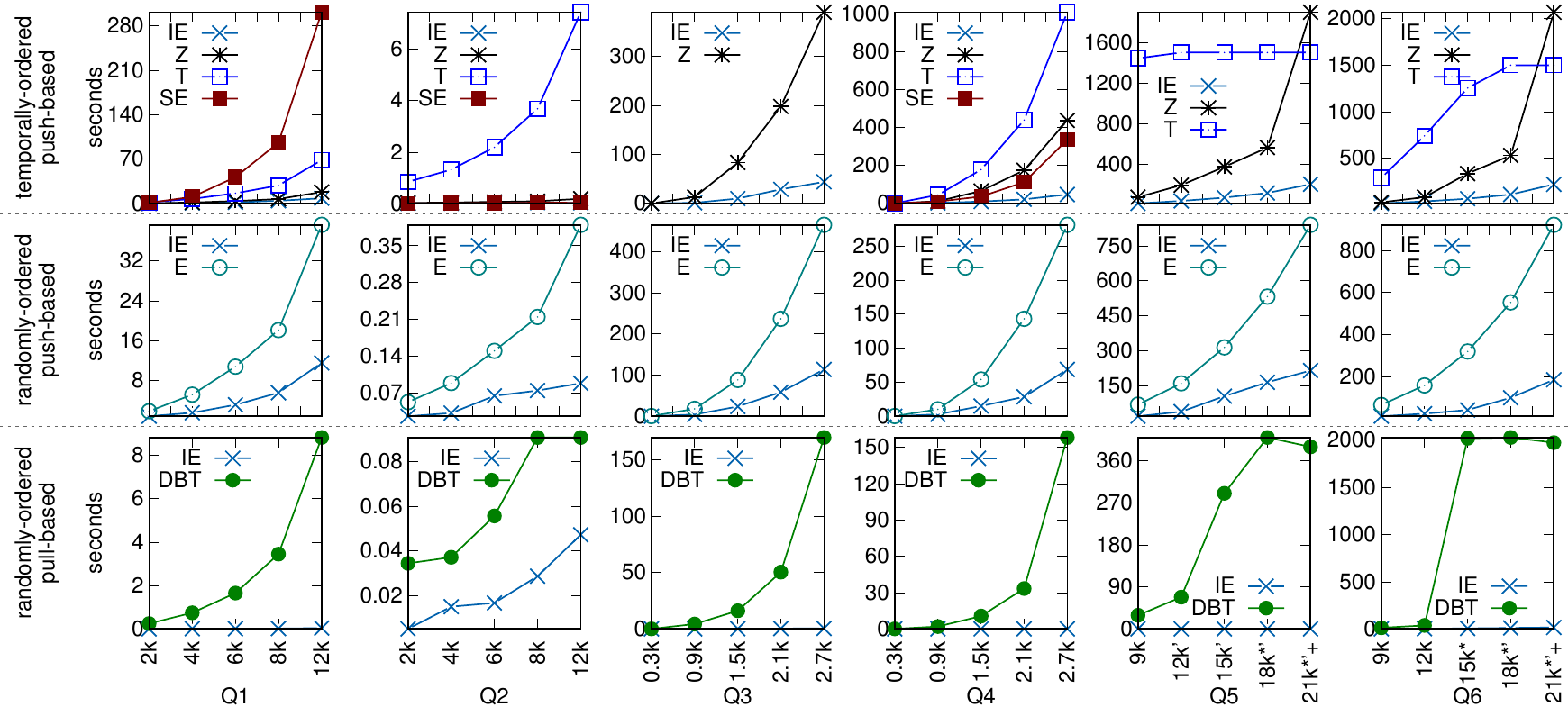}
\caption{\iedyn(IE) VS (\zstream,\dbt,\esper,\tesla, \sase) on full join queries. The X-axis shows stream sizes and the y-axis update delay in $seconds$ (*: \dbt out of memory, +: \zstream out of memory, $'$: \tesla was stopped after 1500 seconds)}
\label{fig:absfigall}
\end{figure*}

\paragraph*{\bf Full join queries} Figure~\ref{fig:absfigall} compares
the update processing time of \iedyn against the competing
systems for full join queries $Q_1$--$Q_6$. We have grouped
experiments that are run under comparable circumstances: in the top
row experiments are conducted for push-based systems on
temporally-ordered update streams (\sase, \tesla, \zstream); in the
second row push-based systems on randomly-ordered update streams (\esper),
and in the bottom row pull-based systems on randomly-ordered update
streams (\dbt). We observe that all of the competing systems have
large processing times even for very small update stream sizes, and
that for some systems execution even failed.
All of these behaviors are due to the low
selectivity of joins on this dataset. Table~\ref{tab:output-sizes} shows the output size of each query for the largest
stream sizes reported in Figure~\ref{fig:absfigall}. We report on
streams that generate outputs of different sizes below.

\begin{table}[t]
\centering
  \label{tab:output-sizes}\scriptsize
\begin{tabular}{|l|r|r|}
    \hline
    Query  & $\left|\text{Stream}\right|$ & $\left|\text{Output}\right|$ \\\hline
    $Q_1$  & 12k &  18,017k  \\ \hline
    $Q_2$  & 12k &       3.8k \\ \hline
    $Q_3$  &  2.7k & 178,847k \\ \hline
    $Q_4$  &  2.7k &  90,425k\\ \hline
    $Q_5$  & 21k & 411,669k \\ \hline
    $Q_6$  & 21k & 297,873k \\ \hline
    $Q_7$  &  2.7k & 114,561k  \\ \hline
    $Q_8$  & 21k & 411,669k \\ \hline
    $Q_9$  & 21k &  99,043k \\ \hline
    $Q_{10}$ &2.7k & 114,561k \\ \hline
    $Q_{11}$&21k & 294,139k\\ \hline
    $Q_{12}$&21k & 297,873k\\ \hline
  \end{tabular}
 \caption{Maximum output sizes per query, k=1000.}
\end{table} 

Figure~\ref{fig:absfigall} is complemented by
Figures~\ref{fig:full-esper-dbt-t} and~\ref{fig:iesez} where
we plot the processing time and memory footprint used by
\iedyn as a percentage of the corresponding usage in the competing
systems. Both, \sase and \zstream support temporally ordered streams, however, \sase supports only queries $Q_1$, $Q_2$, and $Q_4$ and \zstream supports $Q_1$--$Q_6$, therefore in Figure~\ref{fig:iesez} we show \sase (right) and \zstream (left). Note that \iedyn significantly
outperforms the competing systems on all full join
queries. Specifically, it outperforms \dbt up to one order of
magnitude in processing time and up to two orders of magnitude in
memory footprint. It outperforms \tesla up to two orders of magnitude
in processing time, and more than one order of magnitude in memory
footprint.
Moreover, for these queries, even in push-based mode \iedyn can
support the enumeration of query results from its data structures at
any time while competing push-based systems have no such
support. Hence, \iedyn is not only more efficient but also provides
more functionality.

\begin{figure*}[!t]
\centering
\includegraphics[width=1\textwidth]{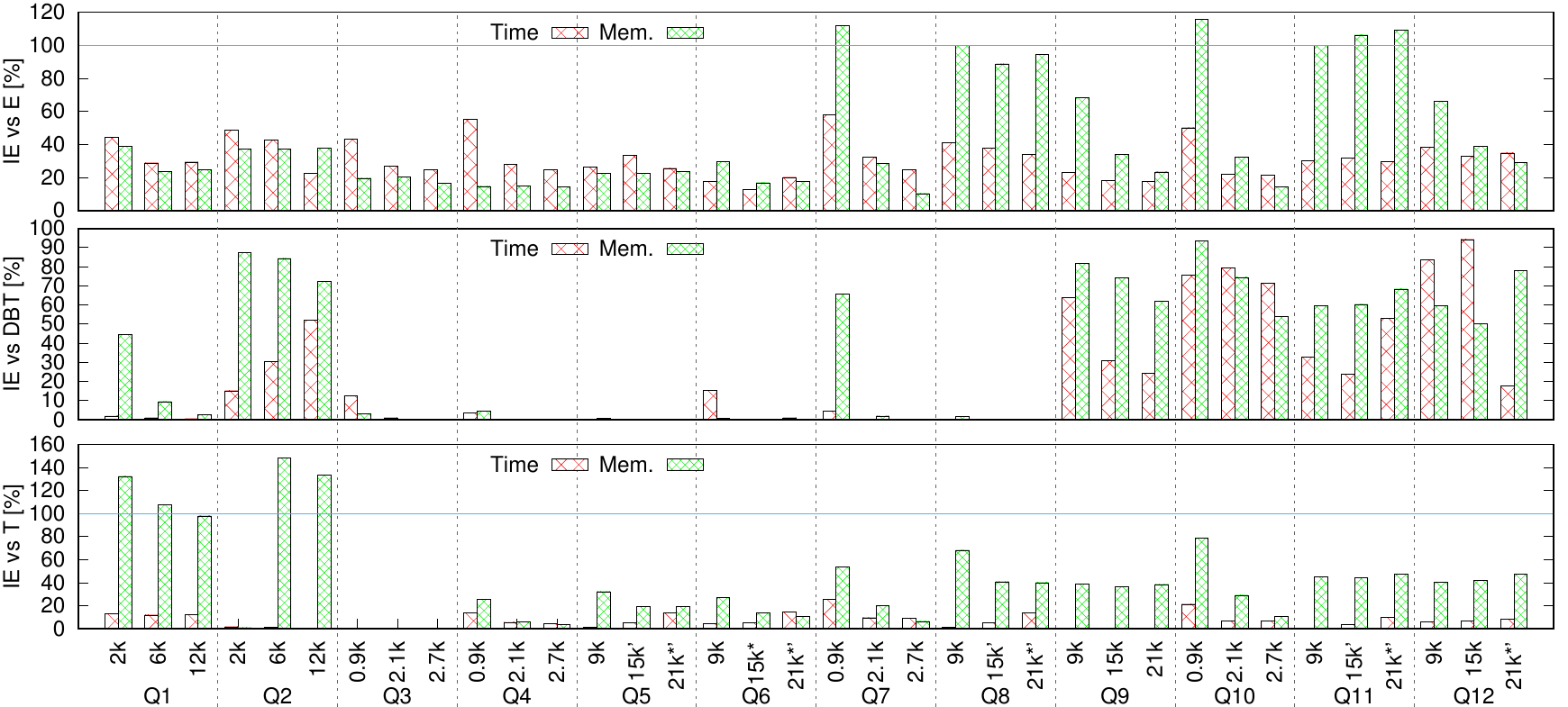}
\caption{\iedyn(IE) VS (\esper, \dbt, \tesla) fulljoin and projection queries, (*: \dbt ran out of memory, $'$: \tesla was stopped after 1500 seconds)}
\label{fig:full-esper-dbt-t}
\vspace*{3mm}
\centering
\includegraphics[width=1\textwidth]{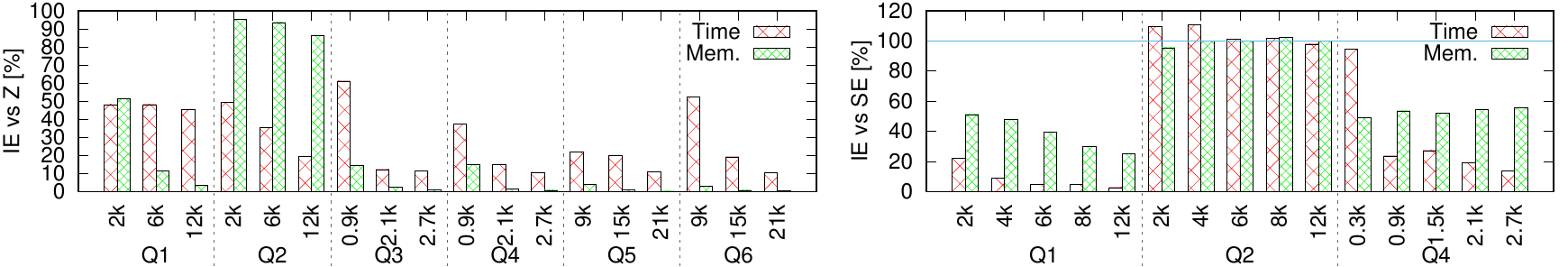}
\caption{\iedyn(IE) VS \sase and \zstream on temporally ordered datasets}
\label{fig:iesez}
\vspace*{3mm}
\begin{minipage}[t]{\columnwidth}
\centering
\includegraphics[width=1\columnwidth]{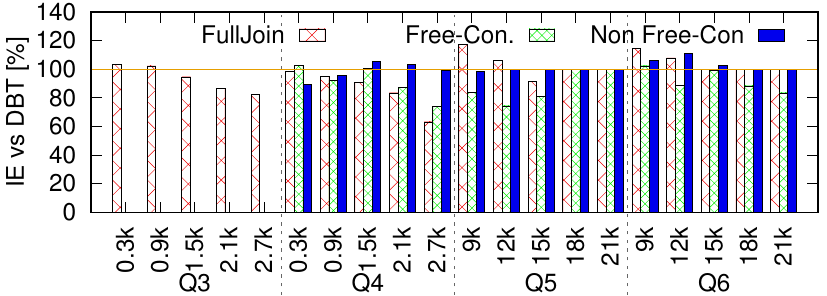}
\caption{Enumeration of query results: \iedyn vs \dbt, different bars for $Q_4,Q_5,Q_6$ show their projected versions}
\label{fig:enum2}
\vspace*{3mm}
\end{minipage}\hfill
\begin{minipage}[t]{\columnwidth}
\centering
\includegraphics[width=1\columnwidth]{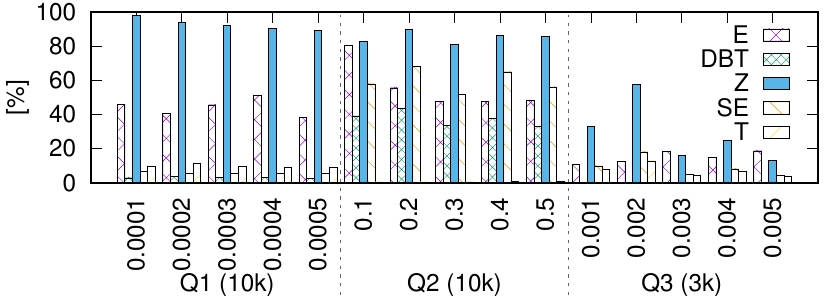}
\caption{\iedyn as percentage of (\esper, \dbt, \sase, \zstream, \tesla) for higher join selectivities. X-axis shows queries with tuples per relation and selectivities}
\label{fig:selectivity}
\end{minipage}
\vspace*{3mm}
\centering
\includegraphics[width=1\textwidth]{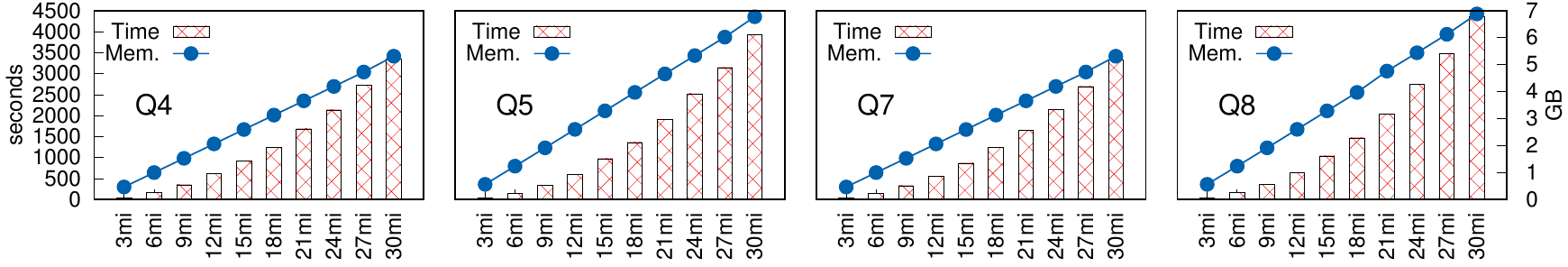}
\caption{\iedyn scalability ($mi = 1,000,000$)}
\label{fig:scalability}
\end{figure*}

\paragraph*{\bf Projections}
Results in Figure~\ref{fig:full-esper-dbt-t} show that \iedyn significantly outperforms both \esper and \dbt on free-connex queries
$Q_7$--$Q_9$: two orders of magnitude improvement over
the throughput of \tesla and more than twofold improvement over that of \esper. Memory usage is also significantly less: one order of magnitude over \esper on the larger
datasets for $Q_7$, and a consistent twofold improvement over
$T$. Similarly, \iedyn outperforms \dbt on free-connex
queries $Q_7$ and $Q_8$ in time and memory by one and two orders of
magnitude, respectively. 

For non-free-connex queries $Q_{10}$--$Q_{12}$, \iedyn continues to outperform \esper, \tesla, and \dbt in terms of processing time. In memory footprint \iedyn outperforms \esper for $Q_{10}$ and $Q_{12}$. Compared to \dbt, \iedyn still improves on memory footprint on non-free-connex queries, though less significantly. In contrast, \iedyn largely improves memory usage over \tesla on larger datasets, even on non-free-connex queries.

\paragraph*{\bf Result enumeration} We know from
Section~\ref{sec:gdyn} that $T$-reps maintained by \iedyn feature
constant delay enumeration (\cde). This theoretical notion, however,
hides a constant factor that could decrease performance in practice
when compared to full materialization. 
In Figure~\ref{fig:enum2}, we show the practical application of \cde
in \iedyn and compare against \dbt which materializes the full query
results. We plot the time required to enumerate the result from
\iedyn's $T$-rep as a fraction of the time required to enumerate the
result from \dbt's materialized views. As can be seen from the figure,
both enumeration times are comparable on average.

Note that we do not compare enumeration time for push-based systems, since for these systems the time required for delta enumeration is already included in the update processing time reported in Figures~\ref{fig:absfigall},~\ref{fig:full-esper-dbt-t} (bottom), and~\ref{fig:iesez}.

\paragraph*{\bf Selective inequality joins} We execute \iedyn over unifromly distributed datasets.
In this case, the inequality joins yield large query results. 
One could argue that this might not be realistic. To
address this problem, we generated datasets
with probability distributions  that are parametrized by a \emph{selectivity}
$s$, such that the expected number of output tuples is $s$ percent of
the cartesian product of all relations in the query. 

Our results show that \iedyn not only outperforms existing systems on less selective inequality joins; we also perform better on very selective inequality joins consistently (see Figure~\ref{fig:selectivity}).  For super selective inequality joins the measurements come similar to what we observe for equality joins, which we investigated in detail in our previous work on equality joins~\cite{dyn:2017}.

\paragraph*{\bf Scalability} To show that \iedyn performs consistently on streams of different sizes, we report the processing delay and the memory footprint each time a $10\%$ of the stream is processed in Figure~\ref{fig:scalability}. These results show that \iedyn has linearly increasing memory footprint as well as update delay as the stream size advances. We show results for queries $Q_4$, $Q_5$, $Q_7$, and $Q_8$ only due to space constraints.


\bibliographystyle{abbrv}
\bibliography{ms}

\begin{appendix}

\section{Proofs of Section~3}
\label{sec:app-acyclicity}

\lemmaRemOne*
\begin{proof}
  The lemma follows from the following observations.
(1) It is straightforward to observe that $T'$ is a valid GJT: the
    construction has left the set of leaf nodes untouched; took care
    to ensure that all nodes (including the newly added node $p$)
    continue to have a guard child; ensures that the connectedness
    condition continues to hold also for the relocated children of $n$
    because every variable in $n$ is present on the entire path
    between $n$ and $p$; and have ensured that also edge labels remain
    valid (for the relocated nodes this is because $\var(p) = \var(g)
    \subseteq \var(n)$).

(2) $N'$ is a connex subset of $T'$ because the subtree of $T$
    induced by $N$ equals to subtree of $T'$ induced by $N'$, modulo
    the replacement of $l$ by $p$ in case that $l$ was in $N$ and $p$
    is hence in $N'$.

(3) $(T,N)$ is equivalent to $(T', N')$ because the construction
    leaves leaf atoms untouched, preserves edge labels, and $\var(N) =
    \var(N')$. The latter is clear if $l \not \in N$ because then $N =
    N'$. It follows from the fact that $\var(l) = \var(p)$ if $l \in
    N$, in which case $N'=N\setminus \{l\} \cup \{p\}$.

(4) All nodes in $\child_T(n) \setminus N$ (and their descendants)
    are relocated to $p$ in $T'$. Therefore, $n$ is no longer a violator
    in $(T', N')$. Because we do not introduce new violators, the
    number of violators of $(T', N')$ is strictly smaller than the number
    of violators of $(T, N)$.\qedhere
\end{proof}

\lemmaRemTwo*
\begin{proof}
  The lemma follows from the following observations.  (1) It is
  straightforward to observe that $T'$ is a valid GJT: the
  construction has left the set of leaf nodes untouched; took care to
  ensure that all nodes (including the newly added node $p$) continue
  to have a guard child; ensures that the connectedness condition
  continues to hold also for the relocated children of $n$ because
  every variable in $n$ is also present in $p$, their new parent; and
  have ensured that also edge labels remain valid (for the relocated
  nodes this is because $\var(p) = \var(n)$).

(2) $N'$ is a connex subset of $T'$ because (i) the subtree of $T$
    induced by $N$ equals to subtree of $T'$ induced by $N' \ \{p\}$,
    (ii) $n \in N$, and (iii) $p$ is a child of $n$ in $T'$. Therefore,
    $N'$ must be connex.

(3) $(T,N)$ is equivalent to $(T', N')$ because the construction
    leaves leaf atoms untouched, preserves edge labels, and $\var(N) =
    \var(N')$. The latter follows because $\var(N') = \var(N \cup
    \{p\})$ and because  $\var(p) = \var(n) \subseteq \var(N)$ since
    $n \in N$.

(4) All nodes in $\child_T(n) \setminus N$ (and their descendants)
    are relocated to $p$ in $T'$. Therefore, $n$ is no longer a violator
    in $(T', N')$. Because we do not introduce new violators, the
    number of violators of $(T', N')$ is strictly smaller than the number
    of violators of $(T, N)$.\qedhere
\end{proof}


\section{Proofs of Section~5}
\label{sec:app-algorithm}

\subtreeQuery*
\begin{proof}
    We proceed by induction on the number of descendants of $n$. If
    $n$ has no descendant then $Q_n$ is a single atom $r(\seq{x})$, so
    we have $\seq{x} = \free(Q_n)=\var(n)$. Then $\pi_{var(n)}Q_n(\db)=Q_n(\db)=\db_{r(\seq{x})}=\rho_n$, concluding the basic case. Now, for the inductive case we distinguish whether $n$ has one or two children.

    Assume $n$ has a single child $c$ and $Q_c=({\mathcal R}\mid
    \Theta)$.  Then, by definition we have $Q_n=({\mathcal R}\mid
    \Theta\cup\pred(n)).$ Therefore
    $Q_n(\db)=\sigma_{\pred(n)}Q_c(\db)$, which implies that
    $\pi_{\var(n)}Q_n(\db)=\pi_{\var(n)}\sigma_{\pred(n)}Q_c(\db)$. Since
    $\pred(n)$ only mentions variables in $\var(c) \cup \var(n)$ and
    $\var(n)\subseteq\var(c)$, as $c$ is a guard of $n$, this is
    equivalent to
    \[ \pi_{\var(n)}Q_n(\db) =
    \pi_{\var(n)}\sigma_{\pred(n)}\pi_{\var(c)}Q_c(\db).\] By
    induction, this equals $\pi_{\var(n)}\sigma_{\pred(n)}\rho_c =
    \rho_n$, showing that $\pi_{\var(n)}Q_n(\db)=\rho_n$.

    Assume now that $n$ has two children $c_1$ and $c_2$, and that
    $Q_{c_i}=\left({\mathcal R}_i\mid \Theta_i\right)$ for
    $i\in\{1,2\}$. We assume w.l.o.g. that $c_1$ is a guard for
    $n$. First, note that by definition $Q_n=\left({\mathcal
        R}_1\Join{\mathcal
        R}_2\mid\Theta_1\cup\Theta_2\cup\pred(n)\right),$ and then we
    have
    $Q_n(\db)=\sigma_{\pred(n)}\sigma_{\Theta_1}\sigma_{\Theta_2}\left({\mathcal
        R}_1\Join{\mathcal R}_2\right)(\db).$ Since $\Theta_i$ only
    mentions variables of atoms in $\mathcal{R}_i$ (for
    $i\in\{1,2\}$), we can push the selections and obtain
    \begin{align*}
      Q_n(\db) & = \sigma_{\pred(n)}\left(\sigma_{\Theta_1}{\mathcal
          R}_1\Join\sigma_{\Theta_2}{\mathcal R}_2\right)(\db) \\ 
      & = \sigma_{\pred(n)}\left(\sigma_{\Theta_1}{\mathcal
          R}_1(\db)\Join\sigma_{\Theta_2}{\mathcal R}_2(\db)\right) \\
      & =\sigma_{\pred(n)}\left(Q_{c_1}(\db)\Join Q_{c_2}(\db)\right)
    \end{align*}
    Therefore,
    $$\pi_{\var(n)}Q_n(\db)=\pi_{\var(n)}\sigma_{\pred(n)}\left(Q_{c_1}(\db)\Join Q_{c_2}(\db)\right).$$
    Since $\var(\pred(n))\subseteq\var(c_1)\cup\var(c_2)\cup \var(n)$
    and $\var(n)\subseteq\var(c_1)$ we have
    $\var(\pred(n))\subseteq\var(c_1)\cup\var(c_2)$. This, combined
    with the fact that, due to the connectedness property of $T$ we,
    have $\var(Q_{c_1})\cap \allowbreak \var(Q_{c_2}) \allowbreak \subseteq\var(c_i)$ for
    $i\in\{1,2\}$, we can add the following projections
    $$=\pi_{\var(n)}\sigma_{\pred(n)}\left(\pi_{\var(c_1)}Q_{c_1}(\db)\Join \pi_{\var(c_2)}Q_{c_2}(\db)\right).$$
    Then, by induction hypothesis we have
    $$\pi_{\var(n)}Q_n(\db)=\pi_{\var(n)}\sigma_{\pred(n)}\left(\rho_{c_1}\Join\rho_{c_2}\right)=\rho_n,$$
    concluding our proof.
\end{proof}

\enumSubRoutine*
\begin{proof}
  Within the proof, we abuse notation and allow for projections over
  supersets of variables. For example, if $\var(Q)\subseteq \seq{x}$
  then $\pi_{\seq{x}}Q=\pi_{\seq{x}\cap\var(Q)}Q$.

  Let $n \in N$ and $\tup{t}\in\rho_n$.  We proceed by induction on the number of
  nodes in $N\cap T_n$. If $N\cap T_n=\{n\}$, we have
  $\var(N)\cap\var(Q_n)=\var(n)$ and therefore
  $\pi_{\var(N)}Q_n(\db)=\pi_{\var(n)}Q_n(\db)$. Then, by
  Lemma~\ref{lem:subtree-query} we have
  $\pi_{\var(N)}Q_n(\db)=\rho_n$. As $\tup{t}\in\rho_n$, this implies
  that the only tuple in $\pi_{\var(N)}Q_n(\db)$ that is compatible
  with $\tup{t}$ is $\tup{t}$ itself. As $n$ is in the frontier of
  $N$, $\routenum_{T,N}(n, \tup{t},\rho)$ will enumerate precisely
  $\{(\tup{t}, \rho_n(\tup{t}))\}$ (Line~\ref{line:frontiertuple}),
  which concludes the base case.
    
    For the inductive step we need to consider two cases depending on the number of children of $n$.

    Case (1).  If $n$ has a single child $c$ then necessarily $c$ is
    a guard of $n$, i.e., $\var(n) \subseteq \var(c)$.  In this case,
    Algorithm~\ref{alg:enumeration-st} will call $\routenum_{T, N}(c,
    \tup{s}, \rho)$ for each tuple $\tup{s}\in
    \rho_c\semijoin_{\pred(n)} \tup{t}$. By induction hypothesis and
    Lemma~\ref{lem:subtree-query}, this will correctly enumerate every
    tuple in $\pi_{\var(N)}Q_c(\db)\semijoin\tup{s}$ for every 
    $\tup{s}$ in $\sigma_{\pred(n)}(\pi_{\var(c)}Q_c(\db)\semijoin \allowbreak
    \tup{t})$. Therefore, this enumerates the set
        $$\pi_{\var(N)}Q_c(\db)\semijoin\sigma_{\pred(n)}(\pi_{\var(c)}Q_c(\db)\semijoin\tup{t}).$$
        As $\var(\pred(n))\subseteq\var(c) \cup\var
(n) = \var(c) \subseteq\var(Q_c)$, we can pull out the projection and selection
        $$=\pi_{\var(N)}\sigma_{\pred(n)}(Q_c(\db)\semijoin(\pi_{\var(c)}Q_c(\db)\semijoin\tup{t})).$$
        Because the variables in $\tup{t}$ are a subset of $\var(c)$, this is the same as
        $\pi_{\var(N)}\sigma_{\pred(n)}(Q_c(\db)\semijoin\tup{t}).$
        Finally, we push the selection and projection inside and obtain
        $$=\pi_{\var(N)}\sigma_{\pred(n)}Q_c(\db)\semijoin\tup{t}=\pi_{\var(N)}Q_n(\db)\semijoin\tup{t}.$$

        Case (2). Otherwise, $n$ has two children $c_1$ and $c_2$. Since $|N\cap T_n|>1$ and $N$ is sibling closed we have $\{c_1,c_2\}\subset N$. In this case,
Algorithm~\ref{alg:enumeration-st} will first enumerate $\tup{t_i} \in
\rho_{c_i}\semijoin_{\pred(n\rightarrow c_1)}\tup{t}$ for $i\in\{1, 2\}$. By Lemma~\ref{lem:subtree-query} this is equivalent to enumerate every $\tup{t_i}$ in $\sigma_{\pred(n\rightarrow c_i)}\pi_{\var(c_i)}Q_{c_i}(\db)\semijoin \tup{t}$. Then, for each such $\tup{t_i}$ the algorithm will enumerate every pair $(\tup{s_i}, \mu_i)$ generated by $\routenum_{T, N}(c_i, \tup{t_i}, \rho)$, which by induction is the same as enumerating every $(\tup{s_i}, \mu_i)$ in $\pi_{\var(N)}Q_{c_i}(\db)\semijoin \tup{t_i}$. Therefore the algorithm is enumerating
        $$\pi_{\var(N)}Q_{c_i}(\db)\semijoin(\sigma_{\pred(n\rightarrow c_i)}\pi_{\var(c_i)}Q_{c_i}(\db)\semijoin \tup{t})$$
        By the same reasoning as in the previous case, this is equivalent to enumerating every $(\tup{s_i}, \mu_i)$ in
        $$\sigma_{\pred(n\rightarrow c_i)}\pi_{\var(N)}Q_{c_i}(\db)\semijoin \tup{t}.$$
        From the connectedness property of $T$, it follows that
        $\var(Q_{c_1})\cap\var(Q_{c_2}) \subseteq \var(n)$. Thus,
        $\var(Q_{c_1})\cap\var(Q_{c_2})$ is a subset of the variables
        of $\tup{t}$. Hence, every tuple $\tup{s_1}$ will be compatible with
        every tuple $\tup{s_2}$, and the enumeration of every pair
        $(\tup{s_1}\cup \tup{s_2},\mu_1\times\mu_2)$ is the same as
        the enumeration of
        \begin{multline*}
            \left[\sigma_{\pred(n\rightarrow c_1)}\pi_{\var(N)}Q_{c_1}(\db)\semijoin \tup{t}\right]\Join\\
            \left[\sigma_{\pred(n\rightarrow c_2)}\pi_{\var(N)}Q_{c_2}(\db)\semijoin \tup{t}\right].
        \end{multline*}
        We can now push the projections and selections outside and obtain
        \begin{multline*}
            =\pi_{\var(N)}\sigma_{\pred(n\rightarrow c_1)}\sigma_{\pred(n\rightarrow c_2)}\\
            [(Q_{c_1}(\db)\semijoin \tup{t})\Join (Q_{c_2}(\db)\semijoin \tup{t})]
        \end{multline*}
        Since $\pred(n)=\pred(n\rightarrow c_1)\cup\pred(n\rightarrow c_2)$ and the variables in $\var(Q_{c_1})\cap\var(Q_{c_2})$ are contained in the variables of $\tup{t}$, we have
        \begin{align*}
        & =\pi_{\var(N)}\sigma_{\pred(n)}[(Q_{c_1}(\db)\Join Q_{c_2}(\db))\semijoin \tup{t}]\\
        & =\pi_{\var(N)}[\sigma_{\pred(n)}(Q_{c_1}(\db)\Join Q_{c_2}(\db))]\semijoin \tup{t}\\
        & =\pi_{\var(N)}Q_n(\db)\semijoin \tup{t}  \qedhere
        \end{align*}
\end{proof}


\section{Proofs of Section~6}
\label{sec:app-gyo}

\subsection{Proof of Proposition~\ref{prop:confluence}}
\label{sec:proof-prop-confluence}

Because no infinite sequences of reduction steps are possible, it
suffices to demonstrate local confluence: 
\begin{proposition}
  \label{prop:local-confluence}
    If $\trip{H} \rewr \trip{I}_1$ and $\trip{H} \rewr \trip{I}_2$
  then there exists $\trip{J}$ such that both $\trip{I}_1 \rewr^* \trip{J}$
  and $\trip{I}_2 \rewr^* \trip{J}$. 
\end{proposition}
Indeed, it is a standard result in the theory of rewriting systems
that confluence (Lemma~\ref{prop:confluence}) and local confluence 
(lemma~\ref{prop:local-confluence}) coincide when infinite sequences of
reductions steps are impossible~\cite{baader}.

Before proving Lemma~\ref{prop:local-confluence}, we observe that the
property of being isolated or being a conditional subset is preserved
under reductions, in the following sense.

\begin{lemma}
  \label{lem:preservation}
  Assume that $\trip{H} \rewr \trip{I}$. Then $\preds(\trip I)
  \subseteq \preds(\trip H)$ and for every hyperedge $e$ we have
  $\ext_{\trip I}(e) \subseteq \ext_{\trip H}(e)$,
  $\equijoinvars_{\trip I}(e) \subseteq \equijoinvars_{\trip H}(e)$,
  and $\isolated_{\trip{H}}(e) \subseteq
  \isolated_{\trip{I}}(e)$. Furthermore, if $e \cse_{\trip H} f$  then also $e \cse_{\trip I} f$.
\end{lemma}
\begin{proof}
  First observe that $\preds(\trip I) \subseteq \preds(\trip H)$,
  since reduction operators only remove predicates.  This implies that
  $\ext_{\trip I}(e) \subseteq \ext_{\trip H}(e)$ for every hyperedge
  $e$. Furthermore, because reduction operators only remove hyperedges
  and never add them, it is easy to see that $\equijoinvars_{\trip
    H}(e) \subseteq \equijoinvars_{\trip I}(e)$. Hence, if $x \in
  \isolated_{\trip H}(e)$ then $x \not \in \equijoinvars_{\trip H}(e)
  \supseteq \equijoinvars_{\trip I}(e)$ and $x \not \in
  \var(\preds(\trip H)) \supseteq \var(\preds(\trip I))$. Therefore, $x
  \in \isolated_{\trip I}(e)$. As such, $\isolated_{\trip I}(e)
  \subseteq \isolated_{\trip H}(e)$.

  Next, assume that $e \cse_{\trip H} f$.  We need to show that
  $\equijoinvars_{\trip{I}}(e) \subseteq f$ and $\ext_{\trip{I}}(e
  \setminus f) \subseteq f$. The first condition follows since
  $\equijoinvars_{\trip{I}}(e) \subseteq \equijoinvars_{\trip H}(e)
  \subseteq f$ where the last inclusion is due to $e \cse_{\trip H}
  f$. The second also follows since $\ext_{\trip{I}}(e \setminus f)
  \subseteq \ext_{\trip{H}}(e \setminus f) \subseteq f$ where the last
  inclusion is due to $e \cse_{\trip H} f$.
\end{proof}

\begin{proof}[Proof of Proposition~{\protect\ref{prop:local-confluence}}]
  If $\trip{I}_1 = \trip{I}_2$ then it suffices to take $\trip{J}
  =\trip{I}_1 = \trip{I}_2$. Therefore, assume in the following that
  $\trip{I}_1 \not = \trip{I}_2$. Then, necessarily $\trip{I}_1$ and
  $\trip{I}_2$ are obtained by applying two different reduction
  operations on $\trip{H}$. We make a case analysis on the types of
  reductions applied.

  \emph{\uline{(1) Case (ISO, ISO)}}: assume that $\trip{I}_1$ is obtained by
  removing the non-empty set $X_1 \subseteq \isolated_{\trip H}(e_1)$
  from hyperedge $e_1$, while $\trip I_2$ is obtained by removing
  non-empty $X_2 \subseteq \isolated_{\trip H}(e_2)$ from $e_2$ with
  $X_1 \not = X_2$. There are two possibilities.

  (1a) $e_1 \not = e_2$. Then $e_2$ is still a hyperedge in $\trip I_2
  $ and $e_1$ is still a hyperedge in $\trip I_1$. By
  Lemma~\ref{lem:preservation}, $\isolated_{\trip H}(e_1) \subseteq
  \isolated_{\trip I_2}(e_1)$ and $\isolated_{\trip H}(e_2) \subseteq
  \isolated_{\trip I_1}(e_2)$. Therefore, we can still remove $X_2$
  from $\trip{I}_1$ by means of rule ISO, and similarly remove $X_1$
  from $\trip{I}_2$. Let $\trip{J}_1$ (resp. $\trip{J}_2$) be the
  result of removing $X_2$ from $\trip{I}_1$
  (resp. $\trip{I}_2$). Then $\trip{J}_1 = \trip{J}_2$ (and hence
  equals triplet $\trip{J}$):
  \begin{align*}
    \hypergraph(\trip{J}_1) & = \hypergraph(\trip{H}) \setminus
    \{e_1,e_2\} \cup \{ e_1 \setminus X_1 \mid e_1 \setminus X_1 \not
    = \emptyset\} \\ & \quad \cup \{ e_2 \setminus
    X_2 \mid e_2 \setminus X_2 \not = \emptyset\}  \\
    & = \hypergraph(\trip{J}_2) \\
    \preds(\trip{J}_1) & = \preds(\trip H) = \preds(\trip J_2)
  \end{align*}

  (1b) $e_1 = e_2$. We show that $X_2 \setminus X_1 \subseteq
  \isolated_{\trip I_1}(e_1 \setminus X_1)$ and similarly $X_1
  \setminus X_2 \subseteq \isolated_{\trip I_1}(e_2 \setminus
  X_1)$. This suffices because we can then apply ISO to remove $X_2
  \setminus X_1$ from $\trip I_1$ and $X_1 \setminus X_2$ from $\trip
  I_2$. In both cases, we reach the same triplet as removing $X_1 \cup X_2
  \subseteq \isolated_{\trip H}(e_1)$ from $\trip H$.\footnote{Should
    $X_2 \setminus X_1$ be empty, we don't actually need to do
    anything on $\trip I_1$: $X_1 \cup X_2$ is already removed from
    it. A similar remark holds for $\trip I_2$ when $X_1 \setminus
    X_2$ is empty.}
  
  To see that $X_2 \setminus X_1 \subseteq \isolated_{\trip I_1}(e_1
  \setminus X_1)$, let $x \in X_2 \setminus X_1$. We need to show $x
  \not \in \equijoinvars_{\trip I_1}(e_1 \setminus X_1)$ and
  $x \not \in \var(\preds(\trip I_1))$. Because $x \in
  X_2 \subseteq \isolated_{\trip H}(e_1)$ we know $x \not \in
  \equijoinvars_{\trip H}(e_1)$. Then, since $x \not \in X_1$, also $x
  \not \in \equijoinvars_{\trip H}(e_1 \setminus X_1)$. By
  Lemma~\ref{lem:preservation}, $\equijoinvars_{\trip I_1}(e_1
  \setminus X_1) \subseteq \equijoinvars_{\trip H}(e_1 \setminus
  X_1)$. Therefore, $x \not \in \equijoinvars_{\trip I_1}(e_1
  \setminus X_1)$. Furthermore, because $x
  \in \isolated_{\trip H}(e_1)$ we know $x \not \in \var(\preds(\trip
  H))$. Since $\var(\preds(\trip I_1)) \subseteq \var(\preds(\trip
  H))$ by Lemma~\ref{lem:preservation}, also $x\ not \in
  \var(\preds(\trip I_1))$. 
  
  $X_1
  \setminus X_2 \subseteq \isolated_{\trip I_1}(e_2 \setminus
  X_1)$ is shown similarly.

  \medskip
  \emph{\uline{(2) Case (CSE, CSE)}}: assume that $\trip{I}_1$ is obtained by
  removing hyperedge $e_1$ because it is a conditional subset of
  hyperedge $f_1$, while $\trip{I}_2$ is obtained by removing $e_2$,
  conditional subset of $f_2$. Since $\trip{I}_1 \not = \trip{I}_2$ it
  must be $e_1 \not = e_2$. We need to further
  distinguish the following cases.

  (2a) $e_1 \not = f_2$ and $e_2 \not = f_1$.  In this case, $e_2$ and
  $f_2$ remain hyperedges in $\trip{I}_1$ while $e_1$ and $f_1$ remain
  hyperedges in $\trip I_2$. Then, by Lemma~\ref{lem:preservation}, $e_2
  \cse_{\trip I_1} f_2$ and $e_1 \cse_{\trip I_2} f_2$.  Let $\trip
  J_1$ (resp. $\trip J_2$) be the triplet obtained by removing $e_2$
  from $\trip I_1$ (resp. $e_1$ from $\trip I_2$). Then $\trip J_1 =
  \trip J_2$ since clearly $\free(\trip J_1) = \free(\trip J_2)$ and 
  \begin{align*}
    \hypergraph(\trip{J}_1) & = \hypergraph(\trip{H}) \setminus
    \{e_1,e_2\}  = \hypergraph(\trip{J}_2) \\
    \preds(\trip{J}_1) & = \{ \theta \in \preds(\trip{H}) \mid
    \var(\theta) \cap (e_1 \setminus f_1)= \emptyset, \\
    & \qquad  \var(\theta) \cap
    (e_2 \setminus f_2) = \emptyset
    \}\\
    & = \preds(\trip{J}_2) 
  \end{align*}
  From this the result follows by taking $\trip J = \trip J_1 = \trip
  J_2$.

  (2b) $e_1 \not = f_2$ but $e_2 = f_1$. Then $e_1 \cse_{\trip{H}}
  e_2$ and $e_2 \cse_{\trip{H}} f_2$ with $f_2 \not = e_1$. It
  suffices to show that $e_1 \cse_{\trip{H}} f_2$ and $e_1 \setminus
  f_2 = e_1 \setminus f_1$, because then (CSE) due to $e_1\cse_{\trip
    H} f_1$ has the same effect as CSE on $e_1 \cse_{\trip H} f_2$,
  and we can apply the reasoning of case (2a) because $e_1 \not = f_2$
  and $e_2 \not = f_2$.

  We first show $e_1 \setminus f_2 = e_1 \setminus f_1$. Let $x \in
  e_1 \setminus f_2$ and suppose for the purpose of contradiction that
  that $x \in e_2 = f_1$. Then, since $e_1 \not = e_2$, $x \in
  \equijoinvars(e_2) \subseteq f_2$ where the last inclusion is due to
  $e_2 \cse_{\trip H} f_2$. Hence, $e_1 \setminus f_2 \subseteq e_1
  \setminus f_1$. Conversely, let $x \in e_1 \setminus f_1$. Since
  $f_1 = e_2$, $x \not \in e_2$. Suppose for the purpose of
  contradiction that $x \in f_2$. Because $e_1 \not = f_2$, $x \in
  \equijoinvars_{\trip H}(e_1) \subseteq e_2$ where the last inclusion
  is due to $e_1 \cse_{\trip H} e_2$. Therefore, $e_2 \setminus f_1 =
  e_1 \setminus f_2$. 

  To show that $e_1 \cse_{\trip H} f_2$, let $x
  \in \equijoinvars_{\trip H}(e_1)$. Because $e_1 \cse_{\trip{H}}
  e_2$, $x \in e_2$. Because $x$ occurs in two distinct hyperedges in
  $\trip{H}$, also $x \in \equijoinvars_{\trip H}(e_2)$. Then, because
  $e_2 \cse_{\trip H} f_2$, $x \in f_2$. Hence $\equijoinvars_{\trip
    H}(e_1) \subseteq f_2$. It remains to show $\ext_{\trip H}(e_1
  \setminus f_2) \subseteq f_2$. To this end, let $x \in \ext_{\trip
    H}(e_1 \setminus f_2)$ and suppose for the purpose of
  contradiction that $x \not \in f_2$. By definition of $\ext$ there
  exists $\theta \in \preds(\trip H)$ and $y \in \var(\theta) \cap
  (e_1 \setminus f_2)$ such that $x \in \var(\theta) \setminus (e_1
  \setminus f_2)$. In particular, $y \not \in f_2$.  Since $e_1
  \setminus f_2 = e_1 \setminus e_2$, $y \in \var(\theta) \cap (e_1
  \setminus e_2) $ and $x \in \var(\theta) \setminus (e_1 \setminus
  e_2)$. Thus, $x \in \ext_{\trip H}(e_1 \setminus e_2)$. Then, since
  $e_1 \cse_{\trip H} e_2$, $x \in e_2$. Thus, $x \in e_2 \setminus
  f_2$ since $x \not \in f_2$. Hence $x \in \var(\theta) \cap (e_2
  \setminus f_2)$. Furthermore, since $y \not \in e_2$ also $y \not
  \in e_2 \setminus f_2$. Hence, $y \in \var(\theta) \setminus (e_2
  \setminus f_2)$. But then $\theta$ shows that $y \in \ext_{\trip
    H}(e_2 \setminus f_2)$. Then, by because $e_2 \cse_{\trip H} f_2$,
  also $y \in f_2$ which yields the desired contradiction.

  (2c) $e_1  = f_2$ but $e_2 \not = f_1$. Similar to case (2b).

  (2d) $e_1 = f_2$ and $e_2 = f_1$. Then
  $e_1 \cse_{\trip{H}} e_2$ and $e_2 \cse{_\trip{H}} e_1$ and $e_1
  \not = e_2$. Let $\trip K_1$ (resp. $\trip K_2$) be the triplet
  obtained by applying (FLT) to remove all $\theta \in \preds(\trip
  I_1)$ (resp. $\theta \in \preds(\trip I_2)$ for which $\var(\theta)
  \subseteq \var(e_2)$ (resp. $(\var(\theta) \subseteq
  \var(e_2)$. Furthermore, let $\trip{J}_1$ (resp. $\trip{J}_2$) be
  the triplet obtained by applying ISO to removing $\isolated_{\trip
    I_1}(e_2)$ from $\trip K_1$ (resp. removing $\isolated_{\trip
    I_2}(e_1)$ from $\trip K_2$). Here, we take $\trip{J}_1 = \trip
  K_1$ if $\isolated_{\trip K_1}(e_2)$ is empty (and similarly for
  $\trip{J}_2$). Then clearly $\trip H \rewr \trip I_1 \rewr^* \trip K_1
  \rewr^* J_1$ and $\trip H \rewr \trip I_2 \rewr^* \trip K_2 \rewr^*
  \trip J_2$. The result then follows by showing that $\trip J_1 =
  \trip J_2$.  Towards this end, first observe that $\free(\trip J_1)
  = \free(\trip K_1) = \free(\trip I_1) = \free(\trip H) = \free(\trip
  I_2) = \free(\trip K_2) = \free(\trip J_2)$. Next, we show that
  $\preds(\trip J_1) = \preds(\trip J_2)$. We first observe that
  $\preds(\trip J_1) = \preds(\trip K_1)$ and $\preds(\trip J_2) =
  \preds(\trip K_2)$ since the ISO operation does not remove
  predicates. Then observe that 
  \begin{align*}
    \preds(\trip K_1)  & = \{ \theta \in \preds(\trip I_1) \mid
    \var(\theta) \not \subseteq \var(e_2) \} \\
    & = \{ \theta \in \preds(\trip H) \mid
    \var(\theta) \cap (e_1\setminus e_2) = \emptyset \text{ and } \\
    & \qquad \qquad \qquad \quad\  \var(\theta) \not \subseteq e_2 \}, \\
    \preds(\trip K_2) & = \{ \theta \in \preds(\trip I_2) \mid
    \var(\theta) \not \subseteq e_1 \} \\
    & = 
 \{ \theta \in \preds(\trip H) \mid
    \var(\theta) \cap (e_2\setminus e_1) = \emptyset \text{ and }\\
    & \qquad \qquad \qquad \quad\  \var(\theta)
    \not \subseteq e_1 \}.
  \end{align*}
  We only show the reasoning for $\preds(\trip K_1) \subseteq
  \preds(\trip K_2)$, the other direction being similar. Let $\theta
  \in \preds(\trip K_1)$. Then $\var(\theta \cap (e_1 \setminus e_2)
  =\emptyset$ and $\var(\theta) \not \subseteq e_2$.  Since
  $\var(\theta) \not \subseteq e_2$ there exists $y \in \var(\theta)
  \setminus e_2$. Then, because $\var(\theta) \cap (e_1 \setminus e_2)
  = \emptyset$, $y \not \in e_1$. Thus, $\var(\theta) \not \subseteq
  e_1$. Now, suppose for the purpose of obtaining a contradiction,
  that $\var(\theta) \cap (e_2 \setminus e_1) \not = \emptyset$. Then
  take $z \in \var(\theta) \cap (e_2 \setminus e_1)$. But then $y \in
  \ext_{\trip H}(e_2 \setminus e_1)$. Hence, $y \in e_1$ because $e_2
  \cse_{\trip H} e_1$, which yields the desired contradiction with $y
  \not \in e_2$. Therefore, $\var(\theta) \cap (e_2 \setminus e_1) =
  \emptyset$, as desired. Hence $\theta \in \preds(\trip K_2)$.

  It remains to show that $\hypergraph(\trip J_1) = \hypergraph(\trip
  J_2)$. To this end, first observe
  \begin{align*}
    \hypergraph(\trip J_1) & = \hypergraph(\trip K_1) \setminus
    \{e_2\}
    \cup \{ e_2 \setminus \isolated_{\trip K_1}(e_2) \},\\
    & = \hypergraph(\trip H) \setminus \{ e_1 \} \setminus \{e_2\}
    \cup \{ e_2 \setminus \isolated_{\trip K_1}(e_2) \},\\
    \hypergraph(\trip J_2) & = \hypergraph(\trip K_2) \setminus \{e_1
    \} \cup \{ e_1 \setminus \isolated_{\trip K_2}(e_1) \} \\
    & = \hypergraph(\trip H) \setminus \{ e_2 \} \setminus \{e_1 \}
    \cup \{ e_1 \setminus \isolated_{\trip K_2}(e_1) \}.
  \end{align*}
  Clearly, $\hypergraph(\trip J_1) = \hypergraph(\trip J_2)$ if $e_2
  \setminus \isolated_{\trip K_1}(e_2) = e_1 \setminus
  \isolated_{\trip K_2}(e_1)$. 

  We only show $e_2 \setminus \isolated_{\trip K_1}(e_2) \subseteq e_1
  \setminus \isolated_{\trip K_2}(e_1)$, the other inclusion being
  similar. Let $x \in e_2 \setminus \isolated_{\trip K_1}(e_2)$. Since
  $x \not \in \isolated_{\trip K_1}(e_2)$ one of the following hold.
  \begin{itemize}
  \item $x \in \free(\trip K_1)$. But then, $x \in \free(\trip K_1) =
    \free(\trip I_1) = \free(\trip H) = \free(\trip I_2) = \free(\trip
    K_2)$. In particular, $x$ is an equijoin variable in $\trip{H}$
    and $\trip{K_2}$. Then $x \in \equijoinvars_{\trip H}(e_2)
    \subseteq e_1$ because $e_2 \cse_{\trip H} e_1$. From this and the
    fact that $x$ remains an equijoin variable in $\trip K_2$, we
    obtain $x \in e_1 \setminus \isolated_{\trip K_2}(e_1)$.
  \item $x$ occurs in $e_2$ and in some hyperedge $g$ in $\trip K_1$
    with $g \not = e_2$. Since $e_1$ is not in $\trip K_1$ also $g
    \not = e_1$. Since every hyperedge in $\trip K_1$ is in $\trip
    I_1$ and every hyperedge in $\trip I_1$ is in $\trip H$, also $g$
    is in $\trip H$. But then, $x$ occurs in two distinct hyperedges
    in $\trip H$, namely $e_2$ and $g$, and hence $x \in
    \equijoinvars_{\trip H}(e_2) \subseteq e_1$ because $e_2
    \cse_{\trip H} e_1$. However, because $x$ also occurs in $g$ which
    must also be in $\trip I_2$ and therefore also in $\trip K_2$, $x$
    also occurs in two distinct hyperedges in $\trip K_2$, namely
    $e_1$ and $g$. Therefore, $x \in \equijoinvars_{\trip I_2}(e_1)$
    and hence $x \in e_1 \setminus \isolated_{\trip I_2}(e_1)$, as
    desired.
  \item $x \in \var(\preds(\trip K_1))$. Then there exists $\theta \in
    \preds(\trip K_1)$ such that $x \in \var(\theta)$.  Since
    $\preds(\trip K_1) =\preds(\trip K_2)$, $\theta \in \preds(\trip
    K_2)$. As such, $\theta \in \preds(\trip H)$, $\var(\theta) \cap
    (e_2 \setminus e_1) = \emptyset$, and $\var(\theta) \not \subseteq
    e_1$. But then, since $x \in \var(\theta)$; $x \in e_2$; and 
$\var(\theta) \cap (e_2 \setminus e_1)
    = \emptyset$,  it must be the
    case that $x \in e_1$. As such, $x \in e_1$ and $x
    \in \var(\trip K_2)$. Hence $x \in e_1 \setminus \isolated_{\trip
      K_2}(e_1)$.

  \end{itemize}

  \medskip
  \emph{\uline{(3) Case (ISO, CSE)}}: assume that $\trip{I}_1$ is obtained by
  removing the non-empty set of isolated variables $X_1 \subseteq
  \isolated_{\trip H}(e_1)$ from $e_1$, while $\trip{I}_2$ is obtained
  by removing hyperedge $e_2$, conditional subset of hyperedge
  $f_2$. We may assume w.l.o.g. that $e_1 \not = \isolated_{\trip
    H}(e_1)$: if $e_1 = \isolated_{\trip H}(e_1)$ then the ISO
  operation removes the complete hyperedge $e_1$. However, because no
  predicate in $\trip H$ shares any variable with $e_1$, it is readily
  verified that $e_1 \cse_{\trip H} e_2$ and thus the removal of $e_1$
  can also be seen as an application of CSE on $e_1$\footnote{Note
    that, since $e_1$ does not share variables with any predicate, the
    CSE operation also does not remove any predicates from $\trip
    H_1$, similar to the ISO operation and hence yields $\trip I_1$.},
  and we are hence back in case (2).

  Now reason as follows. Because $e_2 \cse_{\trip H} f_2$ and because
  isolated variables of $e_1$ occur in no other hyperedge in $\trip
  H$, it must be the case that $e_2 \cap X_1 = \emptyset$. In
  particular, $e_1$ and $e_2$ must hence be distinct. Therefore, $e_1
  \in \hypergraph(\trip I_2)$ and $e_2 \in \hypergraph(\trip I_1)$. By
  Lemma~\ref{lem:preservation}, we can apply ISO on $\trip{I}_2$ to
  remove $X_1$ from $e_1$. It then suffices to show that $e_2$ remains
  a conditional subset of some hyperedge $f'_2$ in $\trip{I}_1$ with
  $e_2 \setminus f_2 = e_2 \setminus f'_2$. Indeed, we can then use
  ECQ to remove $e_2$ from $\hypergraph(\trip{I}_1)$ as well as
  predicates $\theta$ with $\var(\theta) \cap (e_2 \setminus f_2) \not
  = \emptyset$ from $\preds(\trip I_1)$. This clearly yields the same
  triplet as the one obtained by removing $X_1$ from $e_1$ in
  $\trip{I}_2$. We need to distinguish two cases.

  (3a) $f_2 \not = e_1$. Then $f_2 \in \hypergraph(\trip I_1)$ and
  hence $e_2 \cse_{\trip I_1} f_2$ by Lemma~\ref{lem:preservation}. We
  hence take $f'_2 = f_2$.

  (3b) $f_2 = e_1$. Then we take $f'_2 = e_1 \setminus X$.  Since $e_1
  \not = \isolated_{\trip H}(e_1)$ it follows that $e_1 \setminus X_1
  \not = \emptyset$. Therefore, $f'_2 = e_1 \setminus X_1 \in
  \hypergraph(\trip I_1)$. Furthermore, since $X \subseteq
  \isolated_{\trip H}(e_1)$, no variable in $X$ is in any other
  hyperedge in $\trip H$. In particular $X \cap e_2 = \emptyset$.
  Therefore, $e_2 \setminus f'_2 = e_2 \setminus (e_1 \setminus X) =
  (e_2 \setminus e_1) \cup (e_2 \cap X) = e_2 \setminus e_1 \setminus
  e_1 = e_2 \setminus f_2$. It remains to show that $e_2 \cse_{\trip
    I_1} e_1 \setminus X_1$.
  \begin{itemize}
  \item $\equijoinvars_{\trip I_1}(e_2) \subseteq e_1 \setminus
    X_1$. Let $x \in \equijoinvars_{\trip I_1}(e_2)$. By
    Lemma~\ref{lem:preservation}, $x \in \equijoinvars_{\trip
      I_1}(e_2) \subseteq \equijoinvars_{\trip H}(e_2) \subseteq e_1$
    where the last inclusion is due to $e_2 \cse_{\trip H} e_1$. In
    particular, $x$ is an equijoin variable in $\trip H$. But then it
    cannot be an isolated variable in any hyperedge. Therefore, $x
    \not \in X_1$.
  \item $\ext_{\trip I_1}(e_2 \setminus e_1) \subseteq e_1 \setminus
    X$. Let $x \in \ext_{\trip I_1}(e_2 \setminus e_1)$.  Then $x \in
    \ext_{\trip I_1}(e_2 \setminus e_1) \subseteq \ext_{\trip H}(e_2
    \setminus e_1) \subseteq e_1$ where the first inclusion is by
    Lemma~\ref{lem:preservation} and the second by $e_2 \cse_{\trip H}
    e_1$. Then, because $x \in \ext_{\trip H}(e_2 \setminus e_1)$ it
    follows from the definition of $\ext$, that $x$ occurs in some
    predicate in $\preds(\trip H)$. However, $X$ is disjoint with
    $\var(\preds(\trip H))$ since it consist only of isolated
    variables. Therefore, $x \not \in X$.
  \end{itemize}

  \medskip
  \emph{\uline{(4): Case (ISO, FLT)}} Assume that $\trip I_1$ is obtained by
  removing the non-empty set $X_1 \subseteq \isolated_{\trip H}(e_1)$
  from hyperedge $e_1$, while $\trip I_2$ is obtained by removing all
  predicates in the non-empty set $\Theta \subseteq \preds(\trip H)$
  with $\var(\Theta) \subseteq e_2$ for some hyperedge $e_2$ in
  $\hypergraph(\trip H)$. Observe that $e_1 \in \hypergraph(\trip
  I_2)$. By Lemma~\ref{lem:preservation}, $X \subseteq
  \isolated_{\trip H}(e_1) \subseteq \isolated_{\trip
    I_2}(e_1)$. Therefore, we may apply reduction operation (ISO) on
  $\trip I_2$ to remove $X_1$ from $e_1$. We will now show that,
  similarly, we may still apply (FLT) on $\trip I_1$ to remove all
  predicates in $\Theta$ from $\preds(\trip I_1) = \preds(\trip
  H)$. The two operations hence commute, and clearly the resulting
  triplets in both cases is the same. We distinguish two
  possibilities. (i) $e_1 \not = e_2$. Then $e_2 \in \trip I_1$ and,
  $\var(\Theta) \subseteq e_2$ and, since (ISO) does not remove
  predicates, $\Theta \subseteq \preds(\trip H) = \preds(\trip
  I_1)$. As such the (FLT) operation indeed applies to remove all
  predicates in $\Theta$ from $\preds(\trip I_1)$. (ii) $e_1 =
  e_2$. Then, since $X \subseteq \isolated_{\trip H}(e_1)$ and
  isolated variables do no occur in any predicate, $X \cap
  \var(\Theta) = \emptyset$. Then, since $\var(\Theta) \subseteq e_2 =
  e_1$, it follows that also $\var(\Theta) \subseteq e_1 \setminus
  X$. In particular, since we disallow nullary predicates and $\Theta$
  is non-empty, $e_1 \setminus X \not = \emptyset$. Thus, $e_1
  \setminus X \in \hypergraph(\trip I_1)$ and hence operation (FLT)
  applies indeed applies to remove all predicates in $\Theta$ from
  $\preds(\trip I_1)$

  \medskip
  \emph{\uline{(5) Case (CSE, FLT):}} assume that $\trip I_1$ is obtained by
  removing hyperedge $e_1$, conditional subset of $e_2$ in $\trip H$,
  while $\trip I_2$ is obtained by removing all predicates in the
  non-empty set $\Theta \subseteq \preds(\trip H)$ with $\var(\Theta)
  \subseteq e_3$ for some hyperedge $e_3 \in \hypergraph(\trip
  H)$. Since the (FLT) operation does not remove any hyperedges, $e_1$
  and $e_2$ are in $\hypergraph(\trip I_2)$. Then, since $e_1
  \cse_{\trip H} e_2$ also $e_1 \cse_{\trip I_2} e_2$ by
  Lemma~\ref{lem:preservation}. Therefore, we may apply reduction
  operation (CSE) on $\trip I_2$ to remove $e_1$ from
  $\hypergraph(\trip I_2)$ as well as all predicates $\theta \in
  \preds(\trip I_2)$ for which $\var(\theta) \cap (e_1 \setminus e_2)
  \not = \emptyset$. Let $\trip J_2$ be the triplet resulting from
  this operation. We will show that, similarly, we may apply (FLT) on
  $\trip I_1$ to remove all predicates in $\Theta \cap \preds(\trip
  I_1)$ from $\preds(\trip I_1)$, resulting in a triplet $\trip
  J_1$. Observe that necessarily, $\trip J_1 = \trip J_2$ (and hence
  they form the triplet $\trip J$).  Indeed, $\free(\trip J_1) =
  \free(\trip I_1) = \free(\trip H) = \free(\trip I_2) = \free(\trip
  J_2)$ since reduction operations never modify output
  variables. Moreover,
  \begin{align*}
    \hypergraph(\trip J_1) & = \hypergraph(\trip I_1) \\
    & = \hypergraph(\trip H) \setminus \{ e_1 \} \\
    & = \hypergraph(\trip I_2) \setminus \{ e_1 \} \\
    & = \hypergraph(\trip J_2)
  \end{align*}
  where the first and third equality is due to fact that (FLT) does
  not modify the hypergraph of the triplet it operates on. Finally,
  observe
  \begin{align*}
    \preds(\trip J_1) & = \preds(\trip I_1) \setminus (\Theta \cap
    \preds(\trip I_1)) \\
    & = \preds(\trip I_1) \setminus \Theta \\
    & = \{ \theta \in \preds(\trip H) \mid \var(\theta) \cap (e_1
    \setminus e_2) = \emptyset \} \setminus \Theta \\
    & = \{ \theta \in \preds(\trip H) \setminus \Theta \mid
    \var(\theta) \cap (e_1 \setminus e_2) = \emptyset \} \\
    & = \{ \theta \in \preds(\trip I_2) \mid
    \var(\theta) \cap (e_1 \setminus e_2) = \emptyset \} \\
    & = \preds(\trip J_2)
  \end{align*}

  It remains to show that we may apply (FLT) on $\trip I_1$ to remove
  all predicates in $\Theta \cap \preds(\trip I_1)$, resulting in a
  triplet $\trip J_1$.  There are two possibilities.
  \begin{itemize}
  \item $e_3 \not = e_1$.  Then $e_3 \in \trip I_1$, $\Theta \cap
    \preds(\trip(I_1)) \subseteq \preds(\trip I_1))$, and $\var(\Theta
    \cap \preds(\trip I_1)) \subseteq \var(\Theta) \subseteq
    e_3$. Hence the (FLT) operation indeed applies to $\trip I_1$ to
    remove all predicates in $\Theta \cap \preds(\trip I_1)$.

  \item $e_3 = e_1$. In this case we claim that for every $\theta \in
    \Theta \cap \preds(\trip I_1)$ we have $\var(\theta) \subseteq
    e_2$. As such, $\var(\Theta \cap \preds(\trip I_1)) \subseteq
    e_2$. Since $e_2 \in \hypergraph(\trip I_1)$ and $\Theta \cap
    \preds(\trip I_1) \subseteq \preds(\trip I_1)$ we may hence apply
    (FLT) to remove all predicates in $\Theta \cap \preds(\trip I_1)$
    from $\trip I_1$. Concretely, let $\theta \in \Theta \cap
    \preds(\trip I_1)$. Because, in order to obtain $\trip I_1$, (CSE)
    removes all predicates from $\trip H$ that share a variable with
    $e_1 \setminus e_2$, we have $\var(\theta) \cap (e_1 \setminus
    e_2) = \emptyset$. Moreover, because $\theta \in \Theta$,
    $\var(\theta) \subseteq e_1$. Hence $\var(\theta) \subseteq e_2$,
    as desired.
  \end{itemize}

  The remaining cases, (CSE, ISO), (FLT, ISO), and (FLT, CSE), are
  symmetric to case (3), (4), and (5), respectively.
\end{proof}


\subsection{Proof of Proposition~\ref{prop:canonical}}
\label{sec:proof-proposition-canonical-pair}

\canonicalPair*
\begin{proof}
	Let $T$ be a GJT.  The proof proceeds in three steps.  \emph{Step
		1.}  Let $T_1$ be the GJT obtained from $T$ by (i) removing all
	predicates from $T$, and (ii) creating a new root node $r$ that is
	labeled by $\emptyset$ and attaching the root of $T$ to it, labeled
	by the empty set of predicates. $T_1$ satisfies the first
	canonicality condition, but is not equivalent to $T$ because it has
	none of $T$'s predicates.  Now re-add the predicates in $T$ to $T_1$
	as follows. For each edge $m \to n$ in $T$ and each predicate
	$\theta \in \preds_T(m \to n)$, if $\var(\theta) \cap (\var(n)
	\setminus \var(m)) \not =\emptyset$ then add $\theta$ to
	$\preds_{T_1}(m \to n)$. Otherwise, if $\var(\theta) \cap (\var(n)
	\setminus \var(m)) = \emptyset$, do the following. First observe
	that, by definition of \gjts, $\var(\theta) \subseteq \var(n) \cup
	\var(m)$. Because $\var(\theta) \cap (\var(n) \setminus \var(m)) =
	\emptyset$ this implies $\var(\theta) \subseteq \var(m)$. Because we
	disallow nullary predicates, $\var(m) \not = \emptyset$. Let $a$ be
	the first ancestor of $m$ in $T_1$ such that $\var(\theta) \not
	\subseteq \var(a)$. Such an ancestor exists because the root of
	$T_1$ is labeled $\emptyset$. Let $b$ be the child of $a$ in
	$T_1$. Since $a$ is the first ancestor of $m$ with $\var(\theta)
	\not \subseteq \var(a)$, $\var(\theta) \subseteq
	\var(b)$. Therefore, $\var(\theta) \subseteq \var(b) \cup \var(a)$
	and $\var(\theta) \cap (\var(b) \setminus \var(a)) \not =
	\emptyset$. As such, add $\theta$ to $\preds_{T_1}(a \to b)$. After
	having done this for all predicates in $T$, $T_1$ becomes equivalent
	to $T$, and satisfies canonicality conditions (1) and (3). Then take
	take $N_1 = N \cup \{r\}$. Clearly, $N_1$ is a 
	connex subset of $T_1$ and $\var(N) = \var(N')$. 
	Therefore, $(T_1,N_1)$ is equivalent to $(T, N)$.
	
	\emph{Step 2.} Let $T_2$ be obtained from $T_1$ by adding, for each
	leaf node $l$ in $T_1$ a new interior node $n_l$ labeled by
	$\var(l)$ and inserting it in-between $l$ and its parent in
	$T_1$. I.e., if $l$ has parent $p$ in $T_1$ then we have $p \to n_l
	\to l$ in $T_2$ with $\preds_{T_2}(p \to n_l) = \pred_{T_1}(p \to n)$
	and $\preds_{T_2}(n_l \to l)= \emptyset$.\footnote{Note that all
		leafs have a parent since the root of $T_1$ is an interior node
		labeled by $\emptyset$.}  Furthermore, let $N_2$ be the connex subset
	of $T_2$ obtained by replacing every leaf node $l$ in $N_1$ by its
	newly inserted node $n_l$. Clearly, $\var(N_2) = \var(N_1) = \var(N)$
	because $var(l) = \var(n_l)$ for every leaf $l$ of $T_1$. By our
	construction, $(T_2, N_2)$ is equivalent to $(T, N)$; $T_2$ satisfies
	canonicality conditions (1), (2), and (4); and $N_2$ is canonical.
	
	\emph{Step 3.} It remains to enforce condition (3). To this end,
	observe that, by the connectedness condition of \gjts, $T_2$
	violates canonicality condition (3) if and only if there exist
	internal nodes $m$ and $n$ where $m$ is the parent of $n$ such that
	$\var(m) = \var(n)$. In this case, we call $n$ a culprit node. We
	will now show how to obtain an equivalent pair $(U, M)$ that removes
	a single culprit node; the final result is then obtained by
	iterating this reasoning until all culprit nodes have been removed.
	
	The culprit removal procedure is essentially the reverse of the
	binarization procedure of Fig.~\ref{fig:binary}.  Concretely, let
	$n$ be a culprit node with parent $m$ and let $n_1,\dots, n_k$ be
	the children of $n$ in $T_2$.  Let $U$ be the GJT obtained from
	$T_2$ by removing $n$ and attaching all children $n_i$ of $n$ as
	children to $m$ with edge label $\preds_U(m \to n_i) =
	\preds_{T_2}(n \to n_i)$, for $1 \leq i \leq k$. Because $\var(n) =
	\var(m)$, the result is still a valid \gjt. Moreover, because
	$\var(n) = \var(m)$ and $T_2$ satisfied condition (4), we had
	$\pred_{T_2}(m \to n) = \emptyset$, so no predicate was lost by the
	removal of $n$. Finally, define $M$ as follows. If $n \in N_2$, then
	set $M = N_2 \setminus \{n\}$, otherwise set $M = N_2$. In the
	former case, since $N_2$ is connex and $n \in N_2$, $m$ must also be
	in $N_2$. It is hence in $M$. Therefore, in both cases, $\var(N) =
	\var(N_2) = \var(M)$. Furthermore, it is straightforward to check
	that $M$ is a connex subset of $U$. Finally, since $N_2$ consisted
	only of interior nodes of $T_2$, $M$ consists only of interior nodes
	of $U$ and hence remains canonical.
\end{proof}

\subsection{Proof of Lemma~\ref{lem:progress}}
\label{sec:proof-lemma-progress}

We first require a number of auxiliary results.

We first make the following observations regarding canonical \gjt pairs.
\begin{lemma}
	\label{lem:nice-removal-facts}
	Let $(T,N)$ be a canonical \gjt pair, let $n$ be a frontier node of $N$
	and let $m$ be the parent of $n$ in $T$.
	\begin{enumerate}
		\item $x \not \in \var(N \setminus \{n\})$, for every $x \in \var(n)
		\setminus \var(m)$.
		\item $\hypergraph(T, N \setminus \{n\}) = \hypergraph(T, N)
		\setminus \{\var(n)\})$. 
		\item $\theta \not \in \preds(m \to n)$, for every $\theta \in
		\preds(T, N \setminus \{n\})$
		\item $\preds(T, N \setminus \{n\}) = \preds(T, N) \setminus
		\preds(m \to n)$.  
		\item $\preds(m \to n) = \{ \theta \in \pred(T, N) \mid \var(\theta) 
		\cap (\var(n) \setminus \var(m)) \not = \emptyset \}$.
		\item $\pred(T, N \setminus \{n\}) = \{ \theta \in \pred(T, N) \mid
		\var(\theta) \cap (\var(n) \setminus \var(m)) = \emptyset \}$.
	\end{enumerate}
\end{lemma}
\begin{proof}
	(1) Let $x \in \var(n) \setminus \var(m)$ and let $c$ be a node in
	$N \setminus \{n\}$. Clearly the unique undirected path
	between $c$ and $n$ in $T$ must pass through $m$. Because $x \not
	\in \var(m)$ it follows from the connectedness condition of \gjts that
	also $x \not \in \var(c)$. As such, $x \not \in
	\var(N \setminus \{n\})$.
	
	(2) The $\supseteq$ direction is trivial. For the $\subseteq$
	direction, assume that $m \in N \setminus \{n\}$ with $\var(m) \not
	= \emptyset$. Then clearly $m \in N$ and hence $\var(m) \in
	\hypergraph(T,N)$. Furthermore, because $N$ is canonical, both $m$ and
	$n$ are interior nodes in $T$. Then, because $T$ is canonical and $m
	\not = n$ we have $\var(m) \not =\var(n)$. Therefore, $\var(m) \in
	\hypergraph(T,N) \setminus \{\var(n)\}$.
	
	(3) Let $\theta \in \pred(T, N \setminus n)$. Then $\theta$ occurs
	on the edge between two nodes in $N \setminus n$, say $m' \to
	n'$. By definition of \gjts, $\var(\theta) \subseteq \var(n') \cup
	\var(m') \subseteq \var(N \setminus \{n\})$. Now suppose for the
	purpose of contradiction that also $\theta \in \pred(m \to
	n)$. Because $T$ is nice, there is some $x \in \var(\theta) \cap
	(\var(n) \setminus \var(m)) \not = \emptyset$. Hence, by (1), $x \not
	\in \var(N \setminus \{n\})$, which contradicts $\var(\theta)
	\subseteq \var(N \setminus \{n\})$. 
	
	(4) Clearly, $\preds(T,N) \setminus \preds(m \to n) \subseteq
	\preds(T, N \setminus \{n\})$. The converse inclusion follows from (3).
	
	(5) The $\subseteq$ direction follows from the fact that $m$ and $n$
	are in $N$, and $T$ is nice. To also see $\supseteq$, let $\theta
	\in \pred(T, N)$ with $\var(\theta) \cap (\var(n) \setminus \var(m))
	\not = \emptyset$. There exists $x \in \var(\theta) \cap (\var(n)
	\setminus \var(m))$. By (1), $x \not \in \var(N \setminus
	\{n\})$. Therefore, $\theta$ cannot occur between edges in $N
	\setminus \{n\}$ in $T$. Since it nevertheless occurs in
	$\pred(T,N)$, it must hence occur in $\preds(m \to n)$.
	
	(6) Follows directly from (4) and (5).
\end{proof}

\begin{lemma}
	\label{lem:nice-rewr-applicable}
	Let $(T,N)$ be a canonical \gjt pair, let $n$ be a frontier node of $N$
	and let $m$ be the parent of $n$ in $T$. Let $\seq{z} \subseteq
	\var(N \setminus \{n\})$. 
	\begin{enumerate}
		\item $\var(n) \cse_{\hypertrip(T, N, \seq{z})} \var(m)$.
		\item $x \not \in \equijoinvars(\hypertrip(T, N, \seq{z}))$, for
		every $x \in (\var(n) \setminus \var(m))$.
	\end{enumerate}
\end{lemma}
\begin{proof}
	For reasons of parsimony, let $\trip H = \hypertrip(T, N,
	\seq{z})$. We first prove (2) and then (1).
	
	(2) Let $x \in \var(n) \setminus \var(m)$. By
	Lemma~\ref{lem:nice-removal-facts}(1), $x \not \in \var(N \setminus
	\{n\})$. Therefore, $x$ occurs in $\var(n)$ in $\trip H$ and in no
	other hyperedge. Furthermore, because $\seq{z} \subseteq \var(N
	\setminus \{n\})$, also $x \not \in \seq{z}$. Hence $x \not \in
	\equijoinvars_{\trip H}(\var(n))$. 
	
	(1) We need to show that $\equijoinvars_{\trip H}(\var(n)) \subseteq
	\var(m)$ and $\ext_{\trip H}(\var(n) \setminus \var(m)) \subseteq
	\var(m)$. Let $x \in \equijoinvars_{\trip H}(\var(n))$. By
	contraposition of (2), we know that $x \not \in (\var(n) \setminus
	\var(m))$. Therefore, $x \in \var(m)$ and thus $\equijoinvars_{\trip
		H}(\var(n)) \subseteq \var(m)$. To show $\ext_{\trip H}(\var(n)
	\setminus \var(m)) \subseteq \var(m)$, let $y \in \ext_{\trip
		H}(\var(n) \setminus \var(m))$. Then $y \not \in \var(n) \setminus
	\var(m)$ and there exists $\theta \in \pred(T, N)$ with
	$\var(\theta) \cap (\var(n) \setminus \var(m)) \not = \emptyset$ and
	$y \in \var(\theta)$. By Lemma~\ref{lem:nice-removal-facts}(5),
	$\theta \in \pred_T(m \to n)$. Thus, $y \in \var(m) \cup
	\var(n)$. Since also $y \not \in \var(n) \setminus \var(m)$, it
	follows that $y \in \var(m)$. Therefore, $\ext_{\trip H}(\var(n)
	\setminus \var(m)) \subseteq \var(m)$.
\end{proof}

\begin{lemma}
	\label{lem:progress-1step}
	Let $(T, N)$ be a canonical \gjt pair and let $n$ be a frontier node
	of $N$. Then $\hypertrip(T,N, \seq{z}) \rewr^* \hypertrip(T,N
	\setminus\{n\}, \seq{z})$ for every $\seq{z} \subseteq \var(N
	\setminus \{n\})$.
\end{lemma}
\begin{proof}
	For reasons of parsimony, let us abbreviate $\trip H_1 =
	\hypertrip(T, N, \seq{z})$ and $\trip H_2 = \hypertrip(T, N
	\setminus\{n\}, \seq{z})$.  We make the following case analysis.
	
	Case (1): Node $n$ is the root in $N$. Because the root of a
	canonical tree is labeled by $\emptyset$ we have $\var(n) =
	\emptyset$. Since $n$ is a frontier node of $N$, $N = \{n\}$. Thus,
	$\hypergraph(T, N) = \emptyset$ and
	$\hypergraph(T, N \setminus \{n\}) =
	\emptyset$.  Furthermore, $\preds(T, N) = \preds(T, N \setminus
	\{n\}) = \emptyset$ and $\seq{z} \subseteq \var(N \setminus \{n\}) =
	\var(\emptyset) = \emptyset$. As such, both $\trip H_1$ and $\trip
	H_2$ are the empty triplet $(\emptyset, \emptyset,
	\emptyset)$. Therefore $\trip H_1 \rewr^* H_2$.
	
	Case (2): $n$ has parent $m$ in $N$ and $\var(m) \not =
	\emptyset$. Then $\var(n) \not = \emptyset$ since in a canonical
	tree the root node is the only interior node that is labeled by the
	empty hyperedge. Therefore, $\var(n) \in \hypergraph(T, N)$, $\var(m) \in
	\hypergraph(T, N)$, and $\var(n) \cse_{\trip H_1} \var(m)$ by
	Lemma~\ref{lem:nice-rewr-applicable}(1). We can hence apply
	reduction (CSE) to remove $\var(n)$ from $\hypergraph(\trip H_1)$
	and all predicates that intersect with $\var(n) \setminus \var(m)$
	from $\preds(\trip H_1)$. By Lemma~\ref{lem:nice-removal-facts}(2)
	and \ref{lem:nice-removal-facts}(6) the result is exactly $\trip H_2$:
	\begin{align*}
	& \hypergraph(\trip H_2) \\
	& = \hypergraph(T, N \setminus \{n\}) \\
	& =  \hypergraph(T, N) \setminus \{ \var(n) \} 
	= \hypergraph(\trip H_1) \setminus \{ \var(n) \} \\[1.5ex]
	&  \preds(\trip H_2) \\
	& =  \preds(T, N \setminus \{ n\}) \\
	& = \{ \theta \in \preds(T,N) \mid \var(\theta) \cap (\var(n)
	\setminus \var(m)) = \emptyset \}\\
	& = \{ \theta \in \preds(\trip H_1) \mid \var(\theta) \cap (\var(n)
	\setminus \var(m)) = \emptyset \}
	\end{align*}
	
	Case (3): $n$ has parent $m$ in $N$ and $\var(m) = \emptyset$. Then
	$\var(n) \not = \emptyset$ since since in a canonical tree the root
	node is the only interior node that is labeled by the empty
	hyperedge.  By definition of \gjts, it follows that for every
	$\theta \in \pred(m \to n)$ we have $\var(\theta) \subseteq \var(n)
	\cup \var(m) = \var(n)$. In other words: all $\theta \in \pred(m \to
	n)$ are filters. As such, we can use reduction (FLT) to remove all
	predicates in $\pred(m \to n)$ from $\trip H_1$. This yields a
	triplet $\trip I$ with the same hypergraph as $\trip H_1$, same set
	of output variables as $\trip H_1$, and
	\begin{align*}
	\preds(\trip I) & = \preds(\trip H_1) \setminus \pred_T(m \to n)
	\\
	& = \preds(T, N) \setminus \pred_T(m \to n) \\
	& = \preds(T, N \setminus \{n\}) = \preds(\trip H_2),
	\end{align*}
	where the third equality is due to
	Lemma~\ref{lem:nice-removal-facts}(4). We claim that every variable
	in $e$ is isolated in $\trip I$. From this the result follows,
	because then we can apply (ISO) to remove the entire hyperedge
	$\var(e)$ from $\hypergraph(\trip I) = \hypergraph(\trip H_1)$ while
	preserving $\free(\trip I)$ and $\preds(\trip I)$. The resulting
	triplet hence equals $\trip H_2$.  To see that $e \subseteq
	\isolated(\trip I)$, observe that no predicate in $\preds(\trip I) =
	\preds(T, N \setminus \{n\})$ shares a variable with $\var(n) =
	(\var(n) \setminus \var(m))$ by
	Lemma~\ref{lem:nice-removal-facts}(6). Therefore $\var(n) \cap
	\var(\preds(\trip I)) = \emptyset$. Furthermore, $\var(n)
	\cap \equijoinvars(\trip I) = \emptyset$ because
	$\equijoinvars(\trip I) = \equijoinvars(\trip H_1)$ and no $x \in
	\var(n) = \var(n) \setminus \var(m)$ is in $\equijoinvars(\trip
	H_1)$ by Lemma~\ref{lem:nice-rewr-applicable}(2).
\end{proof}

\progress*
\begin{proof}
	By induction on $k$, the number of nodes in $N_1 \setminus N_2$. In
	the base case where $k = 0$, the result trivially holds since then
	$N_1 = N_2$ and the two triplets are identical. For the induction
	step, assume that $k > 0$ and the result holds for $k-1$. Because
	both $N_1$ and $N_2$ are connex subsets of the same tree $T$, there
	exists a node $n \in N_1$ that is a frontier node in $N_1$, and
	which is not in $N_2$. Then define $N'_1 = N_1 \setminus
	\{n\}$. Clearly $(T, N'_1)$ is again canonical, and $|N'_1\setminus N_2|
	= k-1$. Therefore, $\hypertrip(T, N'_1, \seq{z}) \rewr^*
	\hypertrip(T, N_2, \seq{z})$ by induction hypothesis. Furthermore,
	by $\hypertrip(T, N_1, \seq{z}) \rewr^* \hypertrip(T, N'_1,
	\seq{z})$ by Lemma~\ref{lem:progress-1step}, from which the result
	follows.
\end{proof}

\subsection{Proof of Lemma~\ref{lem:subset-removal}}
\label{sec:proof-lemma-subsetremoval}

\subsetRemoval*

\begin{proof}
	The proof is by induction on $k$, the number of hyperedges in $H_2
	\setminus H_1$. In the base case where $k = 0$, the result trivially
	holds since $H_1 \cup H_2 = H_1$ and the two triplets are hence
	identical. For the induction step, assume that $k > 0$ and the
	result holds for $k -1$.  Fix some $e \in H_2 \setminus H_1$ and
	define $H'_2 = H_2 \setminus \{e\}$. Then $|H'_2 \setminus H_1| = k
	-1$. We show that $(H_1 \cup H_2, \seq{z}, \Theta) \rewr^* (H_1 \cup
	H'_2, \seq{z}, \Theta)$, from which the result follows since $(H_1
	\cup H'_2, \seq{z}, \Theta) \rewr^* (H_1, \seq{z}, \Theta)$ by
	induction hypothesis. To this end, we observe that there exists $\ell
	\in H_1 \setminus \{e\}$ with $e \subseteq \ell$. Therefore,
	$\equijoinvars_{(H_1 \cup H_2, \seq{z}, \Theta)}(e) \subseteq e
	\subseteq \ell$. Moreover, $e \setminus \ell = \emptyset$. Therefore,
	$\ext_{(H_1\cup H_2, \seq{z}, \Theta)}(e \setminus \ell) = \emptyset
	\subseteq \ell$. Thus $e \cse_{(H_1 \cup H_2, \seq{z}, \Theta)} \ell$. We
	may therefore apply (CSE) to remove $e$ from $H_1 \cup H_2$,
	yielding $H_1 \cup H'_2$. Since no predicate shares variables with
	$e \setminus \ell = \emptyset$ this does not modify
	$\Theta$. Therefore, $(H_1 \cup H_2, \seq{z}, \Theta) \rewr^* (H_1
	\cup H'_2, \seq{z}, \Theta)$.
\end{proof}


\end{appendix}

\end{document}